\documentclass[11pt]{article}
\usepackage{amsmath,amsthm,amsfonts,amscd,eucal,latexsym,amssymb,mathrsfs}
\usepackage{epsfig}  
\usepackage{simplewick}
\usepackage{tikz-cd}
\usepackage{hyperref}
\usepackage{color}

\input xy
\xyoption{all}
 
\def\be{\beta}

\def\si{\sigma}
\def\om{\omega}

\def\Om{\Omega}

\def\WF{{\text{WF}}}

\def\bC{{\mathbb C}}           

\def\bN{{\mathbb N}}

\def\bR{{\mathbb R}}

\def\b1{{\boldsymbol 1}}
 
\def\mA{\mathscr{A}}

\def\mE{\mathscr{E}} 
 
\def\mH{\mathscr{H}}

\def\mF{\mathscr{F}}  
\def\mI{\mathscr{I}}  
\def\mM{\mathscr{M}}

\def\mR{\mathscr{R}} 
\def\mP{\mathscr{P}}

\def\beq{\begin{eqnarray}}
\def\eeq{\end{eqnarray}}
\def\pa{\partial}
\def\at{\left(}               

\def\be{\beta}

\def\si{\sigma}
\def\om{\omega}

\def\Om{\Omega}

\def\supp{\mathrm{supp}\;}

\newcommand{\bra}[1]{\langle{#1}|}
\newcommand{\ket}[1]{|{#1}\rangle}

\def\bx{{\boldsymbol{x}}}
\def\by{{\boldsymbol{y}}}
\def\bz{{\boldsymbol{z}}}
\def\bp{{\boldsymbol{p}}}
\def\bk{{\boldsymbol{k}}}
\def\bq{{\boldsymbol{q}}}
\def\bn{\boldsymbol{n}}
\def\bl{\boldsymbol{l}}
\def\bpp{\boldsymbol{p}_\parallel}
\def\bkp{\boldsymbol{k}_\parallel}
\def\bqp{\boldsymbol{q}_\parallel}
\def\blp{\boldsymbol{l}_\parallel}
\def\bmp{\boldsymbol{m}_\parallel}
\def\bnp{\boldsymbol{n}_\parallel}

\def\bxp{\boldsymbol{x}_\parallel}
\def\bxpI{\boldsymbol{x}_{\parallel 1}}
\def\bxpII{\boldsymbol{x}_{\parallel 2}}
\def\byp{\boldsymbol{y}_\parallel}
\def\bvp{\boldsymbol{v}_\parallel}
\def\bzp{\boldsymbol{z}_\parallel}
\def\bm{\boldsymbol{m}}
\def\omp{\omega_{\boldsymbol{p}}}
\def\omk{\omega_{\boldsymbol{k}}}
\def\omq{\omega_{\boldsymbol{q}}}
\def\omn{\omega_{\boldsymbol{n}}}

\def\omm{\omega_{\boldsymbol{m}}}

\def\ompq{\omega_{\boldsymbol{p}-\boldsymbol{q}}}

\def\ABU{\mathscr{A}_{\mathrm{BU}}}

\def\AV{\mathscr{A}_{ V}}

\def\AO{\mathscr{A}_{0}}
\def\AOon{\mathscr{A}^{\text{on}}_{0}}
\def\moller{\mR_{ V}}

\def\omb{\omega_\beta}

\def\ombi{\omega_{\beta_i}}
\def\ombV{\omega^{ V}_\beta}
\def\omG{\omega_G}
\def\omGV{\omega^{ V}_G}
\def\omN{\omega_N}
\def\omNV{\omega^{V}_N}
\def\ombV{\omega^{V}_{\beta}}

\def\Fmuc{\mF_{\mu \text{c}}} 
\def\Floc{\mF_{\mathrm{loc}}} 
\def\FTloc{\mF_{T\mathrm{loc}}} 
\def\Freg{\mF_{\mathrm{reg}}}

\def\ren{\text{ren}}

\def\at{\alpha_t}
\def\atV{\alpha^{ V}_t}

\newcommand{\wick}[1]{:\!{#1}\!:}

\def\vlim{\mathrm{vH}\text{-}\!\lim}

\newtheorem{theorem}{Theorem}[section]

\newtheorem{proposition}[theorem]{Proposition}

\theoremstyle{definition}

\definecolor{NiColor}{RGB}{32, 130, 150}

\begin{document} 
%
%
 
\par 
\bigskip 
\LARGE 
\noindent 
{\bf Non-equilibrium steady states for the interacting Klein-Gordon field in 1+3 dimensions} 
\bigskip 
\par 
\rm 
\normalsize 
 

\large
\noindent 
{\bf Thomas-Paul Hack$^{1,a}$}, {\bf Rainer Verch$^{1,b}$} \\
\par
\small
\noindent$^1$ Institute f\"ur Theoretische Physik, Universit\"at Leipzig,
\smallskip

\noindent E-mail: 
$^a$thomas-paul.hack@itp.uni-leipzig.de, 
$^b$rainer.verch@uni-leipzig.de\\ 

\normalsize

\par 
 
\rm\normalsize 

\rm\normalsize 
 
 
\par 
\bigskip 
 
\rm\normalsize 
\noindent {\small Version of \today}

\par 
\bigskip

\rm\normalsize

\small 

\noindent Non-equilibrium steady states (NESS) describe particularly simple and stationary non-equilibrium situations. A possibility to obtain such states is to consider the asymptotic evolution of two infinite heat baths brought into thermal contact. In this work we generalise corresponding results of Doyon~et.~al. (J.\ Phys.\ A 18 (2015) no.9) for free Klein-Gordon fields in several directions. Our analysis is carried out directly at the level of correlation functions and in the algebraic approach to QFT. We discuss non-trivial chemical potentials, condensates, inhomogeneous linear models and homogeneous interacting ones. We shall not consider a sharp contact at initial time, but a smooth transition region. As a consequence, the states we construct will be of Hadamard type, and thus sufficiently regular for the perturbative treatment of interacting models. Our analysis shows that perturbatively constructed interacting NESS display thermodynamic properties similar to the ones of the NESS in linear models. In particular, perturbation theory appears to be insufficient to describe full thermalisation in non-linear QFT. Notwithstanding, we find that the NESS for linear and interacting models is stable under small perturbations, which is one of the characteristic features of equilibrium states.

\normalsize

\vskip .3cm

\tableofcontents

\section{Introduction}

Non-equilibrium steady states (NESS) describe particularly simple and stationary situations in non-equilibrium thermodynamics. Many examples of such states have been studied and an axiomatic analysis in the context of quantum statistical mechanics has been carried out in \cite{Ru00}. A possibility to obtain a NESS is to consider long-time limits of initial configurations of systems whose dynamics is such that the initial state can not fully equilibrise. An example of such an initial state is the case of two semi-infinite heat baths brought into thermal contact at a $d-2$-dimensional surface in $d$ spacetime dimensions. The dynamics of such states has been investigated in several models including conformal hydrodynamics \cite{Doyon2} and conformal QFT in $d=2$ spacetime dimensions \cite{Doyon3,HollandsLongo}. 

In this work, we shall consider this situation for the Klein-Gordon field in $d=4$ spacetimes dimensions, with parts of the analysis applying to any $d\ge2$. This setup has already been studied in \cite{Doyon} for linear models, i.e. free fields, where the authors have studied a sharp contact at $x^1 = 0$. In more detail, the initial state analysed in \cite{Doyon} is in equilibrium at inverse temperatures $\beta_1$ and $\beta_2$ for $x^1<0$ and $x^1>0$, respectively, with suitable boundary conditions  at $x^1=0$. The authors of \cite{Doyon} show that this state converges to a limit NESS which is independent of the boundary conditions chosen at $x^1=0$. They also show that the limit NESS is a proper equilibrium state in a different rest frame if and only if $d=2$ and $m=0$, where $m$ is the mass of the field. The failure of the limit NESS to be in proper equilibrium is attributed to the fact that linear QFT models possess infinitely many conserved charges, which, in view of the concept of generalised Gibbs ensembles \cite{Rigol}, obstruct proper thermalisation. Consequently, it has been conjectured in \cite{Doyon} that non-linear models should display a better thermalisation behaviour.

Taking this as a motivation, we shall study NESS arising as asymptotic limits of initial states of the aforementioned kind in interacting models in $d=4$ spacetime dimensions. However, we shall also generalise the work of \cite{Doyon} on linear models in various aspects. Our analysis  will be carried out directly at the level of correlation functions and in the algebraic approach to QFT. Thus we will be able to avoid Fock-space constructions and the introduction of a finite box cut-off in order to make some expressions well-defined. We shall further discuss non-trivial chemical potentials, as well as related condensates, and also inhomogeneous linear models in order to analyse the effect of breaking momentum conservation. Most importantly, we shall not consider a sharp contact at initial time, but a smooth transition region. While this does, as we will show, not have an effect on the limit NESS, the initial states we prepare will be more regular than the ones considered in \cite{Doyon}. In fact, the states we construct for homogeneous linear models will be of the so-called Hadamard type, and thus they will be sufficiently regular  for all our perturbative constructions for interacting models to be well-defined.

After preparing initial states for interacting models, we show that these states converge to a limit NESS. Our results in that respect are explicit only at first order in perturbation theory, though we suspect and provide indications that convergence holds at all orders. The authors of \cite{Doyon} already argued that the limit NESS they found is separately in equilibrium for left- and right-moving modes, though at different temperatures. We will argue that the interacting NESS can also be characterised in this way, which may be traced back to the fact that a non-linear interaction can not mix left- and right-moving modes in perturbation theory unless it breaks translation invariance in space. In that sense, the interacting NESS we construct does not appear  to have thermodynamic properties which deviate considerably from those of the corresponding free NESS. We verify this assessment by computing the entropy production in an interacting NESS with respect to the corresponding free NESS, and find that it is vanishing. However, we show that both the interacting and the free NESS are stable with respect to local perturbations, which demonstrates that they possess at least one of the characteristic properties of equilibrium states.

This paper is organised as follows. In Section \ref{sec_preliminaries}, we briefly review the algebraic description of free and interacting QFT models, as well as the definition of equilibrium and Hadamard states. We also develop a rigorous way to define regular initial states which have prescribed properties in different regions of space at initial time. In the sections \ref{sec_NESSlinear} and \ref{sec_NESSinteracting} we discuss the preparation and convergence of initial states to NESS for linear and interacting models, respectively. We conclude by analysing various properties of these NESS in Section \ref{sec_properties}. Many results require extensive computations and arguments, which are collected in several appendices.

\section{Preliminaries}
\label{sec_preliminaries}
\subsection{KMS states, condensates and thermal domination}
\label{sec_KMS}
In this paper we shall discuss aspects of non-equilibrium thermodynamics in the language of the algebraic approach to QFT. In this algebraic framework, the observables of a QFT model form a topological $*$-algebra $\mA$ containing a unit which we denote by $\b1$. The involution $*$ defined on $\mA$ satisfies $(A^*)^* = A$ for all $A\in\mA$ and $(c A B)^* = \overline{c} B^* A^*$ for all $A,B\in \mA$ and $c\in\bC$ and is an algebraic abstraction of the Hermitean conjugation of operators. States are complex-valued linear functionals on $\mA$ which are continuous with respect to the topology on $\mA$ and satisfy $\om(\b1)=1$ (normalisation) and $\om(A^*A)\ge0$ for all $A\in\mA$ (positivity). For any such state, a represention $\pi_\om$ of $\mA$ on a Hilbert space $\mH_\om$, which is unique up to unitary equivalence, can be obtained by the GNS construction. In this representation, $\om$ corresponds to a vector $\ket{\Om_\om}\in\mH_\om$ and $\om(A) = \bra{\Om_\om}\pi_\om(A)\ket{\Om_\om}$ for all $A\in\mA$. Further details may be found for example in \cite{Haag,Khavkine}.

We consider a linear Klein-Gordon model on $d$-dimensional Minkowski spacetime with equation of motion
\beq\label{eq_kgo}
K\Phi = 0\,,\qquad K \doteq \partial^2_{0} + L\,,\qquad L \doteq - D + V
\eeq
where $\partial_\mu \doteq \frac{\partial}{\partial x^\mu}$, $D = \sum^{d-1}_{i=1}\partial^2_{i}$ is the $d-1$-dimensional Laplacian and $V$ is a time-independent potential such that $L$ is semi-positive and essentially self-adjoint on $C^\infty_0(\bR^{d-1})$, the compactly supported complex-valued smooth functions on $\bR^{d-1}$. In order to allow for non-trivial chemical potentials, the scalar field $\Phi$ is complex-valued. In the discussion of non-linear models, we will consider only real scalar fields indicated by $\phi$. The basic algebra of observables for $\Phi$ is the Borchers-Uhlmann algebra $\ABU$, which is the the algebra generated by a unit $\b1$ and symbols $\Phi(f)$, $\overline{\Phi}(f)$, $f\in C^\infty_0(\bR^{d})$, factored by suitable algebraic relations, such that, for any $f_1,\dots,f_n\in C^\infty_0(\bR^{d})$, $(\Phi(f_1)\dots \Phi(f_n))^* = \overline{\Phi}(\overline{f_n})\dots \overline{\Phi}(\overline{f_1})$, $\Phi(Kf_1)=0$ and $[\overline{\Phi}(f_1),\Phi(f_2)] = 2 i \Delta(f_1,f_2)\b1$. Here, $\Delta = \Delta_R - \Delta_A$ is the retared-minus-advanced Green's function of the Klein-Gordon operator $K$, frequently called commutator function, causal propagator, or Pauli-Jordan distribution. The topology on $\ABU$ is the one induced by the topology of compactly supported smooth functions, see e.g. \cite{Khavkine} for further details. The symbols $\Phi(f)$, $\overline{\Phi}(f)$ may be interpreted as smeared quantized fields and their complex conjugates, respectively, e.g. $\Phi(f) = \int_{\bR^{d}} \Phi(x) f(x) dx$. Consequently, $\Phi(x)$ and $\overline{\Phi}(x)$ become operator valued distributions in a GNS representation of $\ABU$ and a state $\om$ on $\ABU$ is entirely determined by specifying $n$-point correlation functions $\om(\Phi^\sharp(x_1)\dots \Phi^\sharp(x_n))$, for all $n$, where $\Phi^\sharp$ stands for either $\Phi$ or its complex conjugate.

Consider an arbitrary $*$-algebra $\mA$ as well as a one-parameter automorphism group $\alpha_t$ on $\mA$, i.e. $\alpha_t \circ \alpha_s = \alpha_{t+s}$, $\alpha_t(AB) = \alpha_t(A)\alpha_t(B)$, $\at(A^*) = \at(A)^*$ for all $t,s\in\bR$ and $A,B\in\mA$. A state $\omega$ on $\mA$ is said to satisfy the KMS condition with respect to $\alpha_t$ for inverse temperature $\beta$ if, for any pair $A,B\in\mA$ there exists a function $F_{AB}(z)$ which is analytic in the strip $S_\beta \doteq \{z\in \bC: 0 < \Im(z) < \beta\}$ and continuous and bounded on the boundaries of $S_\beta$ such that, for all $t \in \bR$
$$
F_{AB}(t) = \om(A \alpha_t(B))\,,\qquad F_{AB}(t+i\beta) = \om(\alpha_t(B)A)\,.
$$
States satisfying the KMS condition are called KMS states and are the proper generalisation of Gibbs ensembles in the thermodynamic limit \cite{Haag:1967sg}. The KMS condition implies in particular the invariance of any KMS state $\omb$ under $\alpha_t$, $\omb\circ\alpha_t = \omb$, the time-evolution with respect to which $\omb$ is in equilibrium.

In order to discuss KMS states for the linear Klein-Gordon field $\Phi$, we introduce a time-evolution $\alpha^\mu_t$ on $\ABU$ by setting 
\beq\label{eq_timeevol}
\alpha^\mu_t(\Phi(x^0,\bx)) = e^{-i \mu t} \Phi(x^0+t,\bx)\,,\qquad \alpha^\mu_t(\overline{\Phi}(x^0,\bx)) = e^{i \mu t} \overline{\Phi}(x^0+t,\bx)
\eeq
 where $\mu\in\bR$ is the chemical potential. Denoting the unique self-adjoint extension of the semi-positive operator $L$ defined in \eqref{eq_kgo} by $L$ as well, and taking advantage of functional calculus, the KMS states $\om_{\beta,\mu}$ on $\ABU$ for given $\beta\ge 0$ and $\mu$ are specified by having the two-point functions
$$
\om_{\beta,\mu}(\overline{\Phi}(x)\Phi(y)) = \overline{\om_{\beta,\mu}(\overline{\Phi}(y)\Phi(x))} \doteq  2 \Delta_{+,\beta,\mu}(x,y)
$$
$$
\om_{\beta,\mu}(\Phi(x)\Phi(y))  = \om_{\beta,\mu}(\overline{\Phi}(x)\overline{\Phi}(y))  = 0
$$
\beq\label{eq_KMS2pf}
\Delta_{+,\beta,\mu}(x,y) = \frac{1}{2\sqrt{L}}\left(\frac{\exp (-i \sqrt{L}(x^0-y^0))}{1-\exp (-\beta (\sqrt{L}-\mu))}(\bx,\by)+\frac{\exp (i \sqrt{L}(x^0-y^0))}{\exp (\beta (\sqrt{L}-\mu))-1}(\by,\bx)\right)\,.
\eeq
In order for these states to be positive, $\mu$ is restricted by demanding that $\sqrt{L}-\mu$ is semi-positive. The states $\om_{\beta,\mu}$ are Gaussian (quasi-free) states, meaning that their $n$-point correlation functions are vanishing for odd $n$ and sums of products of two-point correlation functions for even $n$. The states $\om_{\beta,\mu}$ are unique for given $\beta$,$\mu$ and $\alpha_t$ as long as $L-\mu^2>0$. Otherwise (Bose-Einstein-)condensates may exist, as we shall describe below.

If a complete set of weak eigenfunctions (modes) of $L$ is known, the integral kernels of the functions of $L$ appearing in $\om_{\beta,\mu}$ can be written as appropriate mode expansions. The most general linear models we shall consider in this work are such that the potential $V$ in $L$ respectively $K$ \eqref{eq_kgo} is of the form 
\beq\label{eq_kgo2}
V(\bx) = m^2 + U(x_1)
\eeq
 with $m\ge 0$ and $U(x^1)$ playing the role of a mass term which is inhomogeneous in the $x^1$ direction. Considering the normalised and complete modes $Y_{k_1}(x^1)$ of the operator $-\partial^2_1 + U(x^1)$, with $U(x^1)$ chosen such that $k_1\in\bR$ and $(-\partial^2_1 + U(x^1)) Y_{k_1}(x^1) = k^2_1 Y_{k_1}(x^1)$,  i.e. 
$$
\int_\bR dk_1 \overline{Y_{k_1}(x^1)} Y_{k_1}(y^1) = \delta(x^1 - y^1)\,,\qquad \int_\bR dx^1 \overline{Y_{k_1}(x^1)} Y_{p_1}(x^1) = \delta(k_1 - p_1)\,,
$$
the two-point function of $\om_{\beta,\mu}$ can be written as 
\beq\label{eq_KMS2pfexplicit}
\Delta_{+,\beta,\mu}(x,y) = \frac{1}{(2\pi)^{d-2}} \int d\bp \frac{1}{2\omp} \sum_{s=\pm 1} e^{i s \omp (x^0 - y^0)} \overline{Y_{p_1,s}(x^1)}Y_{p_1,s}(y^1) e^{i \bpp( \bxp-\byp)} s\, b_{\beta,\mu}(s \omp)\,,
\eeq
$$
\bxp = (x^2,\dots,x^{d-1})\,,\qquad \bpp = (p_2,\dots,p_{d-1})\,,\qquad \omp = \sqrt{|\bp|^2+m^2}\,,
$$
$$
b_{\beta,\mu}(\om) \doteq \frac{1}{\exp(\beta (\om-\mu)) - 1}\,,\qquad Y_{k_1,+} \doteq Y_{k_1}\,,\qquad Y_{k_1,-} \doteq \overline{Y_{k_1}}
$$
If $\overline{Y_{k_1}} = Y_{-k_1}$, like in the case $U(x^1)=0$, then the expression can be further simplified as usual.

Coherent states (condensates) on the algebra $\ABU$ can be obtained as follows. Let $\Psi$ be a smooth solution of $K\Psi = 0$. $\Psi$ induces a $*$-automorphism $\iota_\Psi$ of $\ABU$ by setting
\beq \label{eq_coherent}
\iota_\Psi(\b1) = \b1\,,\qquad\iota_\Psi(\Phi(x)) = \Phi(x)+\Psi(x) \b1\,,\qquad\iota_\Psi(\overline{\Phi}(x)) = \overline{\Phi}(x)+\overline{\Psi}(x) \b1\,.
\eeq
 For any state $\om$ on $\ABU$ we can define a coherent state $\om_\Psi$ by
 \beq\label{eq_coherent2}
\om_\Psi\doteq \om \circ \iota_\Psi\,.
 \eeq
  If $\om$ is Gaussian, $\om_\Psi$ is not of this kind any longer, but all $n$-point functions can be computed in terms of $\Psi$ and the two-point function of $\om$, e.g. $\om_\Psi(\Phi(x)) = \Psi(x)$, $\om_\Psi(\overline{\Phi(x)}\Phi(y)) = \om(\overline{\Phi(x)}\Phi(y)) + \overline{\Psi(x)}\Psi(y)$. Once again, we consider linear models of the form \eqref{eq_kgo},\eqref{eq_kgo2}. We set $\mu = \pm m$ and $\Psi_c(x) = c e^{i \mu x^0}Y_0(x^1)$, where $Y_0(x^1)$ is the zero-mode of $-\partial^2_1 + U(x^1)$ and $c\in\bC$ is a constant. We now define 
\beq\label{eq_bec}
\om_{\beta,\mu,c} \doteq \om_{\beta,\mu}\circ \iota_{\Psi_c}\,.
\eeq
For any $\beta\ge 0$, $\mu=\pm m$ and any $c\in \bC$ the state just defined is a KMS state for $\alpha_t$ as in \eqref{eq_timeevol}. This state constitutes a Bose-Einstein-condensate for the given model.

We conclude this section by stating a technical result which we shall use in order to construct states which provide suitable universal initial configurations that converge to a NESS.

\begin{proposition}\label{prop_domination} Let $0 \le \beta_1 \le \beta_2$ and $\mu_1 \ge \mu_2$. The two-point functions $\Delta_{\be_i,\mu_i}$ \eqref{eq_KMS2pf} of the KMS states $\om_{\beta_i,\mu_i}$ dominate each other in the sense that
$$
\Delta_{+,\be_1,\mu_1}(\overline{f},f) \ge \Delta_{+,\be_2,\mu_2}(\overline{f},f)
\qquad \forall \; f\in C^\infty_0(\bR^d)\,.
$$
\end{proposition}
\begin{proof}
The statement follows from the form of the two-point functions given in \eqref{eq_KMS2pf} and the fact that $(1-\exp(-x))^{-1}$ and $(\exp(x)-1)^{-1}$ are monotonically decreasing functions.
\end{proof}

\noindent Note that Proposition \eqref{prop_domination} does not imply that $\om_{\be_1,\mu_1}(A^*A) \ge \om_{\be_2,\mu_2}(A^*A)$ for all $A\in\ABU$. In fact, observables which are loosely speaking spectrally biased towards low energies violate this inequality. As an analogy in a simpler situation, one may think of projectors on ground states in finite-volume Gibbs ensembles, whose expectation values are monotonically increasing, rather than decreasing, as a function of $\beta$.

\subsection{Gluing of Hadamard states}

\label{sec_gluing}

We intend to construct states which constitute initial configurations that evolve into NESS at asymptotic times. Physically, the states should correspond to two semi-infinite reservoirs brought into contact at a $d-2$-dimensional spatial hypersurface which we choose to be the $x^1=0$ hypersurface in $\bR^{d-1}$. For technical reasons, we would like the initial states to be of Hadamard type, which implies a certain regularity of the state in the transition region between the two reservoirs. While any singularities are expected to propagate along lightcones and thus ``disappear'' from the limit state, the Hadamard property of the initial state will be essential for the consistency of all constructions we perform in the perturbative treatment of NESS for interacting models.

The discussion in this section will be restricted to states for a complex scalar field model specified by \eqref{eq_kgo} with a smooth $V$. While we shall use some of the constructions for inhomogeneous linear toy models with $V$ not being smooth, it is not our intention in this paper to prove in all generality that standard results as well as the constructions we perform are valid for all non-smooth $V$ of a particular class. Rather, when dealing with the toy models with $V$ not being smooth, we shall be content with the fact that the all expressions we manipulate will be rigorously defined.

Hadamard states have correlation functions with a prescribed singularity structure. This singularity structure can be either fixed by considering asymptotic expansions of the correlation functions in terms of the geodesic distance of points, and by demanding that the singular terms in such an expansion have a universal form, or by using microlocal analysis in order to demand that the correlation functions of Hadamard states are distributions which all have the same universal wave front set \cite{Radzikowski}. Using standard results in microlocal analysis one finds that certain convolutions and pointwise products of correlation functions of Hadamard states which appear in Feynman graphs of perturbative correlation functions are well-defined distributions \cite{BF00}. For the purpose of this paper we shall be content with defining Hadamard states $\om$ on the Borchers-Uhlmann-algebra $\ABU$ of a complex scalar field model specified by \eqref{eq_kgo} with smooth $V$ by demanding that
$$
(x,y)\mapsto \om(\Phi^\sharp(x)\Phi^\sharp(y)) - \om_{\infty}(\Phi^\sharp(x)\Phi^\sharp(y))\quad \in C^\infty(\bR^d \times \bR^d)\qquad \forall\;\Phi^\sharp \in \{ \Phi, \overline{\Phi} \}\,.
$$
Here $\om_{\infty}$ is the ground state of the model considered, i.e. the KMS state $\om_{\beta,\mu,c}$ with $\beta=\infty$, $c=0$. Further details on the definition and properties of Hadamard states may be found e.g. in \cite{Khavkine}. The above definition encompasses also non-Gaussian states \cite{Sanders}. In fact, if $\omega$ is a Hadamard state and $\Psi$ a smooth solution of $K\Psi = 0$, then the coherent state $\om_\Psi$ defined as in \eqref{eq_coherent},\eqref{eq_coherent2} is Hadamard as well. For later purposes we define the smooth part of the two-point function by
\beq\label{eq_W}
W_{\om} \doteq \Delta_{+,\om} - \Delta_{+,\infty}\,.
\eeq
In addition to being smooth, $W_{\om}$ is real and symmetric, and it encodes all non-universal state-dependent information of a Gaussian and charge-conjugation-invariant Hadamard state. 

We define the initial state evolving into a NESS by prescribing all its correlation functions. Our initial states will be either Gaussian states or coherent excitations thereof and thus will be entirely determined by two-point correlations functions $\om(\Phi^\sharp(x)\Phi^\sharp(y))$, $\Phi^\sharp \in \{ \Phi, \overline{\Phi} \}$ and possibly a smooth solution $\Psi$ of the Klein-Gordon equation. The basic idea is as follows. We would like the correlation functions of the initial state to be a those of a  KMS state $\om_{\beta_1,\mu_1,c_1}$ for correlations evaluated at $x^1<0$, and those of a KMS state  $\om_{\beta_2,\mu_2,c_2}$ for correlations evaluated at $x^1>0$. However, in order to define a proper state, we need to specify the correlation functions for mixed ``left-right-correlations'' as well. The initial state constructed in \cite{Doyon} is a Gaussian state with vanishing two-point function for such mixed correlations. This requirement is unfortunately at variance with the Hadamard condition and thus we need to specify the mixed correlations in a different manner. Consequently, our initial state will be composed of the data of three states, one for each left and right correlations, and one for mixed correlations. In addition to ensuring that the initial state is Hadamard, we need to be sure that it satisfies the positivity (unitarity) requirement in order to be a proper state in the first place. 

In the following, we shall denote our initial states by $\om_G$ (for ``glued together'') for short rather then using a cluttered notation indicating all defining data. We shall first discuss Gaussian initial states. In addition, we shall consider only Gaussian initial states which are invariant under charge conjugation, meaning that
\beq\label{eq_charge}
\om_G(\overline{\Phi}(x)\Phi(y)) = \overline{\om_G(\overline{\Phi}(y)\Phi(x))} \doteq  2 \Delta_{+,G}(x,y)
\eeq
$$
\om_{G}(\Phi(x)\Phi(y))  = \om_{G}(\overline{\Phi}(x)\overline{\Phi}(y))  = 0\,.
$$
$\Delta_{+,G}$ solves the Klein-Gordon equation in both arguments. Thus it is uniquely specified by Cauchy data on an arbitrary Cauchy surface (initial-time surface) at time $x^0 = t$. These data are
$$
\Delta_{+,G}(x,y)|_{x^0=y^0=t}\,,\quad \partial_{x^0}\Delta_{+,G}(x,y)|_{x^0=y^0=t}\,,\quad \partial_{y^0}\Delta_{+,G}(x,y)|_{x^0=y^0=t}\,,\quad \partial_{x^0}\partial_{y^0}\Delta_{+,G}(x,y)|_{x^0=y^0=t}\,.
$$
We can glue together Cauchy data of three Gaussian and charge-conjugation-invariant states $\om_1, \om_2, \om_3$ by using a smooth partition of unity on the Cauchy surface, i.e. two smooth functions $\chi_1,\chi_2 : \bR^3 \to \bR$, $\chi_1+\chi_2 \equiv 1$. We may define 
\begin{align*}
\Delta_{+,G}(x,y)|_{x^0=y^0=t} & = \chi_1(\bx)\chi_1(\by)\Delta_{+,\om_1}(x,y)|_{x^0=y^0=t}+\chi_2(\bx)\chi_2(\by)\Delta_{+,\om_2}(x,y)|_{x^0=y^0=t} +\\
&+\chi_1(\bx)\chi_2(\by)\Delta_{+,\om_3}(x,y)|_{x^0=y^0=t}+\chi_2(\bx)\chi_1(\by)\Delta_{+,\om_3}(x,y)|_{x^0=y^0=t}
\end{align*}
and analogously for the other three Cauchy data. If $\chi_1,\chi_2$ are smooth and $\om_1,\om_2,\om_3$ are Hadamard then $\Delta_{+,G}$ turns out to be Hadamard as well, while positivity of $\Delta_{+,G}$, which is necessary and sufficient for positivity of $\om_G$, can be achieved by demanding that $\Delta_{+,\om_1}$ and $\Delta_{+,\om_2}$ dominate $\Delta_{+,\om_3}$ in the sense of Proposition \eqref{prop_domination}. We shall not state this construction and its properties as a proposition just yet, because we will immediately proceed to introduce a generalisation of this construction which simplifies notation and computations.

To this end, we consider a partition of unity $\psi + (1-\psi) = 1$ on the time-axis $\bR$. For simplicity we consider the case that $\psi$ is smooth (and real-valued) with $\psi(x^0)=1$ for $x^0>\epsilon>0$, $\psi(x^0)=0$ for $x^0<-\epsilon$. Consequently $\dot{\psi} \doteq \partial_0\psi$ is a smooth function with compact support contained in $(-\epsilon,\epsilon)$. We recall that any solution of the Klein-Gordon equation $K\Phi = 0$ may be written as $\Phi = \Delta f$ where $f$ is has compact support in time (timelike compact support) and $\Delta = \Delta_R - \Delta_A$ is the retarded-minus-advanced Green's function of the Klein-Gordon operator $K$, which is a linear operator with integral kernel $\Delta(x,y)$. It is a standard result that such an $f$ may be obtained for example as $f = K \psi \Phi$ and it can be thought of as a ``thickened'' or ``generalised'' initial data for $\Phi$. The identity $\Phi = \Delta K \psi \Phi$ for $\Phi$ which are solutions of the Klein-Gordon equation, i.e. on-shell configurations, may be re-written as 
$$
\Phi = \Delta g \Phi\qquad g = [K,\psi] = \ddot{\psi} + 2 \dot{\psi} \partial_0\,.
$$
We now define a decomposition of any smooth function $\Phi$ into a ``left'' and ``right'' piece based on these observations.

\begin{proposition} \label{prop_sigma} The maps 
$$
\sigma_i: C^{\infty}(\bR^d)\to C^{\infty}(\bR^d)\,,\qquad \sigma_i \Phi \doteq  \Delta\chi_i  g  \Phi\,,\qquad i\in\{1,2\}
$$
are well-defined and satisfy $K\sigma_1 = K \sigma_2 = 0$. If $\Phi $ solves $K\Phi = 0$ then $\sigma_1 \Phi + \sigma_2 \Phi = \Phi$. 
\end{proposition}
\begin{proof} On account of the properties of $\chi_i$ and $\psi$, $\chi_i g \Phi$ is a smooth function with timelike compact support. $\Delta$ maps such functions to smooth functions and $K \Delta = 0$ on smooth functions of timelike compact support. $\sigma_1+\sigma_2$ is the identity on solutions because $g \Delta$ is and $\chi_1 + \chi_2 = 1$.
\end{proof}

The maps $\sigma_i$ take the ``thickened'' initial data $g\Phi$, separate it into left and right part and then propagate this initial data within the future and past lightcones of this data to obtain the corresponding solution of the Klein-Gordon equation. In particular 
$$\supp \si_i\Phi \subset J\left((-\epsilon,\epsilon) \times \supp(\chi_i) \times \bR^2\right)\,,
$$
where $J$ denotes the causal future and past of a set, i.e. $\si_{1/2}\Phi$ vanishes ``on the right/left''. Using these maps, we can define for given charge-conjugation-invariant Gaussian states $\om_1$, $\om_2$, $\om_3$ a bidistribution $\Delta_{+,G}$ by 
\beq \label{eq_defdeltaG}
\Delta_{+,G} \doteq (\sigma_1 \otimes \sigma_1) \Delta_{+,\omega_1} +  (\sigma_2 \otimes \sigma_2) \Delta_{+,\omega_2} + (\sigma_1 \otimes \sigma_2+\sigma_2 \otimes \sigma_1) \Delta_{+,\omega_3}\,.
\eeq
The following two propositions show that $\Delta_{+,G}$ defines a quasifree Hadamard state for suitable $\om_i$. If $\Psi_1$ and $\Psi_2$ are two smooth solutions of the Klein-Gordon-equation, then we can define coherent states as in \eqref{eq_coherent}, \eqref{eq_coherent2} by $\om_{i,\Psi_i}\doteq \om_i \circ \iota_{\psi_i}$, $i=1,2$ and glue them and $\om_3$ together by defining 
\beq\label{eq_coherentG}
\om_{G,\Psi_G} \doteq \om_G \circ \iota_{\Psi_G}\,,\qquad \Psi_G \doteq \sigma_1 \Psi_1 + \sigma_2 \Psi_2\,.
\eeq
Just like $\om_{G}$, $\om_{G,\Psi_G}$ is a well-defined Hadamard state for suitably chosen $\om_i$.

\begin{proposition} \label{prop_Hadamard}
Let $\Delta_\sharp$ be any distribution whose wave front set is of Hadamard $\WF(\Delta_\sharp)\subseteq \WF(\Delta_{+,\infty})$ or anti-Hadamard $\WF(\Delta_\sharp)\subseteq \WF(\Delta^*_{+,\infty})$ type. Then for any $i,j\in \{1,2\}$ the distributions $(1\otimes \sigma_i) \Delta_\sharp$, $(\sigma_i\otimes 1) \Delta_\sharp$ and $(\sigma_i\otimes \sigma_j)  \Delta_\sharp$ are well-defined and their wave front set is contained in $\WF(\Delta_\sharp)$.
\end{proposition}
\begin{proof}
We consider only the case $(\sigma_i\otimes 1) \Delta_\sharp$, the other two cases follow by considering adjoints and / or iterating the argument. We can write the distributional kernel $(\sigma_i\otimes 1) \Delta_\sharp$ as an operator $\sigma_i \Delta_\sharp$. For any test function $f$, $\si_i \Delta_\sharp f$ is well-defined because $\Delta_\sharp f$ is smooth and $\si_i$ is well-defined on all smooth functions. Moreover, $\si_i \Delta_\sharp f$ is smooth as well. Finally $f \mapsto \si_i \Delta_\sharp f$ is continuous because it is a composition of continuous linear operators. Thus $\sigma_i \Delta_\sharp$ is a continuous linear map from test functions to smooth functions which has an integral kernel by the Schwartz kernel theorem. We can apply Theorem 8.2.14 of \cite{Hormander} to bound the wave front set of the composition noting that keeping the ``external vertex'' $x$ of $[\sigma_i \Delta_\sharp f](x)$ fixed, the composition is always on a compact set, namely on $J(x)\cap (\supp \dot\psi \times \bR^{d-1})$.
\end{proof}

\begin{theorem} \label{prop_Hadamardstate} If $\Delta_{+,\om_i}(\overline f,f)\ge \Delta_{+,\om_3}(\overline f,f)$ for $i=1,2$ and all $f\in C^\infty_0(\bR^d)$, then $\Delta_{+,G}$ constructed as in \eqref{eq_defdeltaG} defines a charge-conjugation-invariant, quasifree Hadamard state $\om_G$ on $\ABU$. For any smooth solutions of the Klein-Gordon equations $\Psi_1$, $\Psi_2$, the coherent state constructed as in \eqref{eq_coherentG} is a well-defined Hadamard state.
\end{theorem}
\begin{proof} $\Delta_{+,G}$ is a bisolution of the Klein-Gordon equation by the properties of $\sigma_i$. The antisymmetric part of $\Delta_{+,G}$ is $i\Delta/2$ because $\sigma_1 + \sigma_2$ is the identity on solutions of the Klein-Gordon equation. This together with the Proposition \ref{prop_Hadamard} implies that $\Delta_{+,G}$ has the correct wave front set \cite{Sahlmann:2000zr}. In order to show positivity, i.e. $\Delta_{+,G}(\overline f,f)\ge 0$ for all $f\in C^\infty_0(\bR^d)$, we observe that the adjoint maps $\sigma^*_i = g \chi_i \Delta$ ($\Delta^* = - \Delta$, $g^* = - g$) map smooth functions of compact support to functions of the same type because, for any $f\in C^\infty_0$, $\Delta f$ is smooth and has spacelike compact support while $g$ has timelike compact support. Using this and the facts that $\sigma_i$ commutes with complex conjugation and that $\sigma_1 + \sigma_2$ is the identity on solutions, we find
\begin{align*}
\Delta_{+,G}(\overline f,f) & =\Delta_{+,\om_1}(\overline{\sigma^*_1 f},\sigma^*_1 f) + \Delta_{+,\om_2}(\overline{\sigma^*_2 f},\sigma^*_2 f) + \Delta_{+,\om_3}(\overline{\sigma^*_1 f},\sigma^*_2 f) +  \Delta_{+,\om_3}(\overline{\sigma^*_2 f},\sigma^*_1 f) \\
&  \ge \Delta_{+,\om_3}(\overline f,f) \ge 0\,.
\end{align*}
The coherent state $\om_{G,\Psi_G}$ is a proper state because $\iota_{\Psi_G}$ maps $A^* A$ to $(\iota_{\Psi_G}(A))^* \iota_{\Psi_G}(A)$ for all $A \in \ABU$. It is a Hadamard state because $\Psi_G$ is smooth.
\end{proof}

We remark that the propositions \ref{prop_sigma} and \ref{prop_Hadamard}, and Theorem \ref{prop_Hadamardstate} also apply if we replace $\psi$ in the definition of $\sigma_i$ by the Heaviside distribution $\psi(x^0) = \Theta(x^0)$. This corresponds to gluing initial data on a Cauchy surface.

\subsection{Perturbative algebraic quantum field theory and KMS states for interacting models}
\label{sec_paqft}

In this section briefly review the algebraic treatment of perturbative interacting QFT models and the perturbative construction of KMS states for these models. Throughout this section we shall discuss only real scalar fields denoted by $\phi$, which are the models we shall study later. Consequently, the KMS states we consider shall have vanishing chemical potential. We will also not consider interacting condensates. Extensive reviews of pAQFT can be found in \cite{BF09,FredenhagenRejzner2}. The relation of pAQFT to more standard treatments is discussed for example in \cite{GHP}.

Perturbative algebraic quantum field theory (pAQFT) is a framework developed in \cite{BF00,DF2,HW01,HW02,HW05,BDF} based on ideas of deformation quantization. In this framework, observables are functionals of smooth field configurations $C^\infty(\bR^d)\ni \phi:\bR^d \to \bR$, $F:C^\infty(\bR^d)\to \bC$. These observables correspond to the observables of classical field theory. They are quantized by endowing them with a non-commutative product which is a formal power series in $\hbar$. Consequently, all relevant expressions in pAQFT are such formal power series. We shall, however, not make this explicit in our notation in favour of simplicity, and we shall also always set $\hbar = 1$. In order for the product between functionals to be well-defined, they need to have a certain regularity. To this end, we define the $n$-functional derivative of a functional in the direction of a configurations $\psi_1,\dots,\psi_n$ by setting
$$
\langle F^{(n)}(\phi),\psi_1\otimes \dots\otimes \psi_n\rangle \doteq \frac{d^n}{d\lambda_1\dots d\lambda_n} F(\phi + \lambda_1 \psi_1+\dots+\lambda_n \psi_n)|_{\lambda_1 =\dots=\lambda_n =  0}\,.
$$
We demand that all these derivatives exist as distributions of compact support; they are symmetric by definition. We denote this set of smooth functionals by $\mF$. For pAQFT we need to consider a particular subset $\Fmuc$ of $\mF$, the set of micro-causal functionals. They are defined as 
$$
\Fmuc \doteq \{F\in \mF | \WF(F^{(n)})\cap (\overline{V}^{+n}\cup \overline{V}^{-n}) = \emptyset\quad \forall n\}\,.
$$
Here $\WF$ denotes the wave front set, and $\overline{V^{\pm n}}\subset T^* \bR^d$ is the set of all points $(x_1,\dots,x_n,p^1,\dots,p^n)$ with all covectors $p^i$ being non-zero, time-like or lightlike and future-(+) or past(-)-directed. Two important subsets of $\Fmuc$ are the set of local functionals $\Floc$ which are such that all functional derivatives are supported on the total diagonal, and the set of regular functionals $\Freg$ whose functional derivatives are all smooth functionals. Elements of $\ABU$, i.e.~products of fields at different points smeared with smooth functions of compact support correspond to $\Freg$. By contrast, $\Floc$ contains functionals such as $F(\phi) = \int_{\bR^d} d^dx \;\phi(x)^4 f(x)$ with $f\in C^\infty_0(\bR^d)$, i.e.~products of the field at the same point. $\Fmuc$ contains all relevant products of observables in $\Floc$, which we shall define now.

Initially, elements of $\Fmuc$ are observables for a free / linear scalar field model; we shall promote them to interacting observables later. The linear models we consider have a Klein-Gordon operator of the form $K = \partial^2_0 - D + m^2$ with $m\ge 0$ and $D$ being the spatial Laplacian. Consider any bidistribution $H_+$ of Hadamard type for this model, i.e. $H_+ - \Delta_{+,\infty}$ is a smooth, real-valued and symmetric function on $\bR^d \times \bR^d$. We also require that $H_+$ solves the Klein-Gordon equation in both arguments. We may think of $H_+$ as being the two-point function $\Delta_{+,\om}$ of a Gaussian Hadamard state, but the positivity of $H_+$ is not important for the time being. For any such $H_+$ we define a product $\star_H$ on $\Fmuc$ by 
$$
F \star_H G \doteq \sum^\infty_{n=0}\left \langle F^{(n)}, H_+^{\otimes n} G^{(n)}\right\rangle\,.
$$
This is well-defined on account of the properties of the wave front sets of $F^{(n)}$, $G^{(n)}$ and $H_+$. The $ \star_H $ product is compatible with the involution $*$ on $\Fmuc$ which we define as $F^*(\phi)\doteq \overline{F(\phi)}$. With this definition we have $(F \star_H G)^*  = G^* \star_H F^*$. 

The second product we need is the time-ordered product. To this end, we define the Feynman propagator for $H_+$ as $H_F \doteq H_+ + i \Delta_A$, where $\Delta_A$ is the advanced Green's function for $K$. Setting $H_-(x,y)\doteq H_+(y,x) = \overline{H_+(x,y)}$ we also have $H_F = H_- + i \Delta_R$. Consequently, $H_F(x,y)$ is time-ordered in the sense that its equal to $H_+(x,y)$ ($H_+(y,x)$) if $x \gtrsim y$ ($x \lesssim y$), meaning that $x$ is in the forward (backward) lightcone of $y$. The time-ordered product is initially defined on regular functionals as 
$$
F \cdot_{T_H} G \doteq \sum^\infty_{n=0}\left \langle F^{(n)}, H_F^{\otimes n} G^{(n)}\right\rangle\,.
$$
For pAQFT, we need to define it on local functionals. However, when using the above formula for local functionals, one encounters pointwise products of (convolutions of) $H_F$. These products are not well-defined a priori and renormalisation is needed to make them well-defined. In pAQFT, the renormalised time-ordered product is recursively defined based on the causal factorisation $F \cdot_{T_H} G = F \star_H G$ if $F\gtrsim G$, meaning that the support of all functional derivatives of $F$ is causally later than the support of all functional derivatives of $G$. In addition to the causal factorisation, further axioms are needed in order to define a meaningful renormalised time-ordered product. We refer to \cite{BF00,DF2,HW01,HW02,HW05,BDF} for details. In the following we shall always assume that the time-ordered product is renormalised and thus well-defined. Renormalisation does not give a unique $\cdot_{T_H}$, but there are renormalisation ambiguities which for a renormalisable theory like $\phi^4$ in $d=4$ may be absorbed in the redefinition of the parameters of the theory. $\cdot_{T_H}$ is only defined on local functionals and time-ordered products thereof. Time-ordered products of local functionals are elements of $\Fmuc$ and can be multiplied with $\star_H$.

For a given $H_+$, we define the algebra $\AO$ as $\AO\doteq (\Fmuc,\star_H)$, i.e. the set of microcausal functionals endowed with the product $\star_H$. $\AO$ has an ideal $\mI$ which consists of all functionals $F$ that vanish for smooth solutions of the Klein-Gordon equation. We may consider the quotient $\AOon \doteq \AO / \mI$. We call $\AO$ ($\AOon$) the off-shell (on-shell) algebra of observables of the linear model specified by $K$. The perturbative constructions in pAQFT are initially performed off-shell because this turns out to be more convenient. The subalgebra of $\AOon$ consisting of (equivalence classes of) regular functionals is isomorphic to the Borchers-Uhlmann algebra $\ABU$ discussed in Section \ref{sec_KMS}. $\AOon$ may be viewed as a suitable topological completion of $\ABU$ and one can show that all Hadamard states on $\ABU$ have a unique continuous extension to $\AOon$, see \cite{HollandsRuan} for details. As already anticipated, $\AOon$ contains pointwise products of the fields, which are not contained in $\ABU$. In fact, elements of $\AO$ may be interpreted as quantities normal-ordered with respect to the symmetric part $H_s \doteq \frac12 (H_+ + H_-)$ of $H_+$. For example $F(\phi) = \int_{\bR^d} d^dx\;\phi(x)^2 f(x) $ corresponds in more standard notation to
$$
\int_{\bR^d} d^dx\; \wick{\phi(x)^2}_H f(x) \,,\qquad \wick{\phi(x)^2}_H\; \doteq \;\lim_{y\to x} (\phi(x)\phi(y) - H_s(x,y)\b1)\,.
$$
Consequently, $\star_H$ and $\cdot_{T_H}$ encode Wick's theorem for normal-ordered and time-ordered quantities, respectively.

Clearly, the definition of $\AO$ depends on the choice of $H_+$. However, algebras constructed with different $H_+$ are isomorphic. Consequently, one may think of $\AO$ as an abstract algebra which is represented by choosing an $H_+$. One may also think of the freedom in the choice of $H_+$ as being the ordering freedom one has in quantising products of classical quantities. For two choices $H_+$, $H'_+$, the corresponding products $\star_H$ and $\star_{H'}$ are isomorphic in the following sense
$$
F \star_{H'} G = \gamma_W \left(\gamma^{-1}_W(F) \star_H \gamma^{-1}_W(G)\,,\right)
$$
$$
\gamma_W = \exp\left\langle W, \frac{\delta}{\delta \phi}\otimes\frac{\delta}{\delta \phi} \right\rangle \,,\qquad W \doteq H'_+  - H_+\,.
$$
We recall that $W$ is smooth such that $\gamma_W$ is well-defined.

Let $\om$ be a Gaussian Hadamard state. The expectation values of observables in such a state are computed as follows. We consider the realisation of $\AOon$ given by choosing $H_+ = \Delta_{+,\omega}$. Then for any $A\in\AOon$ we have 
$$
\om(A) \doteq A|_{\phi = 0}\,.
$$
If instead we consider the realisation of $\AOon$ provided by a different choice of $H_+$, then we have
$$
\om(A) \doteq \gamma_W(A)|_{\phi = 0}\,,\qquad W  =  \Delta_{+,\omega} - H_+\,.
$$
This difference is important because, for example when speaking of a ``$\phi^4$''-potential, we usually mean $\wick{\phi^4}_\infty$, i.e. $\phi^4$ normal-ordered with respect to the vacuum state. Following the above discussion, the expectation value of $\wick{\phi^2}_\infty$ in the KMS state $\om_\beta$ is computed -- here for simplicity without smearing -- as
$$
\om_\beta(\wick{\phi(x)^2}_\infty) = \gamma_{W_\beta}(\phi(x)^2)|_{\phi=0} = W_\beta(x,x)\,,\qquad W_\beta \doteq \Delta_{+,\beta} - \Delta_{+,\infty}\,.
$$
For the remainder of this section, we shall denote $\star$- and time-ordered products simply by $\star$ and $\cdot_T$ omitting the reference to an $H_+$.

We shall now discuss the perturbative construction of interacting models. To this end, we consider a local (interaction) functional $V$ which is arbitrary for the time being. $V$ is of the form $V(\phi) = \int_{\bR^d} d^dx \;v(\phi(x)) f(x)$, $f\in C^\infty_0(\bR^d)$. The model defined by $V$ corresponds to the equation of motion
$$
K \phi + v^{(1)}(\phi) = 0\,,\qquad v^{(1)}(\phi) \doteq \frac{d}{d\phi}v(\phi)\,.
$$
$f$ plays the role of an infra-red cutoff. With all ultraviolet singularities present removed by (implicit normal-ordering and) renormalisation of $\cdot_T$, all quantities we shall write in the following will be well-defined on account of the presence of $f$ in $V$. When computing physical quantities such as expectation values in a KMS state, we shall usually consider the adiabatic limit $f\to 1$ of expectation values. The existence of this limit is non-trivial and in principle needs to be proven case by case for each state. The advantage of pAQFT is that all basic constructions are well-defined in a state-independent fashion.

With this in mind, we define for given $V$ the local $S$-matrix $S(V)$ by
$$
S(V) = \exp_{\cdot_T}(i V) = \sum^\infty _{n=0}\frac{i^n}{n!} \underbrace{V\cdot_T \dots \cdot_T V}_{n~\text{times}}\,.
$$
We then define the Bogoliubov-map $\moller:\Floc \to \Fmuc$ by
$$
\moller(F) \doteq S(V)^{\star -1}\star (S(V) \cdot_T F)
$$
$\moller(F)$ may be interpreted as the interacting version of $F$. In fact, a more apt interpretation is to consider $\moller(F)$ as a representation of the interacting version of $F$ in the algebra of the free field $\AOon$. It is evident from the definition that $\moller$ is not only well-defined on $\Floc$, but on arbitrary time-ordered products of local functionals, the set of which we denote by $\FTloc$. We denote the subalgebra of $\AOon$ $\star$-generated by $\moller(F)$, with $F$ varying over $\FTloc$, by $\AV$. The definition of $\moller$ is reminiscent of the Gell-Mann-Low formula for interacting time-ordered products in the interacting vacuum state $\Om_V$ defined in terms of expectation values in the free vacuum $\Om_0$
$$
\langle T_V(\phi(x_1)\dots \phi(x_n) \rangle_{\Om_V} = \frac{\left \langle T(e^{iV} \phi(x_1)\dots \phi(x_n)) \right \rangle_{\Om_0} }{\left \langle T(e^{iV}) \right \rangle_{\Om_0}}\,.
$$
In fact, the expectation value of $\moller(\phi(x_1)\cdot_T\dots\cdot_T\phi(x_n))$ in the free vacuum $\om_\infty$ gives the Gell-Mann-Low formula in the adiabatic limit \cite{Duetsch}.

We now review the construction of interacting KMS states due to \cite{Lindner, FredenhagenLindner} based on earlier work in \cite{HW03}. To this end, we define the free time-evolution $\alpha_t$ via $\alpha_t(\phi(x^0,\bx)) \doteq \phi(x^0+t,\bx)$. This defines $\at$ for any $F\in\AO$. We define the interacting time-evolution $\atV$ by
$$
\atV(\moller(F)) \doteq \moller(\at(F)) = \at\mR_{\alpha_{-t}(V)}(F)\,.
$$
In the adiabatic limit we have $\atV = \at$, but prior to this the cutoff function present in $V$ is responsible for $\atV \neq \at$. We choose this cutoff function to be of the form $f(x^0,\bx) = \psi(x^0) h(\bx)$ where $h$ is smooth and compactly supported and $\psi$ is smooth with $\psi(x^0) = 1$ for $x^0>\epsilon$ and $\psi(x^0) = 0$ for $x^0<-\epsilon$. This $f$ is not compactly supported in time, but $\moller(F)$ is still well-defined for any local $F$ because $\mR_{V+V'}(F) = \mR_{V}(F)$ for all $V'$ with $V' \gtrsim F$. \cite{FredenhagenLindner} have shown that $\atV$ and $\at$ are intertwined by a formal unitary $U_V(t)$ for all $F$ supported in the region where $\psi=1$ and $t>0$.
$$
\atV(\moller(F)) = U_V(t) \star \at(\moller(F)) \star U^*_V(t)\,,
$$
$$
U_V(t) = 1 + \sum^\infty_{n=1} \int_{t S_n} dt_1 \dots dt_n \;\alpha_{t_n}(K_V) \star \dots \star \alpha_{t_1}(K_V)\,,
$$
\beq\label{eq_KV}
K_V \doteq \frac{1}{i}\frac{d}{dt}U_V(t) |_{t=0} = \moller(V(\dot{\psi}h))\,,\qquad V(\dot{\psi}h) \doteq \int_{\bR^d}d^dx \; v(\phi(x))\, \dot{\psi}(x_0) h(\bx)\,.
\eeq
Here $S_n$ denotes the $n$-dimensional unit simplex. $U_V(t)$ formally corresponds to $\exp(it(H_0 + V))\exp(-it H_0)$ with $H_0$ the free Hamiltonian. Using these definitions, we set for any interacting observable $A\in\AV$ 
\beq\label{eq_araki}
\ombV(A) \doteq \frac{\omb(A \star U_V(i\beta))}{\omb( U_V(i\beta))}
\eeq
The authors of \cite{FredenhagenLindner} argue that $\ombV$ is well-defined and satisfies the KMS condition for $\atV$, $t>0$ and observables localised in the region where $\psi = 1$. They also show that for such observables $\ombV(A)$ is independent of the particular choice of $\psi$ and of $\epsilon$ defining the support of $\dot\psi$. Thus, observables don't ``see'' the past temporal cutoff and in that sense the definition of $\ombV(A)$ already includes a temporal adiabatic limit. The authors of \cite{FredenhagenLindner} also show that the spatial adiabatic limit $h\to 1$ of expectation values of $\ombV$ exists for $m>0$, and thus that interacting KMS states exist in the adiabatic limit. The adiabatic limit of the massless case has been treated in \cite{DHP}.

We conclude this section by indicating how expectation values of $\ombV$ can be computed explicitly in terms of Feynman diagrams. One can show that \cite{FredenhagenLindner}
\beq\label{eq_KMSFeynman}
\ombV(A)=\sum^\infty_{n=0}(-1)^n\int_{\beta S_n}\ \omb^\text{conn.}\left(A  \otimes  \alpha_{iu_1}(K_V) \otimes\dots \otimes  \alpha_{iu_n}(K_V)\right)
du_1\dots du_n,
\eeq
where $\omb^\text{conn.}$ indicates the connected part of correlations functions defined in the usual fashion. For example, 
$$
\omb^\text{conn.}(A\otimes B)\doteq \omb(A\star B) - \omb(A)\omb( B)\,.
$$

\section{NESS for linear models}
\label{sec_NESSlinear}

\subsection{Homogeneous models}
\label{sec_NESShomogeneous} 

In this section we construct initial states and study their evolution into a NESS for models of a complex linear scalar field $\Phi$ with Klein-Gordon equation \eqref{eq_kgo} and $V = m^2$, $m\ge 0$ in $d\ge 2$ spacetime dimensions. We shall restrict $d$ to be 4 when discussing interacting models. Following the motivation given in Section \ref{sec_gluing}, we prepare an initial state as follows. We choose three Gaussian KMS states $\om_{\beta_i,\mu_i}$ with
 \beq\label{eq_betamu}
0\le  \beta_1 \le \beta_3 \,,\qquad 0\le \beta_2 \le \beta_3 \,,\qquad \mu_3 \le \mu_1 \le m\,,\qquad\mu_3 \le \mu_2 \le m.
 \eeq
 We further choose a smooth partition of unity $\chi_1 + \chi_2 = 1$ of the $x^1$-axis which is such that $\chi_1(x^1) = 0$ for $x^1> \epsilon$ and $\chi_1(x^1) = 1$ for $x^1< -\epsilon$, as well as a partition of unity $\psi + (1-\psi) = 1$ of the time axis with $\psi(x^0) = 1$ for $x^0> \epsilon$ and $\psi(x^0) = 0$ for $x^0< -\epsilon$. For simplicity we consider only a smooth $\psi$, but the results we obtain are valid also for $\psi$ being the Heaviside distribution. We define a bidistribution by
\beq\label{eq_initialstate} 
\Delta_{+,G} \doteq (\sigma_1\otimes \sigma_1)\Delta_{+,\beta_1,\mu_1}+(\sigma_2\otimes \sigma_2)\Delta_{+,\beta_2,\mu_2}+(\sigma_1\otimes \sigma_2+\sigma_2\otimes \sigma_1)\Delta_{+,\beta_3,\mu_3}\,.
\eeq
 Our analysis so far implies that $\Delta_{+,G}$ defines a Gaussian, charge-conjugation-invariant Hadamard state $\om_G$ on $\ABU$. We can go one step further and prepare a non-trivial condensate as an initial state. Recalling the construction of coherent states in \eqref{eq_coherent}, \eqref{eq_coherent2} we set
 \beq\label{eq_initialcondensate} 
 \mu \doteq \mu_1 = \mu_2 \doteq \pm m\,,\qquad c_1\,,c_2\in\bC\,,\qquad \Psi_i(x) \doteq c_i e^{i\mu x^0}\qquad i=1,2
 \eeq
 $$
 \Psi_G \doteq \sigma_1 \Psi_1 +  \sigma_2 \Psi_2\,,\qquad \om_{G,\Psi_G} \doteq \om_G \circ \iota_{\Psi_G}  \,.
 $$
Again, our discussion so far implies that $\om_{G,\Psi_G} $ is a Hadamard state, which is, however, neither Gaussian nor charge-conjugation invariant. In the following cases
\begin{align}
m>0\,,\quad m-\max(\mu_i)>0\,,\quad d\ge 2\notag\\
\label{eq_lowdim} m>0\,,\quad m-\max(\mu_i)=0\,,\quad d\ge 3\\
m=0\,,\quad \max(\mu_i)=0\,,\quad d\ge 4\notag
\end{align}
all states $\om_G$ is composed of are well-defined. In all other cases, at least one of them is ill-defined as a state on $\ABU$, and $\om_G$ is only well-defined on the subalgebra of $\ABU$ generated by spatial derivatives of the fields $\partial_i \Phi$, $\partial_i \overline{\Phi}$ which we denote by ${\ABU}_{\partial_i}$\footnote{More precisely, ${\ABU}_{\partial_i}$ is the subalgebra of $\ABU$ generated by $\b1$, $\Phi(f)$, $\overline{\Phi}(f)$ with $f\in\{f\in C^\infty_0(\bR^d)| f = \sum^{d-1}_{i=1} n^i \partial_i g, \, \bn\in\bR^{d-1},\, g\in C^\infty_0(\bR^d)\}$.}. Correspondingly, only $\partial_i \otimes \partial_j \Delta_{+,G}$ is well-defined if \eqref{eq_lowdim} is not satisfied. The subalgebra ${\ABU}_{\partial_i}$ is insensitive to condensates because $\partial_i \Psi_G = 0$. Thus, we consider condensates only for the cases \eqref{eq_lowdim}. We collect these observations in the following

\begin{proposition} If one of the conditions in \eqref{eq_lowdim} is satisfied, the bidistribution $\Delta_{+,G}$ defined in \eqref{eq_initialstate}, with parameters as in \eqref{eq_betamu} and $\sigma_i$ defined in Proposition \ref{prop_sigma} defines a charge-conjugation-invariant, Gaussian Hadamard state $\om_G$ on $\ABU$, and the coherent state $\om_{G,\Psi_G}$ constructed as in \eqref{eq_initialcondensate} is a Hadamard state on $\ABU$. If none of the conditions in \eqref{eq_lowdim} is satisfied, the distributions $\partial_i \otimes \partial_j \Delta_{+,G}$, $i,j\in\{1,\dots,d-1\}$ define a charge-conjugation-invariant, Gaussian Hadamard state $\om_G$ on ${\ABU}_{\partial_i}$, the subalgebra of $\ABU$ generated by spatial derivatives of the fields $\partial_i \Phi$, $\partial_i \overline{\Phi}$.
\end{proposition}
\begin{proof}
The statement follows from Proposition \ref{prop_domination} and Theorem \ref{prop_Hadamardstate}.
\end{proof}

The state $\om_{G,\Psi_G}$ is by construction a KMS state with parameters $\beta_1,\mu_1,c_1$ for all observables localised in $M_1 \doteq \bR^d \setminus  J\left((-\epsilon,\epsilon) \times \supp(\chi_2) \times \bR^{d-2}\right)$, i.e. for all polynomials of the field $\Phi$ and its complex conjugate smeared with testfunctions supported in $M_1$. Similarly, $\om_{G,\Psi_G}$ is a KMS state with parameters $\beta_2,\mu_2,c_2$ for all observables localised in $M_2 \doteq \bR^d \setminus  J\left((-\epsilon,\epsilon) \times \supp(\chi_1) \times \bR^{d-2}\right)$. For example, the expectation value of the normal ordered square $\wick{|\Phi(x)|^2}$ for the case of $\psi = \Theta$ at $x^0=0$ is
\begin{align*}
\om_{G,\Psi_G}(\wick{|\Phi(0,\bx)|^2}) &\doteq \lim_{x,y\to (0,\bx)}\left(\om_{G,\Psi_G}(\overline{\Phi}(x)\Phi(y)) - \om_{\infty}(\overline{\Phi}(x)\Phi(y))\right) \\
&=  \chi^2_1(x^1) \left(\om_{\beta_1,\mu_1}(\wick{|\Phi(0,\bx)|^2})  + |c_1|^2\right)+ \chi^2_2(x^1)\left( \om_{\beta_2,\mu_2}(\wick{|\Phi(0,\bx)|^2}) + |c_2|^2\right) \\
&\quad + 2 \chi_1(x^1)\chi_2(x^1)\left(\om_{\beta_3,\mu_3}(\wick{|\Phi(0,\bx)|^2}) + \Re(\overline{c_1}c_2)\right)\,,
\end{align*}
where, $\om_{\beta_i,\mu_i}(\wick{|\Phi(0,\bx)|^2}) = 2 W_{\beta_i,\mu_i}(0,0)$, because $W_{\beta_i,\mu_i}(x,x)$ is constant, see \eqref{eq_W} for the definition of $W_{\beta_i,\mu_i}$. 

As already explained, we need the third state $\om_{\beta_3,\mu_3}$ in order to specify non-trivial two-point cross-correlations which are necessary to assure positivity (unitarity) of $\om_{G}$ and then of $\om_{G,\Psi_G}$. The initial state considered in \cite{Doyon} (with Dirichlet boundary conditions) corresponds to our construction of $\om_{G}$ with $\chi_2$ being a Heaviside distribution; in particular, condensates are not treated in \cite{Doyon}. With this choice it is consistent to set the two-point cross-correlations at initial time to be vanishing. However, the initial state constructed in this way is not of Hadamard type. For technical reasons, we chose the intermediate state $\om_{\beta_3,\mu_3}$, which loosely corresponds to the ``probe'' considered in \cite{HollandsLongo}, to be ``colder'' and with lower chemical potential than the semi-infinite reservoirs. However, there is no obvious physical reason for doing so. 

 In the case of $\om_{G,\Psi_G}$, we need $\mu_1=\mu_2=\pm m$, in order for the limit NESS to (exist and) be stationary with respect to $\alpha^\mu_t$ as in \eqref{eq_timeevol} with a definite $\mu=\mu_1=\mu_2$. If we consider only $\om_{G}$, i.e. no condensate, then the chemical potentials of the two semi-infinite reservoirs can be arbitrary, as long as they are bounded by $m$. This is because the limit NESS in the absense of condensates turns out to be stationary with respect to $\alpha^\mu_t$ for any $\mu$.
 
We shall now investigate the asymptotic time-evolution of the initial state. We set
\beq\label{eq_defNESS}
\om_N \doteq \lim_{t\to\infty}\om_G \circ \alpha^0_t \,,\qquad \om_{N,\Psi_N}\doteq \lim_{t\to\infty}\om_{G,\Psi_G} \circ \alpha^\mu_t
\eeq
 and find the following result. 
 \begin{theorem} \label{prop_NESS} (1) Assume that one of the conditions in \eqref{eq_lowdim} is satisfied. For any admissible choice of data $\beta_1, \beta_2, \beta_3, \mu_1, \mu_2, \mu_3, \psi, \chi_1, c_1, c_2$ the states $\om_N$ and $\om_{N,\Psi_N}$ defined in \eqref{eq_defNESS} are well-defined on $\ABU$ and independent of $\beta_3$, $\mu_3$, $\chi_1$, $\psi$. $\om_N$ is an $\alpha^0_t$-invariant, charge-conjugation-invariant, Gaussian Hadamard state with two-point function
\beq\label{eq_NESS2pf}
\Delta_{+,N}(x,y)\doteq \frac{1}{(2\pi)^{d-1}} \int_{\bR^{d-1}} d\bp \frac{1}{2\omp}\sum_{s=\pm 1} e^{i s \omp (x^0 - y^0)} e^{i \bp( \bx-\by)} s\, b_{\beta(s p_1),\mu(s p_1)} (s\omp)\,,
\eeq
with $\omp$, $b_{\beta,\mu}$ defined in \eqref{eq_KMS2pfexplicit} and
$$
\beta(p_1) \doteq \Theta(p_1) \beta_1  + \Theta(-p_1) \beta_2\,,\qquad \mu(p_1) \doteq \Theta(p_1) \mu_1  + \Theta(-p_1) \mu_2\,.
$$
$\om_{N,\Psi_N}$ is an $\alpha^\mu_t$-invariant coherent excitation of $\om_{N}$, i.e. 
$$
\om_{N,\Psi_N} = \om_N \circ \iota_{\Psi_N}\,,\qquad \Psi_N(x^0)=  \frac{c_1+c_2}{2} e^{i\mu x^0}\,.
$$
(2) Assume that none of the conditions in \eqref{eq_lowdim} is satisfied.  For any admissible choice of data $\beta_1, \beta_2, \beta_3, \mu_1, \mu_2, \mu_3, \psi, \chi_1, c_1, c_2$ the state $\om_N$ is well-defined on ${\ABU}_{\partial_i}$ -- the subalgebra of $\ABU$ generated by spatial derivatives of the fields $\partial_i \Phi$, $\partial_i \overline{\Phi}$ -- and independent of $\beta_3$, $\mu_3$, $\chi_1$, $\psi$. It is an $\alpha^0_t$-invariant, charge-conjugation-invariant, Gaussian Hadamard state with two-point functions $\partial_i \otimes \partial_j \Delta_{+,N}$, $i,j\in\{1,\dots,d-1\}$ defined by applying $\partial_i \otimes \partial_j$ to \eqref{eq_NESS2pf}.
 \end{theorem}
 \begin{proof} Since $\om_G$ is a Gaussian state, it is necessary and sufficient to show that $\lim_{t\to \infty}\Delta_{+,G}(x^0+t,\bx,y^0+t,\by) = \Delta_{+,N}(x,y)$ in the sense of distributions. This convergence is computed in Appendix \ref{sec_convergence2pf}. For (2) one has to compute instead the convergence of $\partial_i\otimes \partial_j \Delta_{+,G}$ to $\partial_i\otimes \partial_j \Delta_{+,N}$. We don't do this explicitely in Appendix \ref{sec_convergence2pf} because the computational steps are entirely analogous to the case without derivatives. $\om_N$ is Hadamard because $W_N \doteq \Delta_{+,N}- \Delta_{+,\infty}$ ($\partial_i\otimes \partial_j W_N \doteq \partial_i\otimes \partial_j (\Delta_{+,N}- \Delta_{+,\infty})$ for (2)) is a smooth function, as can be inferred from \eqref{eq_NESS2pf}. Moreover, the form of $\Delta_{+,N}$ implies that  $\om_N \circ \alpha^\mu_t = \om_N$ for all $\mu$. Given the convergence of $\om_N$, the convergence of $\om_{N,\Psi_N}$ is equivalent to the convergence $\lim_{t\to \infty} e^{-i\mu t} \Psi_G(x^0+t) = \Psi_N(x^0)$. The latter is computed in Appendix \ref{sec_convergencecondensate}.
 \end{proof}
 
The state $\om_N$, for $\mu_1=\mu_2=0$, corresponds to the NESS found in \cite{Doyon}. It satisfies the analyticity part of the KMS condition with respect to $\alpha^\mu_t$ for arbitrary $\mu$ in a strip corresponding to $\min(\beta_1,\beta_2)$ but satisfies the full KMS condition with respect to $\alpha^\mu_t$ only if $\mu=\mu_1=\mu_2$ and $\beta_1 = \beta_2$. The convergence of  $\om_{G,\Psi_G}$ to $\om_{N,\psi_N}$ is $O(t^{-1})$ in the sense that $\Delta_{+,G}(x^0+t,\bx,y^0+t,\by) = \Delta_{+,N}(x,y) +  O(t^{-1})$ and $e^{-i\mu t} \Psi_G(x^0+t) = \Psi_N(x^0)+  O(t^{-1})$ for large $t$. This is evident from the computation in Appendix \ref{sec_convergence} and, for $\omega_N$, corresponds to the Dirichlet-case studied in \cite{Doyon}.  We will analyse further properties of $\om_N$ in the sections to come.

\begin{figure}[ht]
\begin{minipage}[c]{0.6\textwidth}
\includegraphics[scale=0.6]{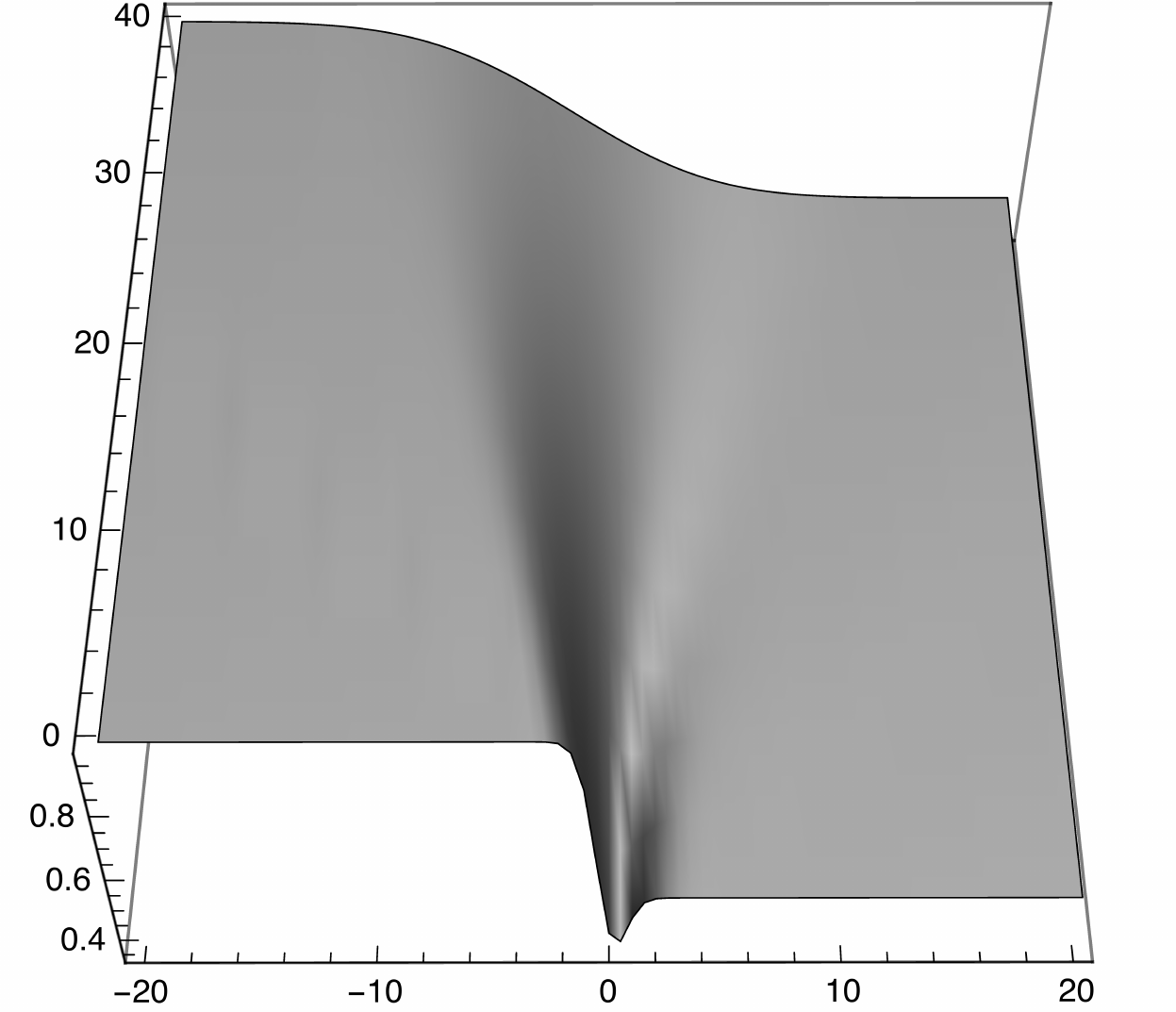}
\end{minipage}
\begin{minipage}[c]{0.4\textwidth}
\caption{$\omega_G(\rho(x))/\omega_{\beta_1}(\rho(x))$ for $\beta_1 m = 50$, $\beta_2$ such that $\omega_{\beta_2}(\rho)=\omega_{\beta_1}(\rho)/2$, $a m = 10^2\sqrt{\beta_1 m}$, $\mu_1=\mu_2=\mu_3=0$, $\beta_3 = \infty$, $d=4$, $\partial_1 \chi_2$  a normalised Gaussian with width $a$. $x^0$ and $x^1$ in units of $a$ with $x^1$ on the horizontal axis and $x^0$ on the vertical axis.\label{fig_phi2}}
\end{minipage}
\end{figure}

\begin{figure}[ht]
\begin{minipage}[c]{0.6\textwidth}
\includegraphics[scale=0.6]{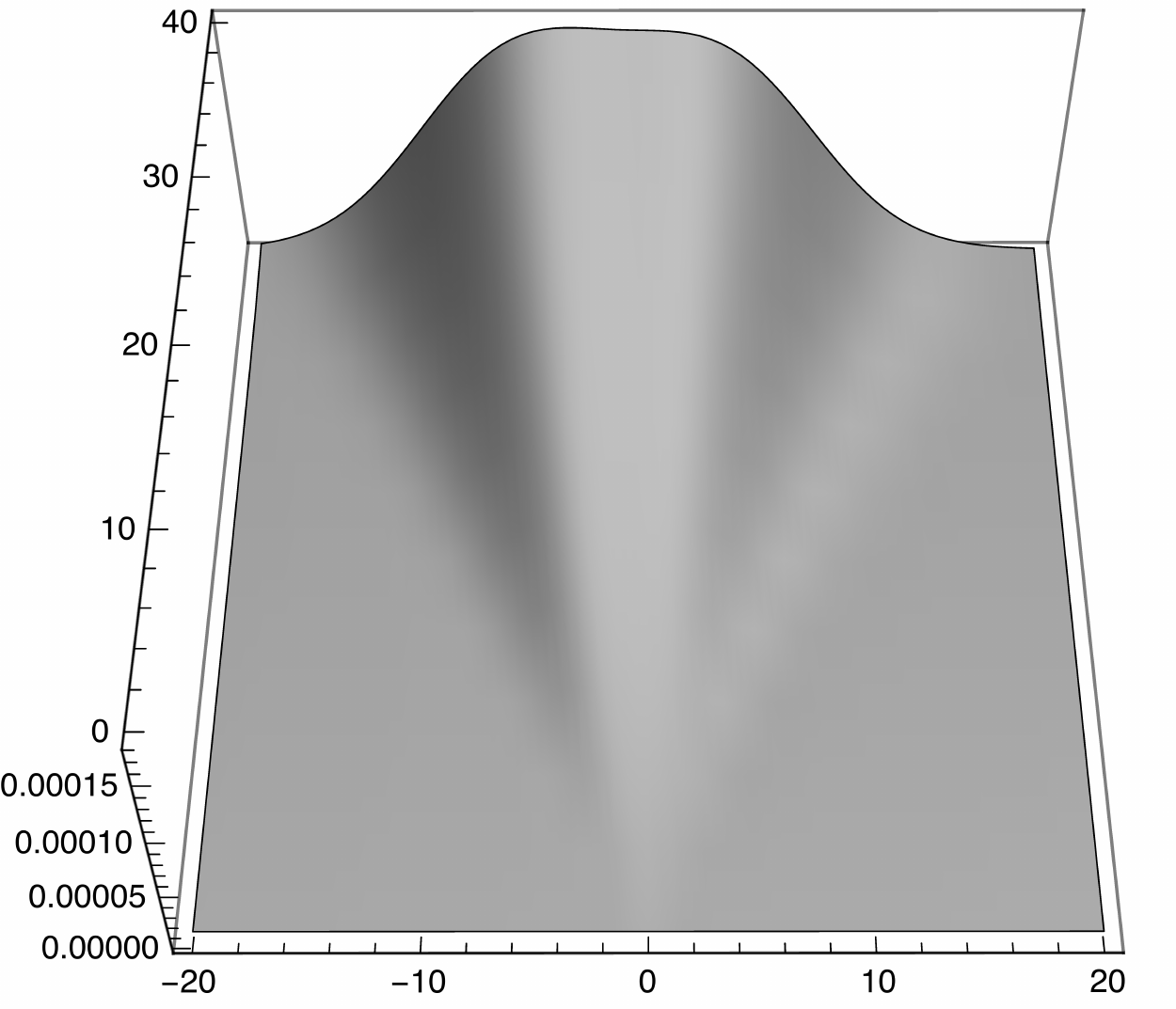}
\end{minipage}
\begin{minipage}[c]{0.4\textwidth}
\caption{$ \omega_G(j_1(x))  /\omega_N(j_1(x)) $ with the same parameters and specifications as in Figure \ref{fig_phi2}. \label{fig_j}}
\end{minipage}
\end{figure}

Using numerical computations, it is possible to investigate the dynamics of the convergence of $\om_G$ at finite times by considering the observables $\rho = T_{00}$ and $j_1 = T_{01}$ with
$$
T_{\mu\nu} = 2 \Re \left(\partial_\mu \overline{\Phi} \partial_\nu \Phi\right) -\frac12 g_{\mu\nu}(\partial_\rho \overline{\Phi} \partial^\rho \Phi + m^2 |\Phi|^2)
$$
In order to simplify computations we choose $\beta_3 = \infty$, $\mu_1=\mu_2=\mu_3=0$ and find
$$
\omega_G(\rho(x)) = \left[(\partial_{0} \otimes \partial_{0} + \partial_{i} \otimes \partial^i + m^2) W_G\right](x,x)\,,\qquad \omega_G(j_1(x)) = \left[ \partial_0 \otimes \partial_1 W_G\right](x,x)
$$
$$
W_G \doteq  \Delta_{+,G} -  \Delta_{+,\infty} =  (\sigma_1 \otimes \sigma_1) W_{\beta_1} + (\sigma_2 \otimes \sigma_2) W_{\beta_2}\,.
$$
We also choose $\psi = \Theta$, $d=4$ and $\chi_1$ such that $\partial_1\chi_1$ is a normalised Gaussian with width $a$. The details of the computation are collected in Appendix \ref{sec_convergencefinite}.

As one can infer from Figure \ref{fig_phi2}, the convergence is so ``slow'' that the shape of the initial temperature profile is not maintained as in the hydrodynamic approximation for strong interactions discussed in \cite{Doyon2}. Still one can say something about the expansion speed of the transition region, which is clearly subluminal. \cite{Doyon2} found that for scale-invariant dynamics in the strongly coupled regime the intermediate region at large times is tilted towards the lower temperature such that colder side approaches the NESS faster than the hotter side. In our case the tilt, which is best seen in Figure \ref{fig_j}, is rather towards the higher temperature. However, it is also possible that the tilt or rather the asymmetry one can see is a remnant of the asymmetry in the magnitude of the gradient of the initial temperature profile.
 
\subsection{NESS and local thermal equilibrium states}

Even if $\om_N$ is not a KMS state with respect to $\alpha_t$, one may wonder if $\om_N$ is a KMS state in a different rest-frame. The authors of \cite{Doyon} have observed that, for vanishing chemical potentials, this is the case if any only if $m=0$ and $d=2$, which is a special case of CFTs in $d=2$ which all exhibit this behaviour \cite{Doyon3,HollandsLongo}. \cite{BOR} have introduced a concept of local thermal equilibrium (LTE) states for describing near-equilibrium non-equilibrium situations, essentially (for real scalar fields, i.e. no chemical potentials) these states have to satisfy
$$
\omega_{LTE}(A(x)) = \omega_{\beta}(A)|_{\beta = \beta(x)}
$$
for sufficiently many local observables $A$, see also \cite{GPV}. This definition may be generalised to mixtures of $\omega_{\beta}|_{\beta = \beta(x)}$ in different rest frames. Using only the Wick-square and the components of the stress-energy tensor as test-observables, we have verified that, for vanishing chemical potentials, $\omega_N$ is not a mixed LTE except for $m=0, d=2$, where it is a proper KMS state in a different rest-frame, as already mentioned. We shall not provide the details of this computation, but only remark that, for $m>0$ or $d>2$, $\om_N$ is too anisotropic to be an LTE state. This suggests a possible anisotropic generalisation of LTE states by allowing for mixtures of thermal currents.

\subsection{NESS are ``modewise'' KMS}

The form of the two-point function $\Delta_{+,N}$ in \eqref{eq_NESS2pf} suggests that $\om_N$ is a KMS state in the following weaker sense: It is a KMS state with respect to $\alpha^{\mu_1}_t$ ($\alpha^{\mu_2}_t$) at inverse temperate $\beta_1$ ($\beta_2$) for all modes which are right-moving (left-moving) in the $x^1$-direction. This observation has already been made in \cite{Doyon} (for vanishing chemical potentials). We can formalise this observation rigorously in the following way. For given $\beta_1$, $\beta_2$, $\mu_1$, $\mu_2$ we define a time evolution $\alpha^{\beta_1,\beta_2,\mu_1,\mu_2}_t$ which is state-independent but well-defined only in states whose correlation functions are Schwartz distributions in all variables. Denoting these correlation functions collectively (for all possible combinations of $\Phi$ and $\overline{\Phi}$) by 
$$
\Delta_{\om,n}(x_1,\dots,x_n)\doteq \om(\Phi^\sharp(x_1),\dots,\Phi^\sharp(x_n))\,,\quad \Phi^\sharp\in \{\Phi,\overline{\Phi}\}\,,
$$
we can consider the Fourier transforms $\widetilde \Delta_{\om,n}$ of $\Delta_{\om,n}$ such that
$$
\Delta_{\om,n}(x_1,\dots,x_n) = \frac{1}{\sqrt{2\pi}^d} \int_{\bR^{nd}} dp^1\dots dp^n \exp\left(i\sum^n_{i=1}(\bp^i\cdot \bx_i - p^i_0 x^0_i )\right)\widetilde \Delta_{\om,n}(p^1,\dots,p^n)\,.
$$
We now set
$$
\om\left(\alpha^{\beta_1,\beta_2,\mu_1,\mu_2}_t(\Phi^\sharp(x_1)),\dots,\alpha^{\beta_1,\beta_2,\mu_1,\mu_2}_t(\Phi^\sharp(x_n))\right)\doteq \frac{1}{\sqrt{2\pi}^d} \int_{\bR^{nd}} dp^1\dots dp^n\quad \times 
$$
\beq\label{eq_timeevolutionleftright}
\times\quad  \exp\left(i\sum^n_{i=1}\Big(\bp^i\cdot \bx_i - p^i_0 (x^0_i) - \big(p^i_0 - s_i \mu(-p^1_0 p^1_1)\big)\beta(-p^1_0 p^1_1)t \Big)\right)\widetilde \Delta_{\om,n}(p^1,\dots,p^n)\,,
\eeq
$$
\beta(x) \doteq \Theta(x) \beta_1  + \Theta(-x) \beta_2\,,\qquad \mu(x) \doteq \Theta(x) \mu_1  + \Theta(-x) \mu_2\,, 
$$
$$
 s_i \doteq  +1\,/\, -1 \quad\text{for} \quad \Phi^\sharp(x_i) = \Phi(x_i) \,/\, \overline{\Phi}(x_i)\,.
$$
Here we assume that possible singularities of $\widetilde \Delta_{\om,n}(p^1,\dots,p^n)$ at $p^i_0=0$ and $p^i_i=0$ are at most of the type ``Heaviside distribution $\times$ continuous function'', such that $\alpha^{\beta_1,\beta_2,\mu_1,\mu_2}_t$ is well-defined\footnote{We recall that the distributional identities $\Theta(x)\Theta(-x)=0$, $\Theta(x)\Theta(x)=\Theta(x)$ can be justified by first considering $\Theta(x)\Theta(\pm x)$ as a distribution on $\bR\setminus\{0\}$ and then extending it to $\bR$ while preserving the Steinmann-scaling degree $0$, cf. e.g. \cite{BF00} for the definition of the Steinmann-scaling degree. }. Using this definition of $\alpha^{\beta_1,\beta_2,\mu_1,\mu_2}_t$ we find
\begin{proposition}\label{prop_NESSKMS}For $m>0$, $\om_N$ is a KMS state with respect to $\alpha^{\beta_1,\beta_2,\mu_1,\mu_2}_t$ at inverse temperature $\beta = 1$.
\end{proposition}
\begin{proof} Since $\om_N$ is Gaussian and charge-conjugation invariant, it is necessary and sufficient to prove the KMS condition for the two-point function $\om_N(\overline{\Phi}(x_1)\Phi(x_2))$; this follows from the form of $\Delta_{+,N}$ given in \eqref{eq_NESS2pf}.
\end{proof}
It is possible to extend the validity of this statement to $m=0$ and $\om_{N,\Psi_N}$ with a little more effort. 

We shall show in the following sections that, at least for $\mu_1=\mu_2=0$, $\om_N$ is stable with respect to perturbations which are ``small'' in a suitable sense. Presumably this can be linked to the fact that $\alpha^{\beta_1,\beta_2}_t\doteq \alpha^{\beta_1,\beta_2,0,0}_t$ commutes with $\alpha_t \doteq \alpha^0_t$.

\subsection{Reasons for non-equilibrium asymptotics and inhomogeneous models}
\label{sec_inhomogeneous}
The expectation prevalent in the literature is that a generic initial state will evolve into a so-called generalised Gibbs ensemble (GGE) \cite{Rigol}, where entropy is maximised for all conserved quantities $I_i$ of system. Formally, a GGE corresponds to a density matrix of the form
$$
\rho = \frac{1}{Z} \exp(-\sum \lambda_i I_i)\,.
$$
If only the Hamiltonian $H$ (and the particle number $N$) are conserved, one expects to have proper thermalisation. 

Formally, the statement made in Proposition \ref{prop_NESSKMS} is tantamount to saying that $\om_N$ corresponds to a density matrix
$$
\rho_N =  \frac{1}{Z_N} \exp(-H_N)
$$
with $H_N$ being the generator of $\alpha^{\beta_1,\beta_2,\mu_1,\mu_2}_t$, i.e. $\frac{d}{dt} \alpha^{\beta_1,\beta_2,\mu_1,\mu_2}_t(A)|_{t=0} = i [H_N,A]$ for suitable observables $A$. \cite{Doyon} have computed $H_N$ in terms of local fields for vanishing chemical potentials and real scalar fields $\phi$. They find
\beq\label{eq_WN}
H_N = \frac{\beta_1 + \beta_2}{2} H + \frac{\beta_1 -\beta_2}{2}(P_1 + Q)
\eeq
$$
H = \int_{\bR^{d-1}} d\bx \;\wick{T_{00}(\bx)}\qquad P_1=\int_{\bR^{d-1}} d\bx \;\wick{T_{01}(\bx)}\qquad Q=\int_{\bR^{2 d-1}} d\bx_1 d\bx_2 \;\wick{q(\bx_1,\bx_2)}\,,
$$
where
$$
T_{\mu\nu} = \partial_\mu \phi \partial_\nu \phi -\frac12 g_{\mu\nu}(\partial_\rho \phi \partial^\rho \phi+ m^2 \phi^2)
$$
and $q$ is quadratic in the fields and non-local, i.e. it is non-vanishing for non-coinciding arguments. One may check that $H$ and $P_1$ are well-defined as a derivation on local observables, i.e. elements of $\ABU$. This means that for all $A\in\ABU$, $[H,A]\in\ABU$ and  $[P_1,A]\in\ABU$. However, $H_N$ is not a derivation on $\ABU$ unless $Q=0$ which is the case only for $m=0$ and $d=2$. Notwithstanding, it is not difficult to see that $[\frac{d}{dt}\alpha_t(H_N),A] = 0$ for all $A\in\ABU$ and $\alpha_t \doteq \alpha^0_t$, cf. \eqref{eq_timeevol}. Essentially, this is the case because $\alpha_t$ commutes with $\alpha^{\beta_1,\beta_2,\mu_1,\mu_2}_t$.

$H$ is the Hamiltonian, i.e. the generator of $x^0$-translations, while $P_1$ is the momentum in $x^1$-direction, the generator of $x^1$-translations. The interpretation is that both quantities are conserved because the linear dynamics is invariant under translations of spacetime. $Q$ is a non-local conserved quantity, which, as \cite{Doyon} show, is vanishing if and only if $d=2$ and $m=0$. The authors of \cite{Doyon} argue that the presence of $Q$ is due to the infinitely many conserved charges of a linear model. For example, formally, the particle number is conserved for each mode of the field. For this reason \cite{Doyon} conjectured that an initial state consisting of two semi-infinite reservoirs brought into contact will evolve into a proper equilibrium state for a non-linear model having only four-momentum as a conserved quantity. Yet, this thermalisation would likely not occur in the rest frame of the initial reservoirs, but in a different rest frame, if the non-linear dynamics considered is still invariant under translations in the $x^1$-direction. For this reason we shall investigate two toy-models with inhomogeneous linear dynamics in order to check whether breaking $P_1$-conservation improves the thermalisation behaviour.

To this end, we consider a complex scalar field with Klein-Gordon equation
$$
K \Phi = 0\,,\qquad K = \partial^2_0 - D + m^2 + U(x^1)\,.
$$
We shall analyse two possible examples for $U(x^1)$. The first case is a ``phase shift'' where $U$ is implicitly specified by providing the normalised and complete (weak) eigenfunctions of $-\partial^2_1 + U(x^1)$ as 
\beq\label{eq_phaseshift}
Y_{p_1}(x^1) = \frac{1}{\sqrt{2\pi}} e^{i p_1 x^1}(\Theta(-x) + \Theta(x) e^{i \delta})\,,\qquad p_1 \in \bR
\eeq
with a constant $\delta \in \bR$. Note that $\overline Y_{p_1} \neq Y_{-p_1}$. In fact, this case corresponds to a possible self-adjoint extension of $-\partial^2_1$, initially defined and symmetric on smooth functions of compact support in $\bR\setminus\{0\}$, to $L^2(\bR)$. The second case we shall consider is $U(x^1) = g \delta(x^1)$ where we choose for simplicity $\delta > 0$ to avoid the discussion of the bound state for $g<0$. The normalised and complete weak eigenfunctions for the repulsive $\delta$-potential are given in Appendix \ref{sec_convergencedelta}.

For both cases, we define an initial state $\om_{G,\Psi_G}$ in complete analogy to Section \ref{sec_NESShomogeneous}, recalling the discussion of KMS states for general spatial potentials in Section \ref{sec_KMS}. The correlation functions of KMS states for the two cases considered here are well-defined distributions. This follows directly from the definition in \eqref{eq_KMS2pf}, but can also be checked explicitly using the mode expansion \eqref{eq_KMS2pfexplicit} and inserting the form of the modes; in fact, one can check that correlation functions are even tempered distributions for the two cases considered here. This is essentially the case because the modes $Y_{p_1}$ in both cases are of the form ``sum of plane waves $\times$ Heaviside distributions''. The explicit form of $\Delta_{+,G}$ is given in \eqref{eq_conv1} and using arguments provided in Appendix \ref{sec_convergence2pf} it is manifest that $\Delta_{+,G}$ is a well-defined distribution also for the inhomogeneous models considered here. Similarly, $\Psi_G$ is well-defined in the sense of a distribution, and thus $\om_{G,\Psi_G}$ is well-defined and its asymptotic limit may be investigated. Unsurprisingly $\om_{G,\Psi_G}$ is not a Hadamard state because of the discontinuity at $x^1=0$ present in $U(x^1)$. However, this is not problematic because we shall need the Hadamard property of $\om_G$ only when discussing non-linear models, where we will only consider non-linear perturbations of a homogeneous linear model.

We collect the above statements and the result of the analysis of the long-time limit of $\om_{G,\Psi_G}$, carried out in Appendix \ref{sec_convergence2pf}.
\begin{theorem}The following statements hold for the phase-shift-model with weak eigenfunctions \eqref{eq_phaseshift} and the repulsive $\delta$-potential $U(x^1)= g\delta(x^1)$, $g>0$, if one of the conditions in \eqref{eq_lowdim} is satisfied.
\begin{enumerate}
\item The bidistribution $\Delta_{+,G}$ constructed as in \eqref{eq_initialstate}, with parameters as in \eqref{eq_betamu} and $\sigma_i$ defined in Proposition \ref{prop_sigma}, is well-defined and defines a charge-conjugation-invariant, Gaussian state $\om_G$ on $\ABU$, and the coherent state $\om_{G,\Psi_G}$ constructed as in \eqref{eq_initialcondensate} is a well-defined state on $\ABU$. 
\item For any admissible choice of data $\beta_1, \beta_2, \beta_3, \mu_1, \mu_2, \mu_3, \psi, \chi_1, c_1, c_2$ the states $\om_N$ and $\om_{N,\Psi_N}$ defined in \eqref{eq_defNESS} are well-defined on $\ABU$ and independent of $\chi_1$, $\psi$. $\om_N$ is an $\alpha^0_t$-invariant, charge-conjugation-invariant, Gaussian state with two-point function \eqref{eq_2pfphaseshift} for the phase-shift-model and \eqref{eq_2pfdelta} for the repulsive $\delta$-potential, respectively. $\om_{N,\Psi_N}$ is an $\alpha^\mu_t$-invariant coherent excitation of $\om_{N}$, i.e. 
$$
\om_{N,\Psi_N} = \om_N \circ \iota_{\Psi_N}\,,\qquad \Psi_N(x^0,x^1)=  \frac{c_1+c_2}{2} e^{i\mu x^0} Y_0(x^1)\,.
$$
\item In the case of the phase-shift-model, $\om_{N,\Psi_N}$ is also independent of $\beta_3$, $\mu_3$.
\end{enumerate}
\end{theorem}

The form of $\Delta_{+,N}$ for the phase-shift model differs from the homogeneous case only by the form of the modes $Y_{p_1}(x^1)$, cf. \eqref{eq_2pfphaseshift}. The situation is different for the $\delta$-potential, where $\Delta_{+,N}$ still depends on $\beta_3$ and $\mu_3$ (but not on the initial profiles $\chi_1$ and $\psi$). The reason for this is the strong singularity of the potential. This can be understood by comparing with the perturbative treatment of non-linear models in the following sections. Anticipating some of the corresponding results, we shall find that the initial state for a quadratic perturbation potential which is smooth in space will converge to a final NESS which is independent of the intermediate / cross-correlation state $\om_3$. The potential $U(x^1)=g\delta(x^1)$ can be considered as a quadratic perturbation by expanding the exact result for $\Delta_{+,N}$ perturbatively in the coupling constant $g$. However, the singularity of this potential prevents the vanishing of $\om_3$-dependent terms. This will become manifest in the discussion in the following sections. On the other hand, the singularity in the phase-shift model is weak enough to allow for these terms to vanish asymptotically.

In the phase-shift case one can argue as in the homogeneous case that $\om_N$ is a KMS state with respect to a time-translations that shifts left- and right-movers independently, or, equivalently a KMS state at different temperatures for left- and right-movers. We have not computed the corresponding generator, though it is rather obvious that $\om_N$ is not a proper KMS state in the initial rest-frame. We don't see any indication that it is closer to a KMS state than the limit NESS in a homogeneous case. Unsurprisingly, the linear dynamics is insufficient to achieve proper thermalisation in the initial rest frame, in spite of the inhomogeneities present.

A direct thermodynamic interpretation of $\om_N$ in the case of the $\delta$-potential seems difficult on account of the rather complicated form of its two-point function displayed in \eqref{eq_2pfdelta}. In contrast to the rather weak phase-shift-inhomogeneity, the $\delta$-potential reflects initially left- or right-moving modes so strongly, that there are not thermodynamically separated any longer at large times, which is essentially what we wanted to achieve in the first place. However, we have not been able to show that the limit NESS is in any way closer to being a proper KMS state in the initial restframe than the homogeneous limit NESS. We leave this for future research.

\section{NESS for interacting models}
\label{sec_NESSinteracting}

\subsection{Preparation of the initial state}

We shall now approach the central result of this work, the construction and analysis of interacting NESS in perturbation theory. Our analysis will be restricted to massive real scalar fields $\phi$ in $d=4$ spacetime dimensions with quartic interaction $V = \int_{\bR^4} d^4x\; \phi(x)^4$. We consider only massive fields for simplicity, but without loss of generality. Massless fields can be treated in a similar fashion as massive ones on account of the thermal mass which arises from a convergent resummation of suitable tadpole graphs, as rigorously proven for KMS states in \cite{DHP}. Although the results of the latter paper do not straightforwardly apply to our situation because we are not dealing with proper KMS states, we do not see any difficulties in extending the results of \cite{DHP} to our case.

As in the free case, the first step is to construct an initial state which corresponds to bringing two semi-infinite reservoirs at different inverse temperatures $\be_1$, $\be_2$ into contact at the $x^1=0$ surface. As explained in the review of pAQFT in 
section \ref{sec_paqft}, the Hadamard property of the state of the linear theory on which perturbation theory is modelled is essential in order for all perturbative constructions to be rigorously defined from the outset. For this reason we shall not consider a sharp contact, but a smooth transition region determined by $-\epsilon < x^1< \epsilon$ like in the analysis of linear models in Section \ref{sec_NESSlinear}. The profile of the transition is prescribed by choosing two smooth functions on $\bR$ with $\chi_1 + \chi_2 = 1$ and $\chi_1(x^1) = 1$ for $x^1 < -\epsilon$, $\chi_1(x^1) = 0$ for $x^1 > \epsilon$. As argued in Section \ref{sec_NESSlinear}, we need to prescribe the  initial state in the transition region and for cross-correlations as well in order to have an positive (unitary) total initial state. For the same technical reasons as in the latter section, we choose this state to be KMS at inverse temperature $\beta_3 \ge \beta_i$, $i=1,2$, though physically it should be possible to consider any (Hadamard) state for that purpose. In order to characterise the initial state entirely, we have to prescribe all its correlation functions. This was particularly simple in linear models, where we considered Gaussian initial states which are completely defined by prescribing their two-point function, or coherent excitations thereof, which require in addition the definition of a one-point function. In the case of interacting models the situation is more complicated because states of these models are highly non-Gaussian. Thus, all correlation functions of all observables need to be prescribed individually. It is not obvious at all how to do this at once and in fact we shall not be able to give a closed formula like \eqref{eq_initialstate} in the case of linear models. Instead we shall define an initial state by prescribing Feynman rules for computing all correlation functions at any order in perturbation theory.

 In order to arrive at these rules, we consider an initial state $\omG$ constructed from the data $\chi_1$, $\psi$, $\beta_1$, $\beta_2$, $\beta_3\ge \max(\beta_1,\beta_2)$, where $\psi$ is a smooth function with the properties of $\chi_2 = 1 - \chi_1$, though evaluated on the time-axis $x^0$. We recall that $\psi$ quantifies a region in time in which the initial data of the state $\omG$ is specified. We consider the state $\omG$ for a mass term $m^2 + \delta$ and expand all correlation functions of $\omG$ perturbatively in $\delta$. The result should correspond to an initial state $\omGV$ in perturbation theory with interaction potential $V =\frac\delta 2  \int_{\bR^4} d^4x\; f(x) \phi(x)^2$ in the adiabatic limit $f\to 1$. The terms in the corresponding expansion can be translated into Feynman graphs by considering a similar expansion for a KMS state with mass $m^2 + \delta$ and comparing with the Feynman graphs of the perturbatively defined KMS state which arise from \eqref{eq_KMSFeynman}. Following this path, we get the Feynman rules for an initial state $\omGV$ for a quadratic perturbation potential $V$. Generalising these rules in a straightforward way to any interaction $V$, we arrive at the following definition of $\omGV$ which we state as a theorem.
\begin{theorem} \label{prop_initialinteracting}Let $\omG$ be the Gaussian Hadamard state constructed from the data $\chi_1$, $\psi$, $\beta_1$, $\beta_2$, $\beta_3\ge \max(\beta_1,\beta_2)$ as in \eqref{eq_initialstate} with $\sigma_i$, $i=1,2$ defined in Proposition \ref{prop_sigma}. Let $V \doteq \int_{\bR^4} d^4x\; f(x)v(\phi(x))$ with $f(x)\doteq \psi(x^0) h(\bx)$, $h\in C^\infty_0(\bR^3)$ and $v$ any polynomial. We define a state $\omGV$ on $\AV$ by defining $\omGV(\moller(A_1)\star\dots\star \moller(A_n))$ for arbitrary $A_1,\dots,A_n \in \FTloc$ as follows.
\begin{enumerate}
\item Take all Feynman graphs of $\ombV(\moller(A_1)\star\dots\star \moller(A_n))$ resulting from \eqref{eq_KMSFeynman}, where $\beta$ is considered as a dummy variable.
\item These graphs have external vertices from the $A_i$, internal vertices from the $V$ in  $\moller(A_i)$, and internal vertices from $K_V$ defined in \eqref{eq_KV}.
\item For any connected subgraph $\gamma$ a of such a graph, consider all subgraphs of $\gamma$ which are connected and contain only internal vertices of $K_V$-type, where all lines connected to these vertices are taken as part of such a subgraph, i.e. the external vertices of these subgraphs are all from $\moller(A_i)$.
\item Let $F_\gamma(x_1, \dots, x_k)$ denote the product of the amplitudes of these subgraphs, where $k$ is the total number of lines ending at all external vertices. This means that if e.g. two lines end at the external vertex $x$, then $x$ appears twice in $F_\gamma(x_1, \dots, x_k)$. Decompose $F_\gamma$ as
$$
F_\gamma = (\si_1+\si_2)^{\otimes k} F_\gamma = F_{\gamma,1}+F_{\gamma,2}+F_{\gamma,3}\,,
$$
$$
F_{\gamma,i} \doteq \si^{\otimes n}_i F_\gamma\,,\quad i=1,2\,,\qquad F_{\gamma,3}\doteq  F_\gamma-F_{\gamma,1}-F_{\gamma,2}\,.
$$
\item Replace all propagators in $F_{\gamma,i}$ by those of $\ombi$, replace all other propagators in $\gamma$ by those of $\omG$, perform the simplex integral for $\beta_i$  and sum over $i$.
\end{enumerate}
 The state $\omGV$ defined in this way is well-defined. It is a KMS state at inverse temperature $\beta_1$ ($\beta_2$) with respect to $\atV$ for all observables with support contained in $\mM_1$ ($\mM_2$) with
 $$
 \mM_{1/2} \doteq \big(\bR^4 \setminus  J\left((-\epsilon,\epsilon) \times \supp(\chi_{2/1}) \times \bR^2\right)\big)\cap ((\epsilon,\infty)\times \bR^3)\,,
 $$
 see Figure \ref{fig_M1M2}.
\end{theorem}
\begin{proof} We first need to check that the Feynman rules yield well-defined distributions. This is the case because the external vertices of $F_\gamma$ correspond to one slot of $\Delta_{+,\beta_i}$ propagators. Proposition \ref{prop_Hadamard} tells us that the application of $\sigma_i$ to these propagators is well-defined and preserves their Hadamard property. The Feynman and $\Delta_{\pm}$ propagators of $\omG$ have the correct wave front set as well because $\omG$ is a Hadamard state. Thus, all Feynman amplitudes we define are proper distributions and they can be integrated with the test functions at their vertices. 

Positivity of $\omGV$ follows from the fact that a formal power series $a$ is positive, i.e. $a = |b|^2$ with $b$ a formal power series, if and only if the lowest term in $a$ is of even order and positive. This is the case for all expectation values $\omGV(A^*\star  A)$ with $A\in \AV$ because the lowest order in $V$ of $\omGV(A^*\star  A)$ is the zeroth order and because $\omG$ is a state on $\AOon \supset \AV$.

The fact that $\omGV$ is a KMS state for observables with suitable support follows directly from the definition of $\omGV$, the causal properties of $\sigma_i$ and $\moller$, and the fact that for unconnected Feynman graphs of $\ombV$ the integration over $\beta S_n$ factorises into products of integrals over $\beta S_{k}$ and $\beta S_{n-k}$ for suitable $k$.
\end{proof}
\begin{figure}[ht]
\begin{center}
\includegraphics[width=.9\textwidth]{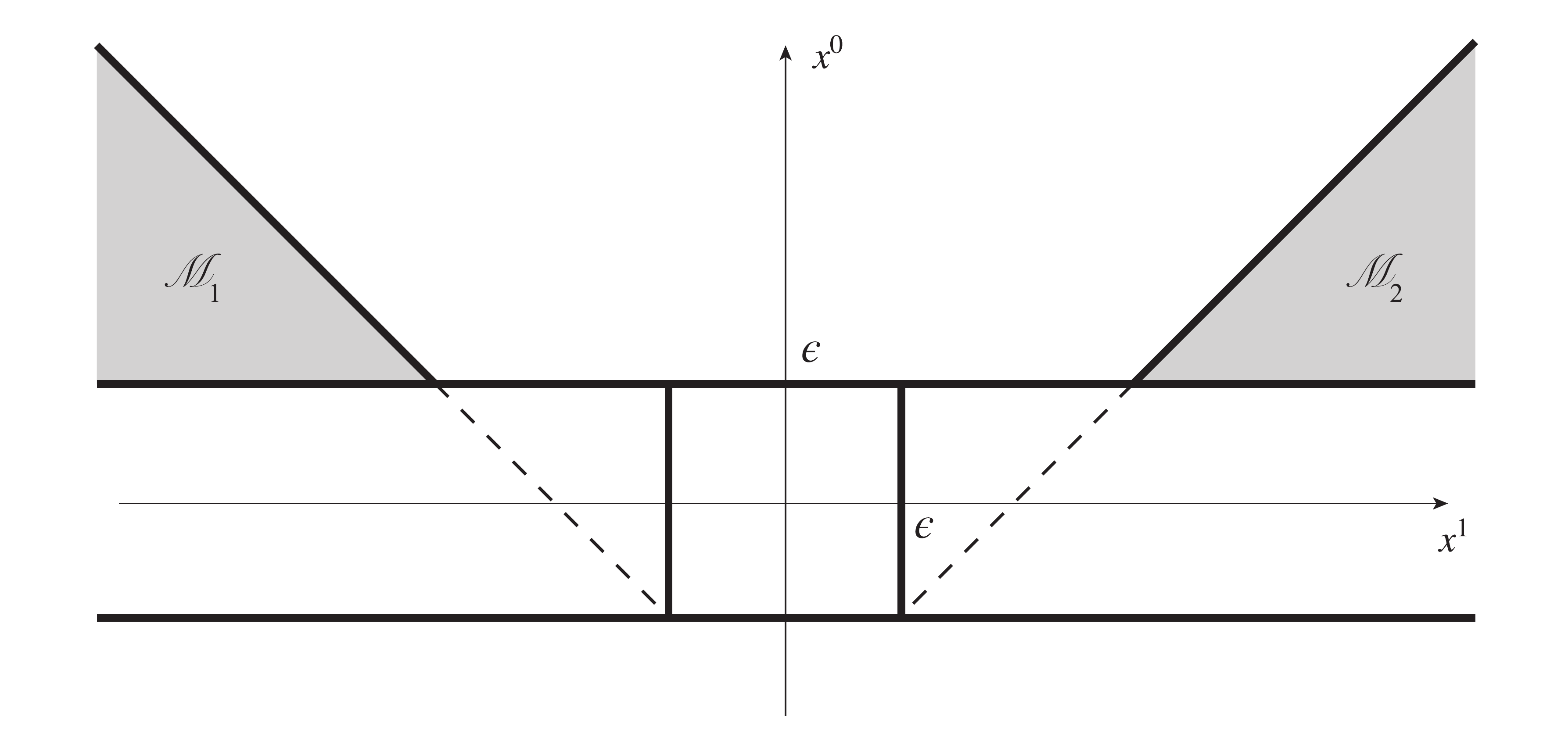}
\end{center}
\vspace{-.5cm}
\caption{The regions  $\mM_{1/2}$ in which $\omGV$ is a KMS state at inverse temperature $\beta_{1/2}$. \label{fig_M1M2}}
\end{figure}
We will provide examples of the Feynman graphs defined above in the next section when we discuss the long-time limit of $\omGV$.

The interacting KMS states $\ombV$ do not depend on the function $\psi$ which defines how the interaction is switched on. $\omGV$ is independent of $\psi$ only for observables localised in $\mM_{1/2}$ for which it is KMS. It depends on $\psi$ for all other observables, which is not surprising because it is not even $\atV$-invariant for these observables.

Note that we have not included an adiabatic limit $h\to 1$ in the definition of $\omGV$. We expect that this adiabatic limit exists at all orders, it clearly does for all observables localised in $\mM_{1/2}$. However, a proof of the existence of the adiabatic limit of $\omGV$ for all observables is considerably more involved than the same proof for $\ombV$ in \cite{FredenhagenLindner}, and we shall leave it for future research. We shall, however, prove the existence of this limit at first order in perturbation theory in the next section.

\subsection{Convergence of the initial state to a NESS}

We expect that the adiabatic limit of the initial state $\omGV$ constructed as in Theorem \ref{prop_initialinteracting}, but for a quadratic interaction $V =\frac\delta 2  \int_{\bR^4} d^4x\; f(x) \phi(x)^2$, converges, at all orders in $\delta$, to a limit NESS $\omNV$  which is equal to the NESS $\omN$ for mass $m^2 + \delta$ perturbatively expanded in $\delta$. We shall now proceed to prove the convergence of $\omGV$ to a limit NESS for a quartic interaction. Our proof will be limited to first order in perturbation theory, and we shall leave a proof to all orders for future research. However, we shall indicate the potential obstructions for an all-order proof and we shall sketch arguments for the absence of these obstructions. 

We shall prove the convergence of the initial state (at first order) both with and without the adiabatic limit. Our motivation for doing so is two-fold. One the one hand, a quartic interaction with a spatial cutoff is both non-linear and breaks translation invariance in all spatial directions including the $x^1$-direction. In view of the discussion in Section \ref{sec_inhomogeneous} about the causes for $\omG$ not converging to a proper KMS state for linear dynamics, there are reasons to hope that an inhomogeneous non-linear dynamics thermalises the system better than a homogeneous one. On the other hand, we shall be able to use the results of our computations for the interaction with spatial cutoff in order to prove stability properties of the NESS under consideration.

\begin{figure}[ht]
\begin{center}
\includegraphics[width=.9\textwidth]{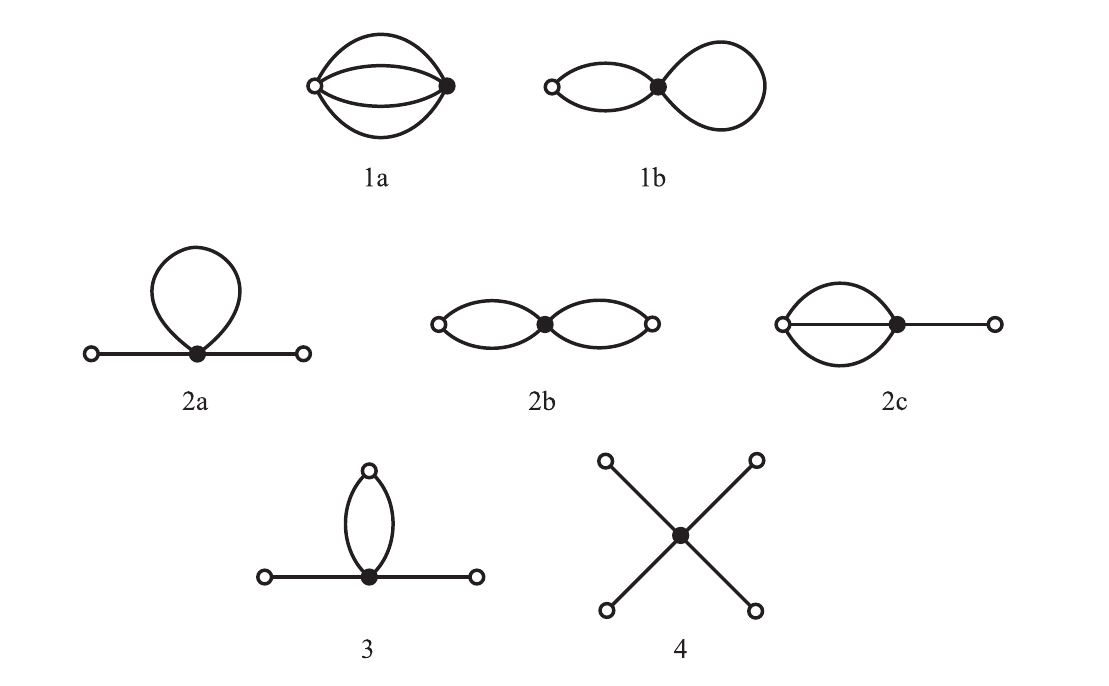}
\end{center}
\vspace{-.5cm}
\caption{The graph topologies we have to consider for convergence at first order. White (black) circles denote external (internal) vertices. \label{fig_graphs}}
\end{figure}

We will analyse the first-order convergence of the initial state graph-by-graph. In Figure \ref{fig_graphs} we enumerate all appearing graph topologies. Thereby we omit all graphs which contain lines that connect external vertices to external vertices. These lines are either $\Delta_{\pm,G}$ or $\Delta_{F,G}$ propagators. We already know that these propagators converge to $\Delta_{\pm,N}$ and $\Delta_{F,N}$ propagators, respectively. The convergence of the graphs we omit ensues from this observation and the convergence of the graphs we consider as follows. We have $W_G = \Delta_{\pm,G} - \Delta_{\pm,\infty} = \Delta_{F,G} -  \Delta_{F,\infty}$. $W_G$ is a smooth function which by Theorem \ref{prop_NESS} converges to the smooth function $W_N = \Delta_{\pm,N} - \Delta_{\pm,\infty} = \Delta_{F,N} -  \Delta_{F,\infty}$. The graphs we consider will, as we will show, converge in the sense of distributions that have wave front sets which make their pointwise multiplication with the aforementioned external $\Delta_{F/\pm,\infty}$ (as well as with any smooth function, of course), well-defined. Thus, in the convergence of graphs with such additional external lines, we do not have to deal with exchanging pointwise products of distributions with asymptotic limits, at least as far as these additional external lines are concerned.

Each graph in Figure \ref{fig_graphs} appears in various versions. Our notation for the various type of propagators and vertices is explained in Figure \ref{fig_graphs_propagators}. $\Delta_F$, $\Delta_\pm$ and $W$ are state-dependent, and we will sometimes indicate the corresponding states by adding appropriate letters to the lines.

We now state our convergence results. To this end, we recall the definition of a van Hove limit, which is the proper formulation for the spatial adiabatic limit. Following \cite{FredenhagenLindner} we define a van Hove-sequence of cutoff functions $\{h_n\}_{n\in\bN}$ by
$$
h_n \in C^\infty_0(\bR^3)\,,\qquad 0\le h_n \le 1\,,\qquad h_n(\bx)=1\,,\; |\bx|<n\,,\qquad  h_n(\bx)=0\,,\; |\bx|>n+1\,.
$$
We say that a functional $F:C^\infty_0(\bR^3) \to \bC$ converges to $L$ in the sense of van Hove, $\vlim_{h\to 1}F(h) = L$, if $\lim_{n\to\infty}F(h_n) = L$ for any  van Hove-sequence. 

\begin{theorem}\label{prop_NESSinteracting} Let $\omGV$ be the state on $\AV$ defined in Theorem \ref{prop_initialinteracting}.
\begin{itemize}
\item[(1)] For all $\chi_1$, $\psi$, $h$, $\beta_1$, $\beta_2$, $\beta_3\ge \max(\beta_1,\beta_2)$, and an arbitrary $A\in\mA_{V}$ localised in the region where $\psi = 1$ it holds that
$$
\om^{V(h)}_N(A) \doteq \lim_{t\to\infty} \omGV(\atV(A))
$$ 
is well-defined at first order in perturbation theory. The state $\om^{V(h)}_N$ defined in this way is independent of $\beta_3$, $\chi_1$ and $\psi$ and the convergence is $O(t^{-1})$.
\item[(2)] The statements in (1) also hold for 
$$
\om^{V}_N(A) \doteq \lim_{t\to\infty}  \vlim_{h\to 1}\omGV(\alpha^V_t(A))\,.
$$ 
\end{itemize}
\end{theorem}
\begin{proof}The validity of the statements is proved by computing the convergence of all Feynman graphs in Figure \ref{fig_graphs}. This convergence is demonstrated in Appendix \ref{sec_convergenceV}. The proof strategy can be subsumed as follows, where we write $V$ as $V(h)$ in all expressions considered prior to taking the adiabatic limit $h\to 1$:
\begin{enumerate}
\item Let $\widetilde\om^{V(h)}_N$ denote the state defined as $\om^{V(h)}_G$, but with the difference that we replace all $\omG$ propagators in all Feynman graphs by $\omN$ propagators, and all $(\sigma_i \otimes 1)\Delta_{\beta_j,+}$ propagators in these graphs by $\Delta^{(s)}_{+,\beta_i}$ ones. The latter bidistributions are defined and discussed in Section \ref{sec_spectral}, and their main properties are (i) $[(\sigma_i \otimes 1)\Delta_{\beta_j,+}](x,y) = \Delta^{(s)}_{+,\beta_j}(x,y)+O(1/x^0)$ for $i=1,2$ and $s=(-1)^i$ and (ii) $\Delta^{(s)}_{+,\beta_j}$ is translation-invariant in space and time.

\item With this in mind, one first argues that (at first order)
$$\om^{V(h)}_G \circ \alpha^{V(h)}_t = \widetilde\om^{V(h)}_N \circ \alpha^{V(h)}_t + O(t^{-1})$$
$$\vlim_{h\to 1}\om^{V(h)}_G \circ \alpha^{V(h)}_t = \vlim_{h\to 1}\widetilde\om^{V(h)}_N \circ \alpha^{V(h)}_t + O(t^{-1})$$
\item Finally, one shows that (at first order)
$$\lim_{t\to\infty} \widetilde\om^{V(h)}_N \circ \alpha^{V(h)}_t = \om^{V(h)}_N
\,,\qquad \lim_{t\to\infty} \vlim_{h\to 1}\widetilde\om^{V(h)}_N \circ \alpha^{V(h)}_t = \omNV\,.
$$
\end{enumerate}

The computations in Appendix \ref{sec_convergenceV} are carried out for a quartic interaction, but they can be directly generalised to any polynomial interaction.
\end{proof}

We now discuss potential obstructions for extending the proof to all orders and argue how these obstructions can be likely overcome. We do not see any difficulties for extending the proof for spatially localised $V=V(h)$ to all orders. The only non-trivial task is to phrase all-order quantities in terms of manageable closed expressions. The situation is of course different in the adiabatic limit $h\to 1$. The difficulty in step 2.~indicated in the proof above is to argue that the $O(t^{-1})$ error term is of this form even in the adiabatic limit. This is non-trivial because the temporal integrations at the interaction vertices may potentially yield secular expressions, i.e. expressions which are growing in time. It is not at all obvious how to generalise our arguments used in the discussion at first order to all orders in a closed form, though we do not see any problems in discussing each order individually. Secular terms are also potentially problematic in the third step. In fact, the computations in Appendix \ref{sec_convergenceV} show that such secular terms do exist, but combine to form translation invariant expressions if we sum all topologically equivalent graphs. This is precisely what happens at the level of graphs and at all orders for the KMS state $\ombV$. We will sketch in the next section that $\omNV$ satisfies a generalised KMS condition as well, which indicates the absence of secular effects in the analysis of convergence of $\omGV$ to $\omNV$ at all orders. In particular, we expect $\omGV$ to converge (both at large times and as a power series) to the corresponding non-perturbatively contructed NESS for a quadratic potential $V$, generalising the same statement in \cite{Drago}. However, we leave a rigorous proof of these facts for future research.

\begin{figure}[ht]
\begin{center}
\includegraphics[width=.9\textwidth]{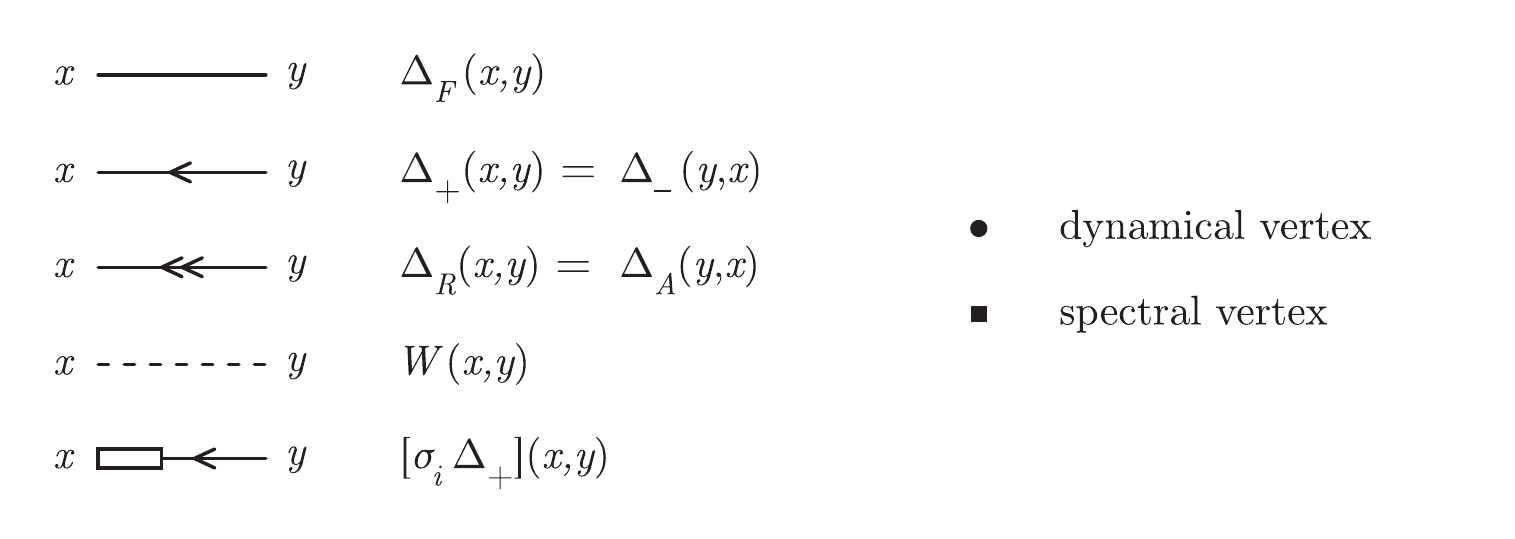}
\end{center}
\vspace{-.5cm}
\caption{The notation we use for various propagators and vertices. By ``dynamical vertex'' and ``spectral vertex'' we mean vertices originating from $V(f)=\int_{\bR^4}d^4x\; \phi(x)^4 f(x)$ and $V(\dot f)$, respectively. We recall that the spectral vertices are shifted in imaginary time $x^0 + iu$ with $u$ being integrated over, cf. \eqref{eq_KMSFeynman}. \label{fig_graphs_propagators}}
\end{figure}

\section{Properties of the NESS}
\label{sec_properties}

We have shown in Proposition \ref{prop_NESSKMS} in which rigorous sense the state $\omN$ is a KMS state at different temperatures for left- and right-moving modes for a linear model. We suspect that an analogous statement holds for the interacting state $\omNV$. We shall not provide a rigorous formulation and proof of this fact, but shall make a few heuristic remarks in that respect, before proving rigorously the stability of $\omNV$, which one would expect to hold if $\omNV$ were KMS. The form of the Feynman graphs of $\omNV$ in the adiabatic limit at first order found in Appendix \ref{sec_convergenceV} motivate the following -- formal -- direct definition of $\omNV$ which is a generalisation of \eqref{eq_araki}.
\beq\label{eq_formalNV}
\omNV(A) = \frac{\omN(A\star U^{\beta_1,\beta_2}_V(i))}{\omN(U^{\beta_1,\beta_2}_V(i))}\,,\quad U^{\beta_1,\beta_2}_V(t) = P_+ U_V(\beta_1 t) + P_- U_V(\beta_2 t)\,,
\eeq
where $P_\pm$ are projectors on right- and left-moving modes\footnote{In more detail, \eqref{eq_formalNV} has to be interpreted in the following way: compute all Feynman graphs resulting from this definition by using \eqref{eq_KMSFeynman}, writing all propagators as Fourier transforms and interpreting $P_\pm$ in a suitable sense like in Proposition \ref{prop_NESSKMS}; then cancel all terms in the resulting expressions which contain $\widetilde{\dot\psi}$ -- the Fourier transform of the derivative of the temporal switch-on function -- evaluated outside of the origin. This last step is not necessary for the proper KMS state $\ombV$, because such terms cancel identically in that case. However, if $\beta_1 \neq \beta_2$, these terms do not cancel, but constitute non-stationary contributions, which, at least at first order, are $O(t^{-1})$ if we evaluate \eqref{eq_formalNV} on $\atV(A)$, rather than $A$.}. This corresponds to an -- also entirely formal -- density matrix of the form 
$$
\rho^V_N = \frac{1}{Z^V_N} \exp\big(- \beta_1 P_+(H + V) - \beta_2 P_- (H + V)\big)\,,
$$
with $H$ being the free Hamiltonian. In particular, the non-linear interaction does not mix left- and right-moving modes, because of translation invariance. We have discussed this rigorously at first order for graphs with a ``spectral'' internal vertex in Section \ref{sec_spectral}, but the same discussion also applies for graphs with a ``dynamical'' internal vertex. This indicates that $\omNV$ is not ``closer'' to a KMS state in the original rest-frame than $\omN$, and since $\omNV$ does not appear to be more isotropic than $\omN$, the same is likely true also in a different rest-frame. If we consider the interacting NESS for an interaction localised in space, i.e. without taking the adiabatic limit, then the interaction does mix left- and right-moving modes, though the fact that at first order all graphs with a ``spectral'' internal vertex vanish at large times for such interactions seems to indicate that this mixing is insufficient to achieve proper thermalisation. The failure of the initial state to thermalise properly in a non-linear model is likely an artifact of perturbation theory. In fact, a state in perturbation theory is a proper KMS state only if it has this property already at zeroth order, which is obviously not the case for $\omNV$. One may hope that non-perturbative treatments or suitable resummations of perturbative terms can reveal a better thermalisation behaviour. Such analyses are, however, beyond the scope of this work.

\subsection{Stability}

A prominent feature of KMS states is their stability with respect to local perturbations, see e.g. \cite{JP03, BR} for a review of corresponding results. The ground work laid by \cite{FredenhagenLindner} opened the possibility to generalise corresponding results in quantum statistical mechanics to perturbative QFT, a path successfully pursued by the authors of \cite{DFP,DFP2}. In fact, in \cite{DFP} it has been shown that the interacting KMS states $\ombV$ are stable under local perturbations in the sense that, for a local $V$ and without considering the adiabatic limit,
$$
\lim_{t\to\infty}\omb \circ \atV = \ombV\,.
$$
This stability holds only for local perturbations, in fact the authors of \cite{DFP} show that the ergodic mean of the adiabatic limit of $\omb \circ \atV$ is not a KMS state, but only a NESS. We remark that local perturbations in quantum statistical mechanics are self-adjoint elements of the algebra of observables. In this sense, in pAQFT, $V$ is not a local perturbation anymore in the adiabatic limit.

The results we have obtained so far enable us to prove a similar result for the NESS constructed in this paper, though only at first order, since we have only rigorously shown the existence of the interacting NESS with this limitation. 

\begin{theorem}\label{prop_stability} Let $V \doteq \int_{\bR^4} d^4x\; f(x)\phi(x)^4$ with $f(x)\doteq \psi(x^0) h(\bx)$, $h\in C^\infty_0(\bR^3)$ and $\psi\in C^\infty(\bR)$, $\psi(x^0) = 0$ for $x^0 < -\epsilon < 0$, $\psi(x^0) = 1$ for $x^0 > \epsilon$ and let $W \doteq \int_{\bR^4} d^4x\; f_W(x)v(\phi(x))$ with $f_W(x)\doteq \psi(x^0) h_W(\bx)$, $h_W\in C^\infty_0(\bR^3)$, and $v$ an arbitrary polynomial. The following statements hold at first order in perturbation theory for an arbitrary $A\in\mA_{V+W}$ localised in the region where $\psi = 1$.
\begin{itemize}
\item[(1)]  Let $\om^{V(h)}_N$, $\om^{V(h)+W(h_w)}_N$ be states defined as in Theorem \ref{prop_NESSinteracting}. It holds
$$
\lim_{t\to\infty}\om^{V(h)}_N(a^{V+W}_t(A)) = \om^{V(h)+W(h_w)}_N(A)\,.
$$
\item[(2)] Let $\om^{V}_G$, $\om^{V+W}_G$ be states defined as in Theorem \ref{prop_initialinteracting}. It holds
\begin{align*}
 &\lim_{t\to\infty}  \vlim_{h\to 1}\om^{V}_G(\alpha^{V+W}_t(A))\\
  = \;& \lim_{t\to\infty}  \vlim_{h\to 1}\om^{V+W}_G(\alpha^{V+W}_t(A)) \\
  \doteq \;&\om^{V+W(h_w)}_N(A)\,.
\end{align*}
\end{itemize}
\end{theorem}
\begin{proof}
The statements follow from the fact that the Feynman graphs with a ``spectral'' internal vertex corresponding to $W$ vanish at large times because of the compact support of $h_W$ as shown in Section \ref{sec_spectral}.
\end{proof}

Note that the last statement can be formally interpreted as $\lim_{t\to\infty} \omNV \circ \alpha^{V+W}_t = \om^{V+W(h_w)}_N$. These stability properties of $\omN$ and $\omNV$ show that, while not being KMS with respect to the ``physical'' time-translations $\at$ and $\atV$, these states share some of the prominent features of such KMS states.

\subsection{Entropy production}

A quantity which may be used to describe NESS quantitatively is entropy production \cite{JP01,JP02,JP03, Oj0,Oj1,Oj2, Ruelle 2,Ruelle 3}. In quantum statistical mechanics, this is defined as follows. We consider a $C^*$-algebra $\mA$ with a strongly continuous one-parameter automorphism group $\at$ and a KMS-state $\om$ on $\mA$ for $\at$ at inverse temperature $\beta = 1$. We denote by $\delta_\om$ the generator of $\at$, formally $\delta_\om(A) = i[H,A]$ with $H$ the Hamiltonian corresponding to $\at$. We further consider a self-adjoint $V\in \mA$ which is in the domain of $\delta_\om$; we call such $V$ a local perturbation. $V$ induces a perturbed time-evolution $\atV$ in essentially the same way we have reviewed in Section \ref{sec_paqft}. One may formally think of $\atV$ as being generated by $H+V$, we refer to the above-mentioned references for a proper definition. We set for any state $\eta$ on $\mA$
$$
\sigma_V \doteq \delta_\om(V)\,,\qquad \mE\!\mP_V(\eta) \doteq \eta(\sigma_V)\,.
$$
$\mE\!\mP_V(\eta)$ is called the entropy production (of $\atV$) in the state $\eta$ (relative to $\om$). In \cite{JP01} it is shown that this name is well-deserved, because, for $\eta$ in the folium\footnote{We refer to eg. \cite{Haag, BR} for the definition of a folium. Essentially, the folium of $\om$ encompasses all states which are density matrices on the GNS-Hilbert space of $\om$, cf. Section \ref{sec_KMS} for a brief definition of this Hilbert space.} of $\om$, 
$$
\mE(\eta\circ \atV,\om) = \mE(\eta,\om)  + \int^t_0 ds\; \mE\!\mP_V(\eta\circ \alpha^V_s)\,,
$$
where $\mE(\eta,\om)$ denotes the relative entropy of $\eta$ with respect to $\om$, see \cite{Araki,BR}. A proper definition of this quantity requires Tomita-Takesaki modular theory and we only remark that for finite dimensional systems with $\rho_\eta$, $\rho_\om$ denoting the density matrices corresponding to $\eta$ and $\om$, the relative entropy is
$$
\mE(\eta,\om) = \text{Tr}(\rho_\eta(\log \rho_\eta - \log\rho_\rho))\,.
$$
In  \cite{JP02, JP03} it is shown under suitable assumptions that the relative entropy of $\omNV$ defined as
$$
\omNV\doteq \lim_{t\to\infty} \int^t_0 ds\; \om\circ \alpha^V_s
$$
is non-vanishing (and positive) if and only if $\omNV$ is not in the folium of $\om$. Physically, this means that $\omNV$ is thermodynamically different from $\om$, in particular, it may not be KMS. In fact in \cite{JP02, JP03} it is argued that a positive entropy production is sufficient for the existence of non-trivial thermal currents.

The difficulty in applying these ideas to perturbative QFT lies in the fact that many strong results on which quantum statistical mechanics is based are not at our disposal in pAQFT. In particular, we are always dealing with formal power series and quantities which are unbounded operators in any Hilbert space representation. Yet, in \cite{DFP2} the authors have succeeded to provide generalisations of the definitions of relative entropy and entropy production in pAQFT and for states which are KMS states translated in time with respect to a perturbed time-evolution. These definitions have been shown to be meaningful in the sense that they possess many qualitative and quantitative properties which the corresponding expressions in quantum statistical mechanics have.

In this section, we aim to compute the entropy production in the state $\omNV$ relative to $\omN$. The non-vanishing of this quantity could be interpreted as an indication that $\omNV$ is thermodynamically different from $\omN$, while its vanishing may indicate that this is not the case. We have argued at the beginning of this section \ref{sec_properties} and seen in Theorem \ref{prop_stability} that $\omNV$ does not seem to be different from $\omN$ and thus we shall not be surprised to find that the aforementioned entropy production is in fact vanishing. In contrast to the analysis in \cite{DFP2}, our initial state $\omGV$ is not a KMS state and thus we shall not be able to use the definitions of relative entropy in \cite{DFP2} in order to motivate the definition of entropy production. Instead we will define entropy production directly by the obvious generalisation of the definition in \cite{DFP2}, and, thus, we shall not be able to give strong arguments why this quantity deserves that name. That said, we recall that $\omN$ is a KMS state at inverse temperature $\beta=1$ with respect to $\alpha^{\beta_1,\beta_2}_t$ by Proposition \ref{prop_NESSKMS}, and set
$$
\mE\!\mP_V\left(\om^{V(h)}_N\right)\doteq \lim_{t\to\infty} \frac{d}{d\tau}\omGV\left.\left(\atV\left(\alpha^{\beta_1,\beta_2}_\tau (K_V)\right)\right)\right|_{\tau=0}\,,
$$
\beq\label{eq_entropyproduction}
e\!\mP_V\left(\omNV\right)\doteq \lim_{t\to\infty}\frac{d}{d\tau}\vlim_{h\to1} \frac{1}{I(h)}\omGV\left.\left(\atV\left(\alpha^{\beta_1,\beta_2}_\tau (K_V)\right)\right)\right|_{\tau=0}\,,
\eeq
$$
 I(h) \doteq \int_{
\bR^3}d\bx \;h(\bx)\,,
$$
with $K_V$ defined in \eqref{eq_KV}. In the definition of the entropy production in the adiabatic limit, we have to divide by the ``volume'' of $h$ in order to be able to get a finite quantity in the first place. Thus, $e\!\mP_V$ may interpreted as quantifying the production of entropy density, cf. \cite{DFP2}. We remark that the lowest order of $\mE\!\mP_V$ and $e\!\mP_V$ is quadratic in the interaction. For this reason, one can unambiguously claim that these quantities have a definite sign (in the sense of perturbation theory) if they are non-vanishing at lowest order. However, we shall find that they do vanish at this order, as anticipated.

\begin{proposition}\label{prop_entropyproduction} Let $\om^{V(h)}_N$ and $\omNV$ be the states defined in Theorem \ref{prop_NESSinteracting}. Their entropy production (density) $\mE\!\mP_V\left(\om^{V(h)}_N\right)$ and $e\!\mP_V\left(\omNV\right)$ defined in \eqref{eq_entropyproduction} is vanishing at lowest non-trivial order in perturbation theory for a quartic interaction.
\end{proposition}
\begin{proof} We recall that $K_V = \moller(V(\dot\psi h))$. $\alpha^{\beta_1,\beta_2}_\tau$ commutes with $\at$ and thus with $\atV$. Consequently
$$
\atV \left(\alpha^{\beta_1,\beta_2}_\tau \left(\moller(V(\dot\psi h))\right)\right) = \mR_{\alpha^{\beta_1,\beta_2}_\tau(V)}\left(\at(\alpha^{\beta_1,\beta_2}_\tau (V(\dot\psi h))\right)\,.
$$
In the case of the adiabatic limit, we can argue as \cite{Lindner,DFP2} that 
\begin{align*}
\vlim_{h\to1} \frac{1}{I(h)}\omGV\left(\atV\left(\alpha^{\beta_1,\beta_2}_\tau (K_V)\right)\right)= \vlim_{h\to1}\omGV\left(\atV\left(\alpha^{\beta_1,\beta_2}_\tau (k_V(\bx))\right)\right)=\text{const.}
\end{align*}
with
$$
k_V(\bx) \doteq \moller\left(\int_{\bR}dx^0 \;\dot\psi(x^0) \phi(x^0,\bx)^4\right)\,.
$$
We have to compute graphs of type 1a and 1b in Figure \ref{fig_graphs}, where the 1b graph has an additional line closing on itself at the external vertex. We know from the proofs of Theorem \ref{prop_NESS} and Theorem \ref{prop_NESSinteracting} that
\beq\label{eq_production1}
\Delta^n_{\pm/F,G}(x,y) = \Delta^n_{\pm/F,N}(x,y) + O(1/x^0)\,,\qquad \Delta_{+,\beta_i,s}(x,y) = \Delta^{(s)}_{+,\beta_i}+ O(1/x^0)\,,
\eeq
cf. Section \ref{sec_spectral} for the definition of the latter propagators. These statements also hold if we translate $x^0$ and $y^0$ with $\alpha^{\beta_1,\beta_2}_\tau$ and derive with respect to $\tau$, because the convergence occurs separately for the left- and right-moving parts of these distributions and the $\tau$ derivatives do not alter the necessary analyticity and decay properties on which our  convergence proof was based. Thus, we can argue as in the proof of Theorem \ref{prop_NESSinteracting} that we can replace the propagators on the left hand sides of \eqref{eq_production1} by those on the right hand side up to $O(t^{-1})$ errors which vanish in the limit $t\to\infty$. However, the propagators on the right hand sides of \eqref{eq_production1} are all invariant under shifting both arguments with $\alpha^{\beta_1,\beta_2}_\tau$ and, thus, the prevailing expressions are constant in $\tau$ and vanish if we derive with respect to $\tau$.\end{proof}

\vspace{1cm}
\noindent{\bfseries Acknowledgements}\\

We would like to thank Stefan Hollands and Nicola Pinamonti for interesting discussions.

\appendix
\section{Convergence of the NESS for linear models}
\label{sec_convergence}

\subsection{Convergence of the two-point function}
\label{sec_convergence2pf}

The first steps of the computation are the same for all linear models treated in the main body of the paper. We compute the long-time limit of the evolution of $\Delta_{+,G}$ starting from
$$
\Delta_{+,G} = \sum^2_{i,j=1} (\sigma_i \otimes \sigma_j) \Delta_{+,\beta_{ij},\mu_{ij}}\qquad \beta_{11} \doteq  \beta_1,\, \beta_{22} \doteq \beta_2,\, \beta_{12}\doteq\beta_{21} \doteq \beta_3
$$
and  $\mu_{ij}$ defined analogously, where we recall that $\sigma_i = \Delta \chi_i g$, $g = \ddot \psi + 2 \dot \psi \partial_{0}$. We consider linear models with equation of motion
$$
(\partial^2_{0} - D + m^2 + U(x^1))\Phi = 0\,,\qquad D = \sum^{d-1}_{i=1}\partial^2_i
$$
Using the normalised and complete modes $Y_{k_1}(x^1)$ of the operator $-\partial^2_1 + U(x^1)$, the two-point function of a KMS state is, cf. \eqref{eq_KMS2pfexplicit}
$$
\Delta_{+,\beta,\mu}(x,y) = \frac{1}{(2\pi)^{d-2}} \int d\bp  \frac{1}{2\omp}\sum_{s=\pm 1} e^{i s \omp (x^0 - y^0)} \overline{Y_{p_1,s}(x^1)}Y_{p_1,s}(y^1) e^{i \bpp( \bxp-\byp)} s\, b_{\beta,\mu}(s \omp)\,,
$$
$$
\bxp = (x^2,\dots,x^{d-1})\,,\qquad \bpp = (p_2,\dots,p_{d-1})\,,\qquad \omp = \sqrt{|\bp|^2+m^2}\,,
$$
$$
b_{\beta,\mu}(\om) \doteq \frac{1}{\exp(\beta (\om-\mu)) - 1}\,,\qquad Y_{k_1,+} \doteq Y_{k_1}\,,\qquad Y_{k_1,-} \doteq \overline{Y_{k_1}}\,.
$$
We shall also need the mode expansion of the causal propagator $\Delta  = \Delta_R - \Delta_A$
$$
\Delta(x,y) = \frac{- i}{(2\pi)^{d-2}} \int d\bp  \frac{1}{2\omp}\sum_{s=\pm 1} s\, e^{i s \omp (x^0 - y^0)} \overline{Y_{p_1,s}(x^1)}Y_{p_1,s}(y^1) e^{i \bpp( \bxp-\byp)} \,.
$$
Thus we have 
$$
\Delta_{+,G}(x,y) = \frac{-1}{(2\pi)^{3(d-2)}} \int dv \,dz \, d\bp \, d\bk\, d\bq \sum_{s_1,s_2,s_3,s_4,s_5 = \pm 1} \frac{s_1\,s_2\,s_3 \,b_{s_4,s_5}(s_3\omq)}{8 \omp \omk \omq} \qquad \times
$$
$$
\times \qquad e^{i s_1 \omp(x^0 - v^0) + i s_2 \omk(y^0 - z^0)}g(v^0)g(z^0)  e^{i s_3 \omq(v^0- z^0)}\chi_{s_4}(v^1)\chi_{s_5}(z^1)\qquad \times
$$
$$
\times \qquad e^{i\bpp(\bxp-\bvp)+i\bkp(\byp-\bzp)+i\bqp(\bvp - \bzp)} \overline{Y_{p_1,s_1}(x^1)}Y_{p_1,s_1}(v^1)\overline{Y_{k_1,s_2}(y^1)}Y_{k_1,s_2}(z^1)\overline{Y_{q_1,s_3}(v^1)}Y_{q_1,s_3}(z^1)
$$
where we set $\chi_{\pm1} \doteq \chi_{1/2}$ and $b_{+1,+1} \doteq b_{\beta_1,\mu_1}$, $b_{-1,-1} \doteq b_{\beta_2,\mu_2}$, $b_{\mp1,\pm1} \doteq b_{\beta_3,\mu_3}$. We can easily perform the integrals w.r.t. $v^0$ and $z^0$ explicitly. We can also perform the integrals w.r.t. $\bvp$ and $\bzp$, which give $\delta$ distributions for the parallel momenta $\bkp$, $\bqp$ and then perform the integrals w.r.t. to these momenta. Further we integrate $\Delta_G(x,y)$ with spatial test functions $f$ and $g$ in order to improve the large momentum behaviour in the complex domain, which will allow us to use the residue theorem later. This step is legitimate as we wish to prove the convergence of $\Delta_{+,G}$ in the sense of distributions after all. Finally, we observe that the cutoff functions $\chi_i$, $i=1,2$ are convolutions of the Heaviside step function $\Theta$ with $\nu \doteq \partial_{1} \chi_2$ of compact support and unit integral. 
\beq\label{eq_chiform}
\chi_1 = (1-\Theta)* \nu\,,\qquad \chi_2 = \Theta* \nu\qquad \Rightarrow\qquad \chi_s(x)=\lim_{\epsilon\downarrow 0}\frac{-s i}{\sqrt{2\pi}}\int d\lambda \frac{\widetilde{\nu}(\lambda)}{\lambda-  is\epsilon}e^{-i\lambda x}
\eeq
where $\widetilde{\nu}(\lambda)$ is the Fourier transform of $\nu$. After these considerations we find
$$
\Delta_{+,G}(x^0+t,y^0+t,f,g) = \frac{1}{(2\pi)^{3(d-2)+1}} \int dv^1 \,dz^1 \, d\bp \, dk_1\, dq_1 \, d\lambda_1\,d\lambda_2\qquad \times 
$$
$$
\times \qquad \sum_{s1,s2,s3,s4,s5 = \pm 1} \frac{s_1\,s_2\,s_3 \,s_4 \,s_5 \,b_{s4,s5}(s_3\omq)\,\tilde{\nu}(\lambda_1)\tilde{\nu}(\lambda_2)}{8\, \omp \,\omega(k_1,\bpp)\, \omega(q_1,\bpp)\,(\lambda_1 - i s_4 \epsilon)\,(\lambda_2 - i s_5 \epsilon)} \qquad \times
$$
\begin{equation}\label{eq_conv1}
\times \qquad e^{i s_1 \omp(x^0 +t) + i s_2 \omega(k_1,\bpp)(y^0 +t)}H(s_1,s_2,s_3,\bp,k_1,q_1)\qquad \times
\end{equation}
$$
\times \qquad e^{-i\lambda_1 v^1} e^{-i\lambda_2 z^1}  Y_{p_1,s_1}(v^1)Y_{k_1,s_2}(z^1)\overline{Y_{q_1,s_3}(v^1)}Y_{q_1,s_3}(z^1) \tilde f_{-s_1}(\bp) \tilde g_{-s_2}(k^1,-\bpp)
$$
where 
$$
\omega(q,\bpp) \doteq \sqrt{q^2 + |\bpp|^2 + m^2}\,,\qquad \tilde f_{s}(\bp) \doteq  \frac{1}{\sqrt{2\pi}^{d-2}} \int d\bx \, e^{i\bpp \bxp} Y_{p_1,s}(x^1)\,,
$$
\begin{align*}
H(s_1,s_2,s_3,\bp,k_1,q_1) & \doteq 2\pi\big(s_1 \omp + s_3\omega(q_1,\bpp)\big)\big(s_2\omega(k_1,\bpp) - s_3\omega(q_1,\bpp)\big)\quad\times\\
&\quad\times\quad \widetilde{\dot\psi}\big(s_1 \omp - s_3\omega(q_1,\bpp)\big)\widetilde{\dot\psi}\big(s_2\omega(k_1,\bpp) + s_3\omega(q_1,\bpp)\big)
\end{align*}
In the homogeneous case $U(x^1)=0$ the Fourier transforms $\tilde f(\bp)$, $\tilde g(\bp)$ are by the Paley-Wiener theorem entire analytic functions which satisfy the bound
$$
|\tilde f(p_1,\bpp)| < C_N (1+|p_1|)^{-N} e^{A|\Im{p_1}|}
$$
where $N$ is arbitrary, $C_N$ is a constant depending on $N$ and $f$ and $A$ is a constant indicating the size of the compact support of $f$ in space. In the inhomogeneous cases we consider modes $Y_{p_1}$ that will be linear combinations of expressions of the form $c_{p_1} \Theta(\pm x^1) e^{i\pm p_1 x^1}$. It is not difficult to see that Fourier transforms of compactly supported smooth functions which are half-sided in the first spatial coordinate are again entire analytic but satisfy in general only the weaker bounds
$$
|\tilde f(p_1,\bp)| < C (1+|p_1|)^{-1} e^{A|\Im{p_1}|}
$$
because in general one can perform only a single partial integration in order to get the polynomial decay due to the boundary term at $x^1 = 0$. The analyticity and decay are of course modulated by the properties of the coefficients $c_{p_1}$ appearing in the modes.

We then perform the integrals with respect to $v^1$ and $z^1$. At this stage the specific form of the modes $Y_{p_1}$ becomes relevant and we discuss each case separately.

\subsubsection{The homogeneous case}

We do not discuss the case $U(x^1)=0$ separately as it is a special case of the following two models.

\subsubsection{Phase shift}

The normalised and complete eigenfunctions of $-\partial^2_1+U(x^1)$ in this case are 
$$
Y_{p_1}(x^1) = \frac{1}{\sqrt{2\pi}} e^{i p_1 x^1}(\Theta(-x) + \Theta(x) e^{i \delta})
$$
with a constant $\delta \in \bR$. Note that $\overline Y_{p_1} \neq Y_{-p_1}$. We sketch the remaining steps in computing the integral without giving all rather long intermediary results.

\begin{itemize}
\item The $v^1$ and $z^1$ integrals can be performed easily using
$$
\int dx\, e^{i\lambda x} \Theta(\pm x) = \frac{\pm i}{\lambda \pm i \epsilon}\,.
$$
As a result, we get expressions of the form $(s_1 p_1 - s_3 q_1 - \lambda_1 \mp i \epsilon)^{-1}$ and $(s_2 k_1 + s_3 q_1 - \lambda_2 \mp i \epsilon)^{-1}$ respectively, each modulated with the constant phase $\delta$ where applicable.

\item We now perform the $k_1$ integral using the residue theorem. The $k_1$ integrand has the following properties 
\begin{enumerate}
\item It has poles at $k_1=-s_2 s_3 q_1+ s_2 \lambda_1 \pm i s_2\epsilon$.
\item Owing to the appearance of $\om(k_1,\bpp)$, it has two branch cuts located on the imaginary axis for $\Im k_1 \ge \sqrt{|\bpp|^2 + m^2}$ and $\Im k_1 \le -\sqrt{|\bpp|^2 + m^2}$.
\item It is analytic everywhere else and, for sufficiently large $t$, rapidly decreasing for large $|k_1|$ in either the upper right and lower left quadrant of the complex plane if $s_2 = 1$, or the upper left and lower right quadrant of the complex plane if $s_2 = -1$.
\item It is bounded by $|k_1|^{-2}$ for large $|k_1|$ because $H/\om(k_1,\bpp) = O(|k_1|^0)$ for large $|k_1|$ and the pole and $\tilde g$ each contribute a $|k_1|^{-1}$ decay.
\end{enumerate}

\begin{figure}[ht]
\includegraphics[width=\textwidth]{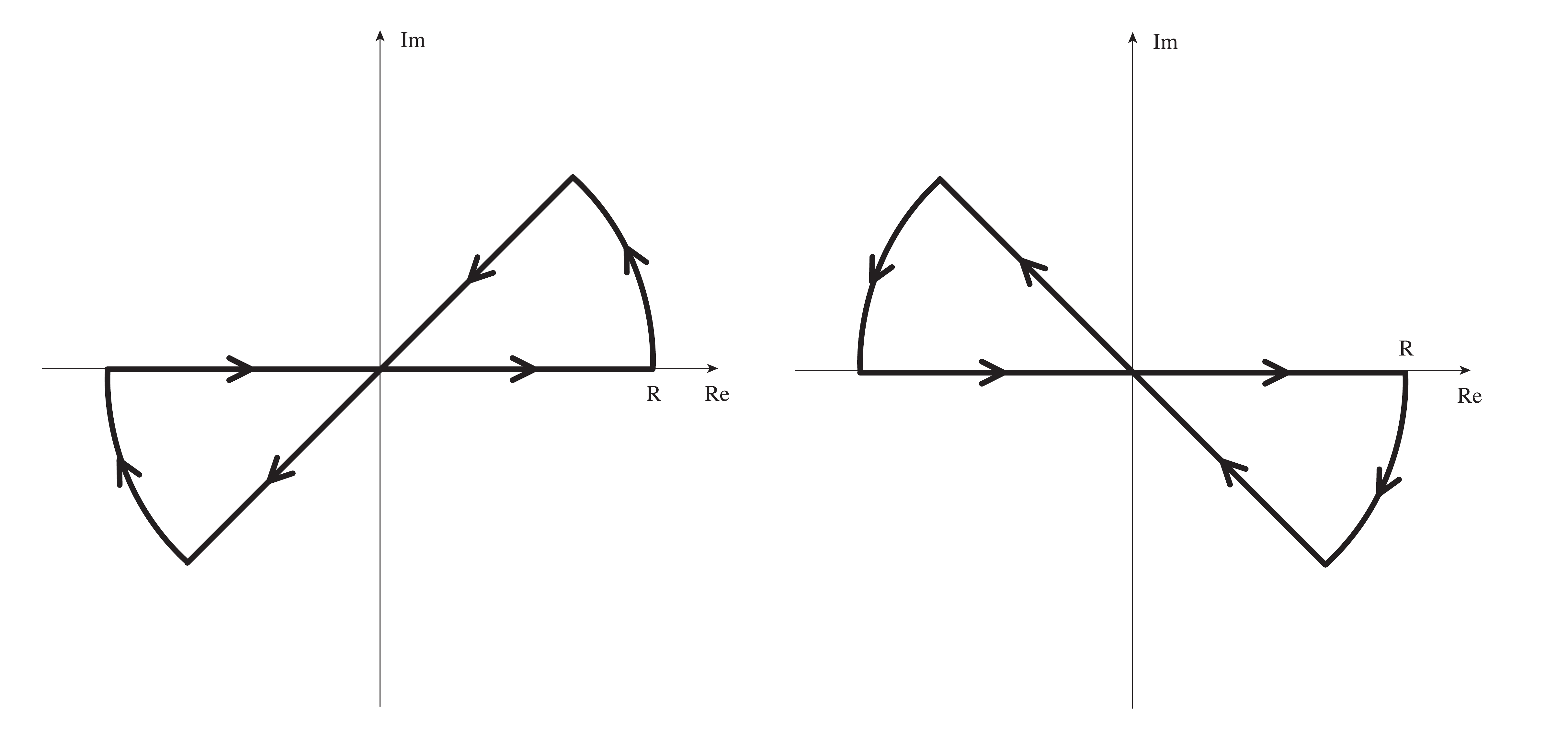}
\caption{``Butterfly'' integration contours used in the limit $R\to \infty$ in various integrals evaluated at asymptotic times. \label{fig_butterfly}}
\end{figure}

Thus we may use left / right ``butterfly'' integration contour depicted in Figure \ref{fig_butterfly} for $s_2 = +1$ / $-1$ in the limit $R\to \infty$. The circular parts of the contour give a vanishing contribution to the contour integral in the limit of infinite $R$ because of the $|k_1|^{-2}$ decay of the integrand. The integrand $I_d(k_1)$ on the diagonal part of the contour can be estimated as
\beq\label{eq_conv_estimate_phase}
|I_d(k_1)|\le C\exp\left(-\left(\frac{y^0+t - A_g - A_\nu}{\sqrt{2}}\right)|k_1|\right)
\eeq
where $A_g$ and $A_\nu$ are the Paley-Wiener constants for $g$ and $\nu$ respectively. Thus, the corresponding integral on the diagonal part of the contour behaves like $O(t^{-1})$ for large $t$ and up to this error term the $k_1$ integral is vanishing or the residue at a pole if the the real and imaginary part of the pole have the correct sign; here the winding number is $+1$ / $-1$ for a pole in the upper right quadrant / lower left quadrant. Thus an expression of the form $(s_2 k_1 + s_3 q_1 - \lambda_2 \mp i s_2 \epsilon)^{-1}F(k_1)$ gives $\pm  2 \pi i s_2 \Theta(\mp( s_2 s_3 q_1 - s_2 \lambda_2))F(-(s_2 s_3 q_1 - s_2 \lambda_2))$ after integration.

\item Next we perform the $\lambda_2$ integral. The integrand inherits the growth/decay properties of the $k_1$ integrand and has a pole at $0 + i s_5 \epsilon$ and branch cuts at the imaginary axis located at $\Re \lambda_2 = s_3 q_1$. Thus we may use the same butterfly contours as above but centered at $\lambda_2 = s_3 q_1 + 0 i$ rather than at the origin. Note that the discontinuities of the integrand at this point induced by the Heaviside functions (or by integrable inverse powers of $k_1$ induced by $m=0$ or $m-\mu_i=0$) are inessential. One may deform the contour away from this point and consider the limit of vanishing deformation. As a result, the integral of an expression of the form $\pm  2 \pi i s_2 \Theta(\mp( s_2 s_3 q_1 - s_2 \lambda_2))F(-(s_2 s_3 q_1 - s_2 \lambda_2))(\lambda_2 - i s_5 \epsilon)^{-1}G(\lambda_2)$ gives -- up to an $O(t^{-1})$ error term from the diagonal part of the contour -- $(2 \pi i)^2 s_2 \delta_{s_5,\pm 1}\Theta(-s_2 s_3 s_5 q_1)F(-s_2 s_3 q_1)G(0)$ after integration.

\item Afterwards we perform the $q_1$ integral. This integrand again inherits the growth/decay properties of the $k_1$ integrand which are modulated by the Bose-Einstein factor $b_{ij}$ which induces additional singularities on the imaginary axis and at most improves the decay for large $|q_1|$ otherwise. The integrand has pole factors $(s_1 p_1 - s_3 q_1 - \lambda_1 \mp i \epsilon)^{-1}$ which still survived from the $z^1$ integration. We use again the butterfly-contours centered at the origin and find that an expression of the form $(s_1 p_1 - s_3 q_1 - \lambda_1 \mp i \epsilon)^{-1}F(q_1)$ gives $\pm  2 \pi i s_2 \Theta(\mp( s_1 s_2  p_1 - s_2 \lambda_1))F(s_1 
s_3 p_1 - s_3 \lambda_1)$ after integration.

\item Finally we perform the $\lambda_1$ integral. We use the aforementioned butterfly contours once more, this time centered at $\Re\lambda_1 = s_1 s_2 p_1$. We find that an expression of the form $\pm  2 \pi i s_2 \Theta(\mp( s_1 s_2  p_1 - s_2 \lambda_1))F(s_1 s_2 s_3 p_1 - s_3 \lambda_1)(\lambda_1 - i s_4 \epsilon)^{-1}G(\lambda_1)$ gives\\$(2 \pi i)^2 s_2 \delta_{s_4,\pm 1}\Theta(-s_1 s_2 s_4 p_1)F(s_1 s_2 s_3 p_1)G(0)$ up to an $O(t^{-1})$ error term after integration. 

At this point, $k_1$ and $q_1$ in $\om(k_1,\bpp)$ and $\om(q_1,\bpp)$ have been replaced by $+p_1$ or $-p_1$ such that both expressions are equal to $\omp$. Consequently $H(s_1,s_2,s_3,\bp,k_1,q_1) \to  \omp^2(s_2 - s_1)(s_1 + s_3)$. This enforces $s_1 = -s_2 = s_3$ which cancels the phase factors $e^{i(s_2+s_3)\delta}$ and $e^{i(s_1-s_3)\delta}$, cancels the non-time-translation invariant terms and implies that the remaining Heaviside functions are $\Theta(s_1 s_4 p_1) \Theta(s_1 s_5 p_1)$ which in turn enforces $s_4 = s_5$ and implies that the $\beta_3$-contribution vanishes.  Altogether we find, also using $\tilde \nu(0) = \sqrt{2\pi}^{-1}\int d x^1 \nu(x^1) = \sqrt{2\pi}^{-1}$,
$$
\Delta_{+,G}(x^0+t,y^0+t,f,g) =$$

$$=  \int d\bp \sum_{s1,s4 = \pm 1} \frac{s_1 b_{s4,s4}(s_1\omp)}{2\, \omp} \Theta(s_1 s_4 p_1) e^{i s_1 \omp(x^0-y^0)} \tilde f_{-s_1}(\bp) \tilde g_{s_1}(-\bp) + O(t^{-1})\,,
$$
and thus
$$
\lim_{t\to\infty}\Delta_{+,G}(x^0+t,y^0+t,\bx,\by)  = \Delta_{+,N}(x,y) \doteq 
$$
\beq\label{eq_2pfphaseshift}
\doteq \frac{1}{(2\pi)^{d-2}} \int d\bp \sum_{s=\pm 1} e^{i s \omp (x^0 - y^0)} \overline{Y_{p_1,s}(x^1)}Y_{p_1,s}(y^1) e^{i \bpp( \bxp-\byp)} s\, b_{\beta(sp_1),\mu(sp_1)}(s \omp)\,,
\eeq
$$
\beta(p_1) \doteq \Theta(p_1) \beta_1  + \Theta(-p_1) \beta_2\,,\qquad \mu(p_1) \doteq \Theta(p_1) \mu_1  + \Theta(-p_1) \mu_2\,.
$$

\end{itemize}

\subsubsection{Delta potential}
\label{sec_convergencedelta}
The normalised and complete eigenfunctions for the potential $U(x^1) = g \delta (x^1)$, $g>0$ are, see e.g. \cite{Patil}
$$
Y_{p_1}(x^1) = a_p e^{i p_1 x^1} + b_p e^{-i p_1 x^1} + d_p\epsilon(x^1) \sin(p_1 x^1)\,,
$$
$$
 a_p \doteq \frac{1}{2 \sqrt{2\pi}} \left(1 + \frac{1}{\sqrt{1+c^2_p}}\right)\,,\qquad b_p \doteq \frac{1}{2 \sqrt{2\pi}}\left(-1 + \frac{1}{\sqrt{1+c^2_p}}\right)\,,\quad d_p\doteq \frac{c_p}{\sqrt{2\pi}\sqrt{1+c^2_p}}\,,
$$
$$ 
c_p\doteq \frac{ g}{2 p_1}\,,\qquad \epsilon(x^1)\doteq \Theta(x^1) - \Theta(-x^1)\,,
$$

The computation strategy is analogous to the case of the simple phase shift and we may perform the integrations in the same order and using the same contours, with the only difference being the fact that some terms contain $\delta$ distributions rather than poles and thus their integration is simplified. The coefficients $a_p$, $b_p$ and $d_p$ are analytic apart from two branch cuts on the imaginary axis and are bounded by constants for large $|p_1|$, thus they do not interfere with the requirements for using the residue theorem as we did in the previous case. An identity which is used several times in the computation is
$$
\int dz \,e^{i a z}\epsilon(z) \sin(b z) = \frac 12 \left(\frac{1}{a+b+i\epsilon}+\frac{1}{a+b-i\epsilon}-\frac{1}{a-b+i\epsilon}-\frac{1}{a-b-i\epsilon}\right)\,.
$$
Again we omit the various lengthy intermediary results and only state the final result which is still reasonably ugly
\beq\label{eq_2pfdelta}
\lim_{t\to\infty}\Delta_{+,G}(x^0+t,y^0+t,\bx,\by)  = \Delta_{+,N}(x,y) \doteq 
\eeq
$$
= \frac{1}{(2\pi)^{d-2}}\int d\bp \sum_{s_1,s_2,s_3=\pm1} \frac{s_1 b_{s_2,s_3}(s_1 \omp)}{2\, \omp}  e^{i s_1 \omp(x^0-y^0)} e^{i \bpp(\bxp-\byp)} \qquad \times 
$$
$$
\times \qquad \left(A(p_1,s_1,s_2,s_3)\overline {Y_{p_1,s}(x^1)}Y_{p_1,s}(y^1)+B(p_1,s_1,s_2,s_3)Y_{p_1,s}(x^1)Y_{p_1,s}(y^1)\right)\,,
$$
$$
A(p_1,s_1,s_2,s_3)\doteq  \left(a_p^2 \Theta(s_1s_3 p_1)+b_p^2 \Theta(-s_1s_3 p_1)+\frac14 d^2_p + \frac{i s_1 b_p d_p}{2}\epsilon(p_1)\right)\quad \times $$
$$
\times \quad \left(a_p^2 \Theta(s_1s_2 p_1)+b_p^2 \Theta(-s_1s_2 p_1)+\frac14 d^2_p + \frac{i s_1 b_p d_p}{2}\epsilon(p_1)\right) -\left(\frac{a_p d_p}{2}\right)^2\,,
$$
$$
B(p_1,s_1,s_2,s_3) \doteq \frac{s_1 i a_p d_p}{2}\epsilon(p_1)\quad \times
$$
$$
\times \quad \left(a^2_p(\Theta(s_1s_2 p_1)-\Theta(-s_1s_3 p_1))+b^2_p(\Theta(-s_1s_2 p_1)-\Theta(s_1s_3 p_1))+\frac{s_1 i b_p d_p}{2}\epsilon(p_1)\right)
$$
which manifestly satisfies $\overline{\Delta_{+,N}(x,y)} = \Delta_{+,N}(y,x)$ as it should and also depends on $\beta_3$, but not on the shape of the transition function $\chi_1$. 

\subsection{Convergence of the condensate}
\label{sec_convergencecondensate}

We want to show that $e^{-i\mu t} \Psi_G(x^0+t,x^1) = \Psi_N(x^0,x^1)+  O(t^{-1})$ for large $t$. This follows from $e^{-i\mu t} [\sigma_i f](x^0+t,x^1) = \frac12 f(x^0,x^1)+  O(t^{-1})$ for $f(x^0,x^1)\doteq e^{i\mu x^0}Y_{0}(x^1)$. We show the latter identity for simplicity for the homogeneous case. The inhomogeneous cases can be treated with a little more effort using the computational techniques employed in the proof of the convergence of the two-point function. In the homogeneous case we choose $\sqrt{2\pi}Y_{0}(x^1) = 1$ instead of $Y_{0}(x^1)$. In analogy to the discussion of the convergence of the two-point function, we write $\Delta$ as a Fourier transform with respect to $\bp$ and $\chi_i$ as a Fourier transform with respect to $\lambda$. We can directly perform the temporal and spatial integrals in the convolution of $\Delta$ with $g \chi_i f$. As a result we get $\delta$-distributions $\delta(p^1-\lambda)$ and $\delta(\bpp)$. After performing the $\bp$ integral, we are left with
$$
e^{-i\mu t} [\sigma_i f](x^0) = e^{i\mu x^0}\sum_{s_1 = \pm 1} -i s_1 s_2\int d\lambda \frac{\widetilde \nu(\lambda)}{\lambda - i s_2 \epsilon}\frac{(s_1 \om_\lambda+\mu)\widetilde{\dot\psi}(s_1 \om_\lambda-\mu)}{2  \om_\lambda}e^{i(s_1  \om_\lambda - \mu)(x^0 + t) }e^{-i \lambda x^1}
$$
with $ \om_\lambda \doteq \sqrt{\lambda^2 + m^2}$ and $s_2 = \pm 1 $ for $\chi_{1/2}$. We can not use a butterfly-contour integral because the pole is at the origin. Instead, for $m>0$, we set $\rho =\lambda t$ and find, using that $\tilde\nu$, $\widetilde{\dot\psi}$ are Schwartz with $\tilde\nu(0)=\widetilde{\dot\psi}(0)=\sqrt{2\pi}^{-1}$, and that $\mu = \pm m$,
$$
e^{-i\mu t} [\sigma_i f](x^0) = e^{i\mu x^0} \frac{-i s_2}{2\pi}\int d\rho \frac{e^{i\frac{\mu}{m} |\rho|} }{\rho - i s_2 \epsilon}+ O(t^{-1})
$$
The remaining integral can be computed, in fact
$$
\int d\rho \frac{e^{i a  |\rho|} }{\rho - i s_2 \epsilon} = s_2 \pi i \qquad \forall a \in \bR\,,
$$
which implies $e^{-i\mu t} [\sigma_i f](x^0) = \frac12 e^{i\mu x^0} + O(t^{-1})$ for $m>0$. For $m=0$ and thus $\mu=0$, both terms in the sum over $s_1$ contribute, but $(s_1 \om_\lambda + \mu)/(2\om_\lambda) = s_1/2$ and thus we get the same result as for $m>0$.

\subsection{Time evolution at finite times}
\label{sec_convergencefinite}
Here we explain the main steps for numerical evaluation of the time evolution of $\om_G$ at finite times. We use the notation from the previous sections of the appendix. We have
$$
[(\sigma_i \otimes \sigma_i) W_{\beta_i}](x_1,x_2) = 
$$
$$
= \frac{1}{(2\pi)^9} \int d^4y \, d^4 z \, d\bp \, d\bk \, d\bq \, \frac{\sin(\omp(x^0_1 - y^0)) \sin(\omk(x^0_2 - z^0)) g(y^0)g(z^0) \cos(\omq(y^0 - z^0))}{ \omp \omk \omq (\exp(\omq \beta_i) - 1) }\quad \times
$$
$$
\times \quad \chi_i(y_1)\chi_i(z_1) \exp i \left(\bp(\bx_1 - \by)+\bk(\bz - \bx_2) + \bq(\by - \bz)\right)
$$
$$
= \frac{1}{(2\pi)^5}\int dy^0\,dy^1 \, dz^0\,dz^1 \, \, d\bp \, dk^1 \, dq_1 \, \quad \times
$$
$$\times \quad \frac{\sin(\omp(x^0_1 - y^0)) \sin(\omega(k^1,\bpp)(x^0_2 - z^0)) g(y^0)g(z^0) \cos(\omega(q_1,\bpp)(y^0 - z^0))}{2 \omp \omega(k^1,\bpp) \omega(q_1,\bpp) (\exp(\omega(q_1,\bpp) \beta_i) - 1) }\quad \times
$$
$$
\times \quad \chi_i(y_1)\chi_i(z_1) \exp i \left(\bpp(\bxpI - \bxpII)+(p^1 x^1_1 - k^1 x^1_2) + (q_1 - p^1)y^1 + (k^1 - q_1)z^1\right)
$$
We choose $\chi_i$ to be convolutions of step functions with normalised Gaussians
$$
\chi_i(x) \doteq \int \Theta_\mp (x-y) G_a(y) dy\,,\qquad G_a(y) \doteq \frac{\exp(-\frac{y^2}{2 a^2})}{\sqrt{2\pi}a}\,,
$$
where the upper sign is for $i=1$ and the lower one is for $i=2$. Using the convolution theorem for the Fourier transform we can ``pull out'' the transform of the Gaussians as a multiplicative factor.
$$
[(\sigma_i \otimes \sigma_i) W_{\beta_i}](x_1,x_2) = \frac{1}{(2\pi)^4}\int dy^0\,dy^1 \, dz^0\,dz^1 \, \, d\bp \, dk^1 \, dq_1 \, \quad \times
$$
$$\times \quad \frac{\sin(\omp(x^0_1 - y^0)) \sin(\omega(k^1,\bpp)(x^0_2 - z^0)) g(y^0)g(z^0) \cos(\omega(q_1,\bpp)(y^0 - z^0))}{ \omp \omega(k^1,\bpp) \omega(q_1,\bpp) (\exp(\omega(q_1,\bpp) \beta_i) - 1) }\quad \times
$$
$$
\times \quad \Theta_\mp(y_1) \Theta_\mp(z_1) \exp i \left(\bpp(\bxpI - \bxpII)+(p^1 x^1_1 - k^1 x^1_2) + (q_1 - p^1)y^1 + (k^1 - q_1)z^1\right) \quad \times
$$
$$
\times \quad G_{\frac 1 a}(q_1 - p^1)G_{\frac 1 a}(k^1 - q_1)\frac{1}{a^2}
$$
Finally we perform the necessary derivatives and set $x = x_1 = x_2$
$$
[(\partial_0 \otimes \partial_0)(\sigma_i \otimes \sigma_i) W_{\beta_i}](x,x) = \frac{1}{(2\pi)^4}\int dy^0\,dy^1 \, dz^0\,dz^1 \, \, d\bp \, dk^1 \, dq_1 \, \quad \times
$$
$$\times \quad \frac{\cos(\omp(x^0 - y^0)) \cos(\omega(k^1,\bpp)(x^0 - z^0)) g(y^0)g(z^0) \cos(\omega(q_1,\bpp)(y^0 - z^0))}{  \omega(q_1,\bpp) (\exp(\omega(q_1,\bpp) \beta_i) - 1) }\quad \times
$$
$$
\times \quad \Theta_\mp(y_1) \Theta_\mp(z_1) \exp i \left((p^1 - k^1) x^1 + (q_1 - p^1)y^1 + (k^1 - q_1)z^1\right)  G_{\frac 1 a}(q_1 - p^1)G_{\frac 1 a}(k^1 - q_1)\frac{1}{a^2}\,,
$$
$$
[(\partial_j \otimes \partial^j + m^2)(\sigma_i \otimes \sigma_i) W_{\beta_i}](x,x) = \frac{1}{(2\pi)^4}\int dy^0\,dy^1 \, dz^0\,dz^1 \, \, d\bp \, dk^1 \, dq_1 \, \quad \times
$$
$$\times \quad \frac{\cos(\omp(x^0 - y^0)) \cos(\omega(k^1,\bpp)(x^0 - z^0)) g(y^0)g(z^0) \cos(\omega(q_1,\bpp)(y^0 - z^0))(\bpp^2 + k^1 p^1 + m^2)}{   \omp  \omega(k^1,\bpp) \omega(q_1,\bpp) (\exp(\omega(q_1,\bpp) \beta_i) - 1) }\quad \times
$$
$$
\times \quad \Theta_\mp(y_1) \Theta_\mp(z_1) \exp i \left((p^1 - k^1) x^1 + (q_1 - p^1)y^1 + (k^1 - q_1)z^1\right)  G_{\frac 1 a}(q_1 - p^1)G_{\frac 1 a}(k^1 - q_1)\frac{1}{a^2}\,,
$$
$$
[(\partial_0 \otimes \partial_1)(\sigma_i \otimes \sigma_i) W_{\beta_i}](x,x) = \frac{1}{(2\pi)^4}\int dy^0\,dy^1 \, dz^0\,dz^1 \, \, d\bp \, dk^1 \, dq_1 \, \quad \times
$$
$$\times \quad \frac{\cos(\omp(x^0 - y^0)) \sin(\omega(k^1,\bpp)(x^0 - z^0)) g(y^0)g(z^0) \cos(\omega(q_1,\bpp)(y^0 - z^0))(-i k^1)}{    \omega(k^1,\bpp) \omega(q_1,\bpp) (\exp(\omega(q_1,\bpp) \beta_i) - 1) }\quad \times
$$
$$
\times \quad \Theta_\mp(y_1) \Theta_\mp(z_1) \exp i \left((p^1 - k^1) x^1 + (q_1 - p^1)y^1 + (k^1 - q_1)z^1\right)  G_{\frac 1 a}(q_1 - p^1)G_{\frac 1 a}(k^1 - q_1)\frac{1}{a^2}\,.
$$
The strategy to compute these integrals is as follows. We choose $\psi$ to be a step function and the integrals over $y^0$ and $z^0$ can be easily performed analytically as in the computation of the convergence. Unfortunately this does not apply to the remaining integrals. In order to avoid the numerical computation of multidimensional oscillating integrals we employ approximations which allow us to perform the momentum space integrals analytically.

\begin{itemize}
\item We choose a non-vanishing mass $m$ which sets the scale. We require $\beta_i m\gg 1$ in order to have $1/(\exp(\beta_i \omega - 1)\approx \exp(-\beta_i \omega)$. If $\beta_i m$ is too large we get an overall suppression and the computation is more sensitive to numerical errors. It turns out that $\beta_i m \gtrsim 50$ is sufficient for an overall numerical error in the percent regime.
\item  The mass also leads to a suppression of $\frac{\omega(q_1,\bpp)}{m} \gg \frac{1}{\sqrt{\beta_i m}}$. Thus for sufficiently large $\beta_i m$ we can approximate $\omega(q_1,\bpp)$ by $\omega(q_1,\bpp) \approx m + \frac{{q_1}^2}{2 m} +\frac{|\bpp|^2}{2 m}$ and analogously for $\omega(q_1,\bpp)^{-1}$.
\item The transformed Gaussians lead to a suppression of $|q_1 - p^1| \gg a^{-1}$, $|q_1 - k^1| \gg a^{-1}$. If $\frac{1}{a m}\lesssim \frac{1}{\sqrt{\beta_i m}}$, then we may approximate $\omega(k^1,\bpp)$ and $\omega(p^1,\bpp)$ in analogy to $\omega(q_1,\bpp)$. Consequently, after decomposing the sines and cosines into exponentials, all momentum space integrands are of the form ``Gaussian $\times$ polynomial'' and can be integrated analytically. The expressions are rather lenghty and difficult to manage without a computer algebra system, thus we do not display them here.
\item The remaining two spatial integrals have to be computed numerically. 
\item The errors from the analytic approximation and the numerical integration can be estimated by computing for the case of $\om_G(\wick{|\Phi^2(x)|})$ the expressions without the Heaviside functions but with the Gaussians in order to maintain the validity of the approximations. It is not difficult to see that for all values of $a$ including $a = 0$ one obtains the -- constant -- expectation values of $\wick{|\Phi^2(x)|}$ in the KMS state.
\end{itemize}

\section{Convergence of the NESS in perturbation theory}
\label{sec_convergenceV}

In this section we compute the convergence of the Feynman graphs of $\omGV$ at first order in perturbation theory. At this order, we have only one internal vertex, which is either dynamical or spectral. We consider all graphs with a dynamical vertex first. 

\subsection{Graphs with a dynamical vertex}
\label{sec_dynamical}

We will show that in the limit $t\to\infty$ all $\omG$ propagators in the Feynman graphs converge to $\omN$ propagators. This does not follow immediately from the convergence of $\omG$ to $\omN$ because the integration at the dynamical vertex is over an unbounded time interval. In other words, for each finite $t$ we can replace $V$ by a functional with compact spacetime support. However, the support of this functional is not uniformly bounded in $t$. For this reason, we need to show that, after performing the spacetime integral at the dynamical vertex, we obtain, in some cases -- in the adiabatic limit -- only after summing appropriate related graphs, expressions in which we can replace all $\omG$ propagators by $\omN$ propagators up to an $O(t^{-1})$ error. This already proves that the dynamical graphs are independent of $\chi_1$ and $\beta_3$. We then show that they are independent of $\psi$ in the limit as well, and show that this limit exists.

Before discussing individual graphs, we record a few useful identities. Using $\Delta_{F/\pm,G} = \Delta_{F/\pm,\infty} + W_G$, $\Delta_F-\Delta_- = i \Delta_R$, we can simplify renormalised pointwise products of $\Delta_{F,G}$ as follows
\beq \label{eq_loopexp}
(\Delta^2_{F,G})_\ren - \Delta^2_{-,G} = (\Delta^2_{F,\infty})_\ren - \Delta^2_{-,\infty} + 2 i \Delta_R W_G\,,
\eeq
and similarly for higher powers. This allows us to write comparatively simple expressions for renormalised Feynman loops. We shall need the Fourier transforms of (pointwise products of) the various propagators. To this avail, we use the following conventions for the Fourier transform for the remainder of this section. We define the Fourier transform of a tempered distribution $u(x)$, $x\in \bR^4$ as
$$
\widetilde u (p^0,\bp) \doteq \frac{1}{(2\pi)^2}\int_{\bR^4} d^4 x\; e^{i p^0 x^0 - i \bp \cdot \bx} u(x)\,.
$$
For Fourier transforms of $u(x)$, $x\in \bR^{d}$, $d<4$ we set
$$
\widetilde u (\bp) \doteq \frac{1}{(2\pi)^2}\int_{\bR^d} d\bx\; e^{- i \bp \cdot \bx} u(\bx)\,.
$$
Abusing notation, we consider translation invariant propagators as being tempered distributions in one variable which we denote by the same symbol. Recalling
$$
\Delta_{\pm,\infty}(x) = \frac{1}{(2\pi)^3}\int_{\bR^3}d\bp \; \frac{e^{\mp i \omp x^0 + i \bp \cdot \bx}}{2 \omp}\,,\qquad \omp\doteq \sqrt{|\bp|^2 + m^2}\,,
$$
$$
 \Delta_{F,\infty}(x) = \frac{-i}{(2\pi)^4}\int_{\bR^4}d^4p \; \frac{e^{-i p^0 x^0 + i \bp \cdot \bx}}{p^2 + m^2 - i\epsilon}\,,\qquad p^2 \doteq - (p^0)^2 + |\bp|^2\,,
$$
we find for pointwise squares
$$
\widetilde{\Delta^2_{\pm,\infty}}(p) = \frac{1}{(2\pi)^4} \int_{\bR^3}d\bq \frac{\delta(p^0 \mp (\omq +\ompq))}{2 \omq \omega_{\bp-\bq}}
$$
\begin{align*}
\widetilde{(\Delta^2_{F,\infty})_\ren}(p) & = -\frac{1}{(2\pi)^6} \int_{\bR^4}d^4q \left(\frac{1}{(q^2 + m^2 - i\epsilon)((p-q)^2 + m^2 - i\epsilon)}-\frac{1}{(q^2 + m^2 - i\epsilon)^2}\right)+C\,,\\
&= \frac{i}{(2\pi)^5}\int_{\bR^3}d\bq\left[\frac12 \frac{1}{p^0 + \omq - \om_{\bp-\bq}}\left(\frac{1}{\omq}\frac{1}{p^0+\omq + \ompq-i\epsilon}+\frac{1}{\ompq}\frac{1}{p^0-\omq - \ompq+i\epsilon}\right)+\right.\\
&\qquad \qquad \qquad \qquad + \left.\frac14 \frac{1}{\omq^3}\right] + C\\ 
& =\int_{\bR^3}d\bq\left( \frac{f_+(p^0,\bp,\bq)}{p^0-\omq - \ompq+i\epsilon}+\frac{f_-(p^0,\bp,\bq)}{p^0+\omq + \ompq+i\epsilon} + f_r(p^0,\bp,\bq)\right) + C
\,,
\end{align*}
$$
f_+(p^0,\bp,\bq) \doteq \frac{i}{(2\pi)^5}\frac12 \frac{1}{p^0 + \omq - \ompq}\frac{1}{\ompq}\qquad f_-(p^0,\bp,\bq) \doteq \frac{i}{(2\pi)^5}\frac12 \frac{1}{p^0 + \omq - \ompq}\frac{1}{\omq}
$$
where $C$ is a renormalisation constant and $f_\pm$ encodes the ``positive/negative frequency'' part of $(\Delta^2_{F,\infty})_\ren$ corresponding to $\Delta^2_{\pm,\infty}$. Recall that we are only considering $m>0$ for simplicity and note that the first factor in the last expressions is not a pole term because the integrand is regular for case where its denominator is vanishing. We shall also need Fourier transforms of $W_G(x,y)$ and its coinciding point limit
$$
w_G(x) \doteq W_G(x,x)
$$
which we can read off from e.g. \eqref{eq_conv1}, viz.
\begin{align*}
\widetilde {W_G}(p,k)& = (2\pi)^4 \int_{\bR}d\lambda \sum_{s_1,s_2,s_3,s_4,s_5 = \pm 1} \frac{f_W(\bp, \bk,\lambda,s_1,s_2,s_3,s_4,s_5)}{(p_1 - \lambda - i s_4 \epsilon)(k_1 + \lambda - i s_5 \epsilon)}\quad \times\\
& \qquad\times \quad \delta(p^0 + s_1 \omp)\delta(k^0 + s_2 \omk)\delta(\bkp+\bpp)
\end{align*}
where 
\begin{align*}
f_W(\bp, \bk,\lambda,s_1,s_2,s_3,s_4,s_5) &\doteq \frac{1}{(2\pi)^6}\frac{s_1 s_2 s_4 s_5\, b_{s_4,s_5}(\omega(\lambda,\bpp)\tilde \nu(p_1 - \lambda)\tilde \nu(k_1 + \lambda)}{8 \omp \omk \omega(\lambda,\bpp)}\quad \times\\
& \qquad \times \quad H(s_1,s_2,s_3,\bp,k_1,\lambda)\,.
\end{align*}
and we are using the same notation as in Section \ref{sec_convergence2pf}. For $w_G$ we have
\begin{align*}
\widetilde{w_G}(k) & = (2\pi)^2 \int_{\bR^3\times \bR}d\bl d\lambda\;  \sum_{s_1,s_2,s_3,s_4,s_5 = \pm 1} \frac{f_W(\bk-\bl, \bl,\lambda,s_1,s_2,s_3,s_4,s_5)}{(k_1 - l_1 - \lambda - i s_4 \epsilon)(l_1 + \lambda - i s_5 \epsilon)}\quad \times\\
& \qquad \times \quad \delta(k^0 + \om_{\bk,\bl,s_1,s_2})\delta(\bkp)\,.
\end{align*}
$$
 \om_{\bk,\bl,s_1,s_2} \doteq s_1 \om_{\bk - \bl}+ s_2 \om_{\bl}
$$
Finally, we record the Fourier transforms of the temporal switch-on function $\psi$ and of $\psi(x^0) \Theta(a-x^0)$
$$
\widetilde\psi(\lambda) = \frac{1}{i} \frac{\widetilde{\dot \psi}(\lambda)}{\lambda - i\epsilon}\,,\qquad \widetilde{[\psi \Theta(a-\cdot)]}(\lambda) = -\frac{1}{i\lambda}\left(e^{-i\lambda a}\widetilde{\dot \psi}(0)- \widetilde{\dot \psi}(\lambda)\right)\,,\quad a > \epsilon\,.
$$
We note that $ \widetilde{[\psi \Theta(a-\cdot)]}(\lambda)$ is a Schwartz function with
$$
 \widetilde{[\psi \Theta(a-\cdot)]}(0)= a + \gamma_\psi\,,\qquad \gamma_\psi\doteq  \int_{\bR} dx\;x \dot\psi(x) \,.
$$

\subsubsection{Graphs of type 1a and 1b}

We have to consider the following graphs, with all propagators being $\omG$ propagators. Note the minus sign between each pair of graphs of the same topology.
\begin{center}
\includegraphics[width=\textwidth]{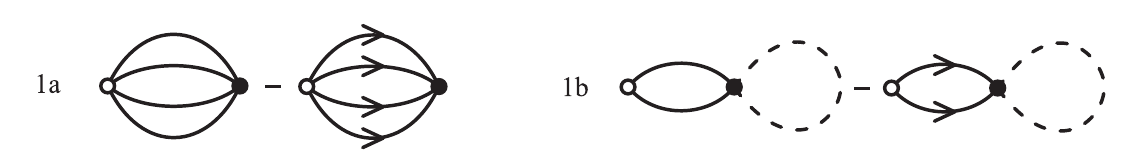}
\end{center}
We only discuss the 1b graphs, the convergence of the 1a graphs can be shown with similar arguments. Using \eqref{eq_loopexp} we decompose these graphs into a sum of two terms. The first of these terms is of the form
\begin{align*}
A_{1\text{b},1}(x) & = \int_{\bR^4} dz\; u(x-z) \psi(z^0) h(\bz) w_G(z) \\
&= C \int_{\bR^4 \times \bR^4} d^4p \,d^4k\; \tilde u(p) \widetilde \psi(k^0 - p^0) \widetilde h(\bp - \bk) \widetilde w_G(k) e^{-i p^0 x^0 + i \bp \cdot \bx}
\end{align*}
where $u \doteq (\Delta^2_{F,\infty})_\ren - \Delta^2_{-,\infty}$, $x$ is the external vertex, $C$ is a suitable constant and we use the notation introduced at the beginning of this section \ref{sec_dynamical}. The finite renormalisation freedom of  $(\Delta^2_{F,\infty}(x-z))_\ren $ is a multiple of  $\delta(x-z)$ and thus gives a contribution to $A_{1\text{b},1}(x)$ proportional to $w_G(x) = W_G(x,x)$. We already know that this converges to $W_N(x,x)$ by Theorem \ref{prop_NESS} proven in Section \ref{sec_convergence2pf}, thus we ignore this contribution. In order to evaluate the other contribution to $A_{1\text{b},1}(x)$, we perform the $k^0$ and $p^0$ integrals, using the residue theorem (with a contour closed in the upper half plane) and $x^0>0$ sufficiently large (larger than a radial bound of the support of $h$) for the pole contributions of the latter. We find
$$
A_{1\text{b},1}(x) = \int_{\bR^{13}}d\bp \,d\bk\, d\bq \,d\bl\,  d\lambda\; \widetilde h (\bp - \bk)e^{i\bp \cdot \bx}\sum_{s_1,s_2,s_3,s_4,s_5 = \pm 1} \frac{f_W(\bk-\bl, \bl,\lambda,s_1,s_2,s_3,s_4,s_5)\delta(\bkp)}{(k_1 - l_1 - \lambda - i s_4 \epsilon)(l_1 + \lambda - i s_5 \epsilon)}\quad \times
$$
$$
\times \quad \left[ C_1\sum_{s_6 = \pm 1}s_6\frac{e^{-s_6 i(\omq+\ompq)x^0}}{4 \omq \ompq} \frac{\widetilde{\dot \psi}\big(-\om_{\bk,\bl,s_1,s_2}- s_6(\omq + \ompq)\big)}{-\om_{\bk,\bl,s_1,s_2} - s_6(\omq + \ompq) - i\epsilon}\quad + \right.
$$
$$
+ \quad \left. C_2\; e^{i \om_{\bk,\bl,s_1,s_2} x^0} \left(\sum_{s_7 = \pm 1}\frac{f_{s_7}(- \om_{\bk,\bl,s_1,s_2} ,\bp,\bq)}{- \om_{\bk,\bl,s_1,s_2} -s_7(\omq + \ompq)+i\epsilon}+f_r(- \om_{\bk,\bl,s_1,s_2} ,\bp,\bq)\right)\widetilde{\dot \psi}(0)\right]\,.
$$
for suitable constants $C_1$ and $C_2$. The $C_1$ term corresponds to the $\Delta^2_{-,\infty}$ contribution and the residue of the $f_+$ pole in $(\Delta^2_{F,\infty})_\ren$, the $C_2$ term corresponds to the residue of the pole in $\widetilde \psi$.

We first consider the temporal limit without the spatial adiabatic limit and start with the $C_2$ term. We first perform the $k_1$ integral using a butterfly-contour centered at $k_1 = 0$ like in Section \ref{sec_convergence2pf}. On account of poles this fixes $k_1$ to be a suitable function of the other integration variables including $\lambda$ and $p_1$. Afterwards we perform the $\lambda$ and $p_1$ integrals -- in the order explained below -- using again a butterfly contour centered at values of these variables corresponding to $k_1(\lambda,p_1,...)=0$. The integrands are analytic on these contours and the circular parts of the contours vanish because of $\widetilde h$, whereas the diagonal parts give $O(1/x^0)$ contributions for sufficiently large $x^0$. For each integration we get a sum of two residues coming from either the poles in $\widetilde w_G$ (``$W$-poles'') or the poles in $(\Delta^2_{F,\infty})_\ren$  (``$F$-poles''). If we integrate a residue from an $F$-pole further, we only have $W$-poles left to consider. Thus, after the $k_1$, $\lambda$ and $p_1$ integration, we have a sum of terms which can be classified, in obvious notation indicating the order of integration, by ${k_1\atop \to}W{\lambda\atop \to}W{p_1\atop \to}F$, ${k_1\atop \to}W{\lambda\atop \to} F{p_1\atop \to} W$, ${k_1\atop \to}F{p_1\atop \to} W{\lambda\atop \to} W$. Let us discuss each of these contributions separately. After the first two integrations in ${k_1\atop \to}W{\lambda\atop \to}W{p_1\atop \to}F$, we get an expression which corresponds to replacing $w_G(z)$ by $w_N(z) = \text{const.}$, modulo an $O(1/x^0)$ error. In particular, this amounts to setting $\om_{\bk,\bl,s_1,s_2} = 0$ above which ``removes'' the $F$-poles because $m>0$ and thus $\omq + \ompq \ge 2m > 0$. Consequently, there are no poles left to integrate with respect to $p_1$ and we can perform the $p_1$ integral without using the residue theorem. We shall do this in a moment. First we shall argue that the contributions ${k_1\atop \to}W{\lambda\atop \to} F{p_1\atop \to} W$ and ${k_1\atop \to}F{p_1\atop \to} W{\lambda\atop \to} W$ are $O(1/x^0)$. This is the case because the residues at the $F$-poles contain $\Theta(-\om_{\bk,\bl,s_1,s_2}\pm 2m)$ factors, with $k_1$ replaced accordingly. However, after performing the subsequent integrations, we will always end up with $\om_{\bk,\bl,s_1,s_2} = 0 \neq \pm 2m$ up to $O(1/x^0)$ errors and with $k_1$ replaced accordingly. These considerations imply that the $C_2$ contribution to $A_{1\text{b},1}(x)$ is proportional to
$$
\int_{\bR^{6}}d\bp \,d\bq\; \widetilde h (\bp ) e^{i\bp \cdot \bx}\left(\sum_{s_7 = \pm 1}\frac{f_{s_7}(0 ,\bp,\bq)}{-s_7(\omq + \om_{\bp-\bq})}+f_r(0,\bp,\bq)\right)\widetilde{\dot \psi}(0) + O(1/x^0)\,.
$$
This integral converges absolutely and is manifestly independent of $\beta_3$, $\chi_1$, and also of $\psi$ because $\widetilde{\dot \psi}(0) = \sqrt{2\pi}^{-1}$ for all admissible $\psi$. In particular, the $\bq$ integral converges because of the renormalising subtraction in the definition of $(\Delta^2_{F,\infty})_\ren$. We remark that the convergence of the integral and the independence of $\psi$ are analogous to the same properties of the corresponding Feynman graph and integral for a perturbative KMS state.

We now discuss the $C_1$ contribution to $A_{1\text{b},1}(x)$, still without performing the spatial adiabatic limit. This contribution can be discussed in close analogy to the $C_2$ contribution. The pole structures are the same, with the only difference being that the poles of $\widetilde \psi$ in $C_1$ play the role of the $F$-poles in $C_2$. In order to follow the arguments in the analysis of the $C_2$ contribution, we first need to perform an integration in order to make the integrand oscillating in $k_1$. We can achieve this by integrating in either $p_2$ or $p_3$ along a butterfly contour centered at the origin. This gives the residue at the $\widetilde \psi$-pole, and thus $s_6(\omq + \ompq) = - \om_{\bk,\bl,s_1,s_2}$, up to an $O(1/x^0)$ error. The $k_1$ and $\lambda$ integrals may now be performed as before with the result that  we can replace $w_G$ in the $C_1$ contribution by $w_N$ modulo $O(1/x^0)$. Hence, the $C_1$ contribution to $A_{1\text{b},1}(x)$ is proportional to
$$
\int_{\bR^{6}}d\bp \,d\bq\; \widetilde h (\bp ) e^{i\bp \cdot \bx}\sum_{s_6 = \pm 1}s_6\frac{e^{-s_6 i(\omq+\ompq)x^0}}{4 \omq \ompq} \frac{\widetilde{\dot \psi}(- s_6(\omq + \ompq))}{ - s_6(\omq + \ompq) }+O(1/x^0) = O(1/x^0)\,,
$$
where the last identity follows by stationary phase arguments or by integrating w.r.t. e.g. $p_1$ along a butterfly contour.

We now discuss the spatial adiabatic limit. We find the same expressions as above, but with $\widetilde h$ proportional to a $\delta$-distribution. It is not difficult to see that the corresponding integrals all exist after integrating $A_{1\text{b},1}(x)$ with a test function in $x$. Thus, $A_{1\text{b},1}(x)$ exists as a distribution in the adiabatic limit. After performing the $\bp$ integration, the remaining integrand is meromorphic and we can use the same arguments as the ones used in the absence of the adiabatic limit in order to arrive at the conclusion that $w_G$ can be replaced in $A_{1\text{b},1}(x)$ by $w_N$ modulo an $O(1/x^0)$ error. In order to have all integrals converging, we can integrate in $x^0$ with respect to a test function $f(x^0-t)$, $f\in C^\infty_0(\bR)$ and consider $t\to \infty$. We are allowed to do this because we want to show convergence of $A_{1\text{b},1}(x)$ in the distributional sense. Consequently, we find in the adiabatic limit
$$
A_{1\text{b},1}(x)= C_1 \int_{\bR^{3}}\,d\bq\; \left(\sum_{s_7 = \pm 1}\frac{f_{s_7}(0 ,0,\bq)}{-2s_7\om_{\bq}}+f_r(0,0,\bq)\right)\widetilde{\dot \psi}(0) \quad +
$$
$$
+ \quad C_2  \int_{\bR^{3}}d\bq\; \sum_{s_6 = \pm 1}s_6\frac{e^{-2 s_6 i\omq x^0}}{4 \omq^2} \frac{\widetilde{\dot \psi}(- 2 s_6 \omq)}{ - 2 s_6 \omq }+O(1/x^0) \,,
$$
with suitable constants $C_1$ and $C_2$. The first integral vanishes identically because its integrand does, and the second one is $O(1/x^0)$ by e.g. stationary phase arguments.

We discuss the second contribution to the 1b graphs, which is of the form
\begin{align*}
A_{1\text{b},2}(x) & = \int_{\bR^4} dz\; u(x,z) \psi(z^0) h(\bz) w_G(z) 
\end{align*}
with 
$$
u(x,z) = 2 i \Delta_R(x-z)W_G(x,z) = 2i \Theta(x^0-z^0) \Delta(x-z)W_G(x,z)\,.
$$ 
We recall that $W_G$ is smooth such that $u$ is a well-defined (tempered) distribution. We first discuss the situation without the adiabatic limit. We have
$$
A_{1\text{b},2}(x^0+t,\bx) = 2i \int^{x^0}_{-\epsilon-t}\int_{\bR^3}dz^0 \,d\bz\; \Delta(x-z)W_G(x+t,z+t) \psi(z^0+t) h(\bz) w_G(z+t) 
$$
where we recall that $\psi(z^0) = 0$ for $z^0 < -\epsilon$. We decompose the $z^0$ integral into two parts,
\begin{align*}
A_{1\text{b},2}(x^0+t,\bx) & = 2i \int^{x^0}_{x^0-a}\int_{\bR^3}dz^0 \,d\bz\; \Delta(x-z)W_G(x+t,z+t) \psi(z^0+t) h(\bz) w_G(z+t) \\
&\quad + 2i \int^{x^0-a}_{-\epsilon-t}\int_{\bR^3}dz^0 \,d\bz\; \Delta(x-z)W_G(x+t,z+t) \psi(z^0+t) h(\bz) w_G(z+t) \,,
\end{align*}
where $a>1$ is such that $h(\bx) = 0$ for $|\bx|>a-1$.
The first contribution can be interpreted as evaluating a distribution of compact support ($\Delta$ times two Heaviside distributions) on a smooth function which is continuous in $t$. Thus, we may replace $W_G$ and $w_G$ by $W_N$ and $w_G$ up to $O(t^{-1})$ errors in this contribution. The integrand of the second term is a smooth function. On its integration domain, and uniformly in $x$ varying over compact sets, we have
$$
|\Delta(x-z)| \le  \frac{C_\Delta}{(x^0 - z^0)^{3/2}}\,, \qquad w_G(z+t) = w_N(0) + \frac{C_{w}}{t}\,,
$$
$$
W_G(x+t,z+t)  = W_N(x-z) +  \frac{C_{W,1}}{t}\,,\qquad |W_N(x-z)| \le  \frac{C_{W,2}}{x^0 - z^0}
$$
for suitable constants $C_\Delta$, $C_{W,1}$, $C_{W,2}$, $C_w$ and $t>0$. The decay estimates for $\Delta$ and $W_N$ are derived in Section \ref{sec_asymptotic}. These considerations imply that the $O(t^{-1})$ error terms are still $O(t^{-1})$ after the integration, and that the integral of the other terms exists.

We now discuss $A_{1\text{b},2}(x)$ in the adiabatic limit starting from 
\begin{align*}
A_{1\text{b},2}(x) & =  C \int_{\bR^{16}} d^4k \,d^4l \,d^4m\,d^4n\;e^{-i (k^0+m^0) x^0 + i (\bk+\bm) \cdot \bx}\quad\times \\
&\quad\times\quad\widetilde \Delta(k) \widetilde{[\Theta(x^0-\cdot) \psi]}\big(-k^0+n^0+l^0\big) \widetilde W_G(m,n) \delta \big(\bk-\bn-\bl\big)\widetilde w_G(l) 
\end{align*}
with $C$ a constant. We insert the Fourier transforms and perform the integrals over the temporal components of the momenta, finding
\begin{align*}
A_{1\text{b},2}(x) & =  C \int_{\bR^{17}} d\bk \,d\bl \,d\bm\,d\bn\,d\bq\,d\lambda\,d\mu\;\sum_{s_i,r_i = \pm1}s_6\,e^{-i (s_6 \omk - s_1 \omm) x^0 + i (\bk+\bm) \cdot \bx}\quad\times \\
&\quad\times\quad\frac{f_W(\bm, \bn,\lambda,s_1,s_2,s_3,s_4,s_5)}{(m_1 - \lambda - i s_4 \epsilon)(n_1 + \lambda - i s_5 \epsilon)} \frac{f_W(\bl-\bq, \bq,\mu,r_1,r_2,r_3,r_4,r_5)}{(l_1 - q_1 - \mu - i r_4 \epsilon)(q_1 + \mu - i r_5 \epsilon)}\quad\times \\
&\quad\times\quad  \frac{1}{\omk} 
\widetilde{[\Theta(x^0-\cdot) \psi]}\big(-s_6 \omk -s_2 \omn -\om_{\bl,\bq,r_1,r_2}\big)\delta \big(\bk-\bn-\bl\big) \delta(\bmp+\bnp)\delta(\blp)\,.
\end{align*}
Once again, it is not difficult to see that all integrals exist after integrating $A_{1\text{b},2}(x)$ with a test function in $x$; this implies that  $A_{1\text{b},2}(x)$ exists as a distribution in the adiabatic limit. We set
$$
A_{1\text{b},2}(x) \doteq B_{1\text{b},2}(x,\tau)|_{\tau=x^0}
$$
where $B_{1\text{b},2}(x,\tau)$ is defined by replacing $\widetilde{[\Theta(x^0-\cdot) \psi]}$ with $\widetilde{[\Theta(\tau-\cdot) \psi]}$ in $A_{1\text{b},2}(x)$. We know that we can replace $W_G$ and $w_G$ in $B_{1\text{b},2}(x,\tau)$ by $W_N$ and $w_G$ up to an $O(1/x^0)$ error. However, since the integrand of $B_{1\text{b},2}(x,\tau)$ is $O(\tau)$ if the argument of $\widetilde{[\Theta(x^0-\cdot) \psi]}$ is vanishing, we don't know if the $O(1/x^0)$ errors are still vanishing for large $x^0$ also after setting $\tau = x^0$. To show that this is actually the case, we consider $\partial_\tau B_{1\text{b},2}(x,\tau)$ which is manifestly uniformly bounded in $\tau$. We define $D(t,\tau)\doteq \int_{\bR^4}d^4x\; f(x^0-t,\bx)\partial_\tau B_{1\text{b},2}(x,\tau)$ with $f$ a test function. We can perform the integrals in $D(t,\tau)$ in the following order. We first integrate with respect to $\bk$, $\blp$ and $\bnp$. The resulting integrand is oscillating in $l_1$ and $n_1$ and we can perform the $l_1, \mu, n_1, \lambda$ integrals using butterfly contours. The result of each of these four integrals is a residue plus an error term corresponding to the integral on the diagonal part of the butterfly contour, which is an integral of an absolutely integrable function. Thus we get a sum of $2^4$ terms, one of which corresponds to considering only the residue of each integral, while the others correspond to consider at least one of the error terms. We denote the former contribution by $D_N(t,\tau)$ and thus have $D(t,\tau) = D_N(t,\tau) + D_E(t,\tau)$, with $D_E(t,\tau) \doteq D(t,\tau) - D_N(t,\tau) $. The $(t,\tau)$-dependence of remaining integrand in $D_N(t,\tau)$, which is absolutely integrable with respect to the remaining variables $\bm$ and $\bq$, is $\exp(-i(s_6+s_2)\omm(t-\tau))$. This follows from the fact that the residues of the $l_1, \mu, n_1, \lambda$ integrals enforce $l_1 = n_1-m_1 = 0$. Thus, $D_N(t,\tau)$ is bounded in $(t,\tau)$, and constant for $t=\tau$. The situation is different for $D_E(t,\tau)$, which is asymptotically decaying in $t$ uniformly in $\tau$. To see this, let us consider the contribution $D_{E_1}(t,\tau)$ to $D_E(t,\tau)$ which corresponds to considering the residues of the $l_1, \mu, n_1$ integrals and the error terms of the $\lambda$-integral. This contribution reads
$$
D_{E_1}(t,\tau) =   \int_{\bR^{4}} d\bm \;d\lambda\; F(\lambda, \bm) \exp\left(+is_1 \omm t+i\sqrt{(\lambda - i s_6(\lambda - m_1))^2+\bmp^2}(-s_6(t-\tau)+s_2 \tau)\right)
$$
where $F$ is such that the integrand is absolutely integrable for sufficiently large $t$. By stationary phase arguments, $D_{E_1}(t,\tau) = O(t^{-2})$ uniformly in $\tau$. With similar arguments one can also show that $D_{E}(t,\tau) = O(t^{-2})$ uniformly in $\tau$. This implies that we can replace $W_G$ and $w_G$ in $A_{1\text{b},2}(x)$ by $W_N$ and $w_N$ up to an $O(1/x^0)$ error, viz.
\begin{align*}
A_{1\text{b},2}(x) & =  C \int_{\bR^{3}} d\bp\sum_{s_6,s_2=\pm1}\frac{s_6\, e^{-i (s_6+s_2)\omp x^0 }b_{\beta(-s_2 p_1)}(\omp)\widetilde{[\Theta(x^0-\cdot) \psi]}(-(s_6+s_2)\omp)}{\omp^2} + O(1/x^0)\\
&= i C \int_{\bR^{3}} d\bp\sum_{s=\pm1}\frac{s\, b_{\beta(-s p_1)}(\omp)}{\omp^2}  \frac{\left(\widetilde{\dot \psi}(0)-e^{-i 2 s \omp x^0 } \widetilde{\dot \psi}(-2 s \omp)\right)}{-2 s \omp}+ O(1/x^0)
\end{align*}
with $b_{\beta(p_1)}(\omp) = \exp\big((\beta_1 \Theta(p_1)+\beta_2 \Theta(-p_1))\omp-1\big)^{-1}$. Note in particular that the $s_6 + s_2 = 0$ terms in the first expression, which are both $O(x^0)$ and thus potentially obstructing the existence of the large-time limit, cancel each other. The remaining $\widetilde{\dot \psi}(0)$-term is manifestly independent of $\psi$, finite and constant in $x^0$. The $\widetilde{\dot \psi}(-2 s \omp)$-term is $O(1/x^0)$ by stationary phase arguments.

\subsubsection{Graphs of type 2a, 2b and 2c}

The graphs with two external vertices and a dynamical internal vertex are enumerated in the following figure, with all propagators being $\om_G$ propagators. We omit all graphs which are related by exchanging the two external vertices.
\begin{center}
\includegraphics[width=\textwidth]{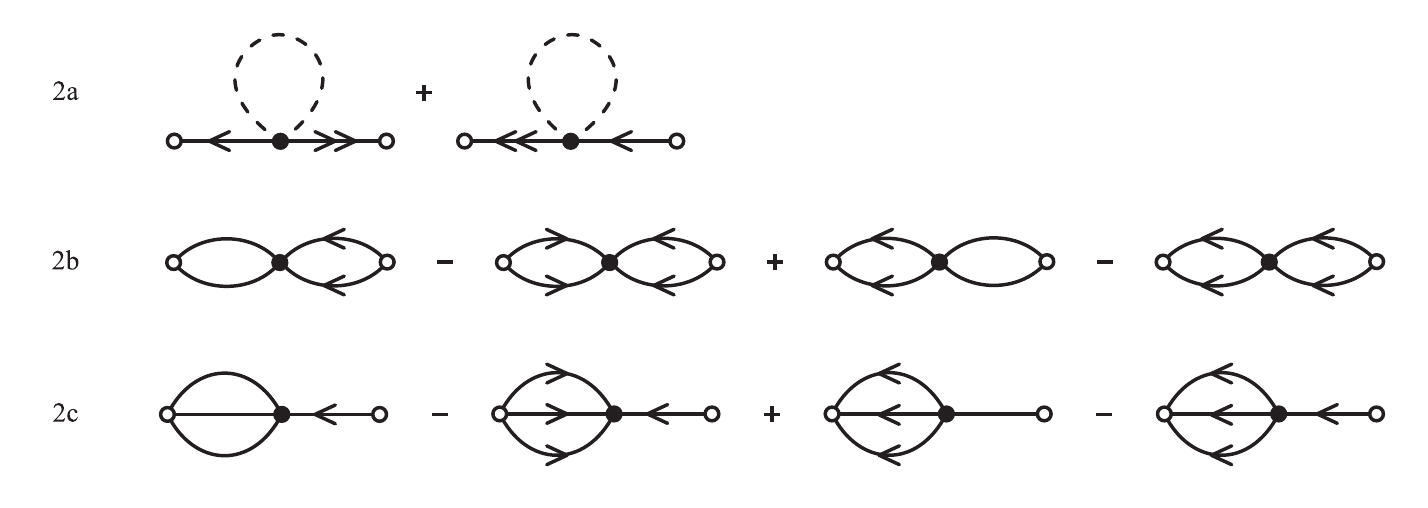}
\end{center}
We have seen in the discussion of the graphs with one external vertex that in the adiabatic limit terms which grow linearly with time may arise. In the graphs with one vertex, these terms identically cancel and thus do not pose obstructions for the long time limit. Similar terms appear in the adiabatic limit for graphs with more external vertices. In the case of two external vertices, there are three distinct mechanisms which prevent such terms to be problematic at large times. We shall discuss them now. To this end, we remark that these terms only appear in the contributions with $\om_N$ propagators and not in the error terms of the $\om_G - \om_N$ contributions which vanish at large times. With this in mind, we consider the following expression
$$
A_{D,W_1,W_2}(x,y) \doteq \int_{\bR^4}dz\;W_1(x-z)\Delta_R(x-z) \psi(z^0) D(z-y)W_2(z-y)
$$
where we assume that the integral exists, that all quantities are tempered distributions, and that
$$
W_i(x) = W_i(-x) = \overline{W_i(x)}\,,\qquad W_i(x^0,\bx) = W_i(x^0,-\bx)\,,\qquad i\in\{1,2\}\,,
$$
$$
D(-x) = \overline{D(x)}\,,\qquad D(x^0,\bx) = D(x^0,-\bx)\,.
$$
We have
\begin{align*}
A_{D,W_1,W_2}(x,y) & = \int_{\bR^{16}}d^4k\,d^4p\,d^4q\,d^4l\;\widetilde{W_1}(k)\widetilde{\Delta}(p)\widetilde D(q)\widetilde{W_2}(l)e^{i(\bk+\bp)\cdot (\bx-\by)}\delta(\bq + \bl - (\bk + \bp))\quad \times\\
& \qquad \times  \quad \widetilde{[\Theta(x^0-\cdot) \psi]}(q^0+l^0-(k^0+p^0))e^{-i\big((k^0+p^0)x^0-(q^0+l^0)y^0\big)} \,.
\end{align*}
We observe that terms linear in $x^0$ can only arise if the integrand has support where $q^0+l^0=k^0+p^0$. For example, if we consider the 2c graphs with all propagators being vacuum ones, then this does not occur because four vectors on the mass shell (recall that we consider only $m>0$) can not sum to zero if three are pointing in the same direction. The next observation is that, if $D$ is real (and thus symmetric), then $\widetilde D(-p) = \widetilde D(p)$, whereas $\widetilde \Delta(-p) = - \widetilde \Delta(p)$, because $\Delta$ is antisymmetric. This implies that all contributions with $q^0+l^0=k^0+p^0$ cancel for a symmetric $D$ after integration. The last mechanism which renders potentially growing terms harmless is a cancellation which is similar to the one we have observed in graphs with one external vertex. To this end we first consider $B(x,y)\doteq A_{D,W_1,W_2}(x,y) + A_{\overline{D},W_1,W_2}(y,x)$  and find after a few manipulations
\begin{align*}
B(x,y) & =  \int_{\bR^{16}}d^4k\,d^4p\,d^4q\,d^4l\;\widetilde{W_1}(k)\widetilde{\Delta}(p)\widetilde D(q)\widetilde{W_2}(l)e^{i(\bk+\bp)\cdot (\bx-\by)}\quad \times\\
& \qquad \times \quad\left( \widetilde{[\Theta(x^0-\cdot) \psi]}(q^0+l^0-(k^0+p^0)) e^{-i\big((k^0+p^0)x^0-(q^0+l^0)y^0\big)}\quad -\right.\\
 & \qquad \quad -\quad \left. \widetilde{[\Theta(y^0-\cdot) \psi]}(k^0+p^0-(q^0+l^0)) e^{-i\big((q^0+l^0)x^0-(k^0+p^0)y^0\big)}\right)\quad \times \\
& \qquad \times \quad\delta(\bq + \bl - (\bk + \bp))\,.
\end{align*}
Clearly, all contributions of the integral for $q^0+l^0=k^0+p^0$ are of the form $(x^0 - y^0) \times g(x^0-y^0,\bx-\by)$ for a suitable $g$, and thus harmless in the limit $x^0 + t$, $y^0 + t$ for $t\to\infty$. Another cancellation occurs in $C(x,y)\doteq A_{D,W_1,1}(x,y) +A_{ \overline{D},1,W_1}(y,x)$ with $D = \Delta_{+,\infty}$. For such a contribution we find
\begin{align*}
C(x,y) & =  \int_{\bR^{16}}d^4k\,d^4p\,d^4q\;\widetilde{W_1}(k)\widetilde{\Delta}(p)\widetilde{\Delta_{+,\infty}}(q)\delta(\bq + \bl - (\bk + \bp))\quad \times\\
& \qquad \times \quad\left( \widetilde{[\Theta(x^0-\cdot) \psi]}(q^0-p^0-k^0) e^{-i\big((k^0+p^0)x^0-q^0 y^0\big)}e^{i(\bk+\bp)\cdot (\bx-\by)}\quad -\right.\\
 & \qquad \quad -\quad \left. \widetilde{[\Theta(y^0-\cdot) \psi]}(p^0-q^0-k^0) e^{-i\big((k^0+q^0)x^0-p^0y^0\big)}e^{i(\bk+\bq)\cdot (\bx-\by)}\right)\,.
\end{align*}
If $q^0-p^0-k^0 = 0$ or $p^0-q^0-k^0=0$ respectively for the two contributions, then we can exchange $\bq$ for $\bp$ in the second term because of the similarities of $\widetilde{\Delta}$ and $\widetilde{\Delta_{+,\infty}}$. Consequently, all contributions to $C(x,y)$ which grow linearly in either $x^0$ or $y^0$ combine to form time-translation-invariant expressions.

If we consider all graphs of the types 2a, 2b and 2c with all propagators being $\om_N$ ones and expand all such propagators into sums of $\om_\infty$ propagators and $W_N$, we can observe that all contributions are harmless in the large time limit because they fall into one of the cases we just discussed.

Apart from this observation, the analysis of the large time limit of the graphs with two external vertices $(x^0+t,\bx)$ and $(y^0+t,\by)$ and one dynamical vertex can be made using arguments and computational steps which are very similar to the ones used and made in the analysis of graphs with one external vertex. In particular, the integrals exist in the adiabatic limit, all $\om_G$ propagators can be replaced by $\om_N$ propagators up to $O(t^{-1})$ error terms, and the limit $t \to \infty$ of the remaining expressions exists and is independent of $\psi$.

\subsubsection{Graphs of type 3 and 4}
One again we enumerate all possible graphs of these type, omitting all graphs which are related by exchanging external vertices. The empty squares are placeholders for either type of single arrow, with the limitation that in the first two graphs of type 3 both arrows must point in the same direction.
\begin{center}
\includegraphics[width=\textwidth]{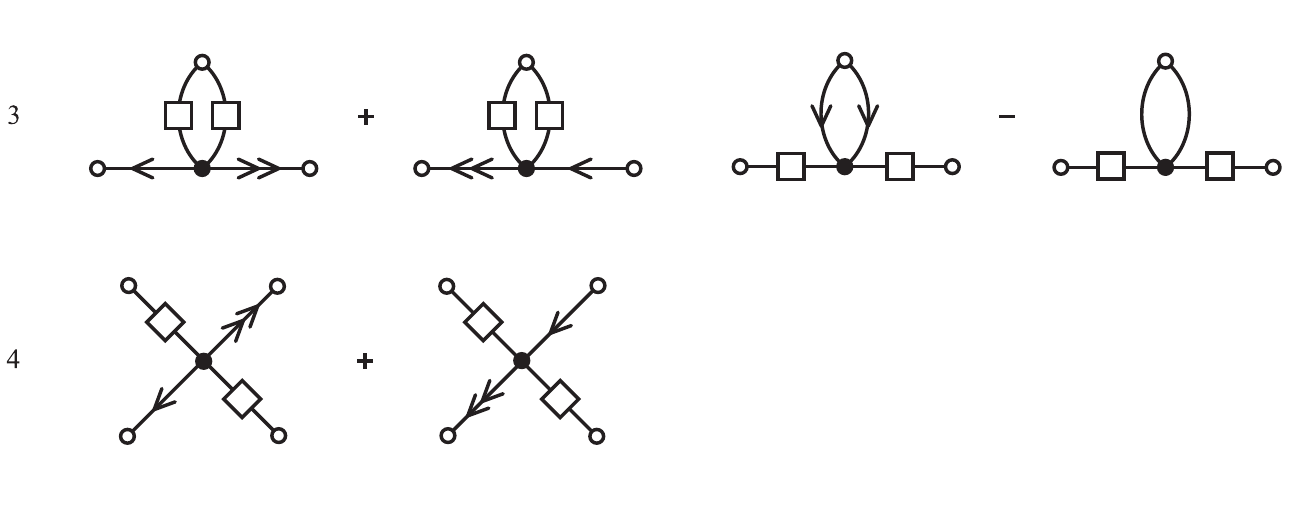}
\end{center}
As with the graphs of type 2, we remark that all these expressions exist in the adiabatic limit, and that either with our without this limit the propagators of $\om_G$ can be replaced by $\om_N$ propagators up to $O(t^{-1})$ errors. The long time-limits of the $\om_N$ contributions exist and are independent of $\psi$. All these statements can be verified using the same arguments we already employed in discussing the graphs with less vertices. In particular, the fact that potential linearly growing terms are absent in the large time limit can be seen for the first two classes of type 3 graphs and all type 4 graphs in the above figure by slightly generalising the corresponding discussion for two external vertices by adding an additional vertex factor which depends only on the remaining (one or two) vertices and thus does not interfere with the mechanisms that prevent or cancel the linearly growing terms arising from the two vertices considered. For the last two classes of type 3 graphs, the linearly growing terms can be seen to cancel much like the ones in the discussion of type 1 graphs.

\subsection{Graphs with a spectral vertex}
\label{sec_spectral}

We now discuss graphs with the only internal vertex being a spectral vertex. In these graphs, all propagators connected to external vertices are of the form $[\sigma_i \Delta_{+,\beta_j}](x,z)$ with $x$ an external vertex, $z$ the internal vertex and combinations of $i\in\{1,2\}$, $j\in\{1,2,3\}$ which are ``allowed'' by the definition of $\omGV$ in Theorem \ref{prop_initialinteracting}. Moreover, all propagators connecting the internal vertex to itself are of the form $W_{\beta_j}(z,z)=\text{const.}$. $[\sigma_i \Delta_{+,\beta_j}](x,z)$ is distribution of Hadamard type by Proposition \ref{prop_Hadamard}. Thus, its pointwise products exist, which was of course necessary for $\omGV$ to be well-defined in the first place. In the analysis of the existence of the adiabatic limit of the spectral graphs, as well as in the analysis of their asymptotic limits with or without the adiabatic limit, it does, as will be evident from the following discussion, not make a difference whether the external vertices are separate or coinciding. Thus it is sufficient to analyse only the graphs of type 2a and 4 which are displayed in the following figure.
\begin{center}
\includegraphics[width=0.6\textwidth]{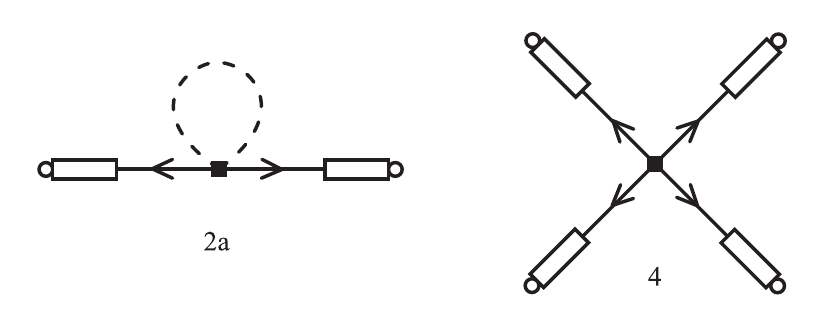}
\end{center}
We define
$$
\Delta_{+,\beta,s}(x,y)\doteq [\sigma_s \Delta_{+,\beta}](x,y)\,,\qquad s=\pm 1\,,\quad \sigma_{\pm 1} \doteq \sigma_{1/2}\,
$$
and investigate the behaviour of this distribution for large $x^0$. We shall in fact do this for the analytically continued distributions $\Delta_{+,\beta,i}(x,y^0+iu,\by)$, $u\in [0,\beta]$, which appear in the graphs with a spectral index. We have, using the notation of Section \ref{sec_convergence2pf},
\begin{align*}
\Delta_{+,\beta,s}(x,y^0+iu,\by) & =\int_{\bR^4}d^4z\; \Delta(x-z) g(z^0) \chi_s(z^1) \Delta_{+,\beta}(z,y^0+iu,\by)\\
&= \frac{-s i}{(2\pi)^4}\int_{\bR^6}d\bp\,d\bk\; \sum_{s_1,s_2=\pm1}\frac{s_1s_2e^{i(s_1\omp x^0 - s_2\omk y^0)+s_2\omk u}b_\beta(s_2\omk)}{4 \omp \omk}\quad \times\\
&\quad\times\quad(s_1\omp + s_2 \omk) \widetilde{\dot\psi}(s_1\omp - s_2 \omk)\frac{\tilde\nu(k_1-p_1)}{k_1 - p_1 - is_3 \epsilon}e^{i(\bp \bx - \bk \by)}\delta(\bpp - \bkp)\,,
\end{align*}
Using a butterfly contour as in Section \ref{sec_convergence2pf} we can perform the $p_1$ integral and find
$$
\Delta_{+,\beta,s}(x,y^0+iu,\by)  = \Delta^{(s)}_{+,\beta}(x,y^0+iu,\by) + O(1/x^0)\,,
$$
$$
\Delta^{(s)}_{+,\beta}(x,y) \doteq \frac{1}{(2\pi)^3}\int_{\bR^3}d\bp\; \sum_{s_1\pm1}\Theta(s s_1 p_1)\frac{s_1 e^{i(s_1\omp (x^0 -  y^0)}b_\beta(s_1\omp)}{2 \omp }e^{i\bp \cdot(\bx - \by)}\,.
$$
We note that $\Delta^{(s)}_{+,\beta}$ is a Schwartz distribution and that
$$
\Delta_{+,N}=\Delta^{(+1)}_{+,\beta_1}+\Delta^{(-1)}_{+,\beta_2}\,.
$$
Moreover, the analysis of the asymptotic properties of $\Delta_{+,N}$ in Section \ref{sec_asymptotic} implies that, for sufficiently large $x^0$ and $u\in [0,\beta]$
$$
\Delta^{(s)}_{+,\beta}(x^0+iu,\bx) \le \frac{C}{|x^0|}
$$
uniformly in $\beta$ and in $\bx$ varying over compact subsets. These considerations allow us to straightforwadly discuss the graphs with a spectral vertex in the absence of the adiabatic limit. In fact, their amplitudes are of the form
$$
A_{\beta_i,s_1,\dots,s_n}(x_{(1)},\dots,x_{(n)}) = W_{\beta_i}(0,0)^{\delta_{n,2}}\int_{\bR^4}d^4z\int^{\beta_i}_{0}du \; \dot\psi(z^0)h(\bz)  \prod^n_{j=1}\Delta_{+,\beta_i,s_j}(x_{(j)},z^0+iu,\bz)
$$
with $\beta_i = \beta_{1/2}$ for $s_1 =\dots=s_n = \pm1$ and $\beta_i = \beta_3$ otherwise; additionally, $n\in\{2,4\}$. For sufficiently large $x^0_{(j)}$ for all $j$, the integrand is bounded by $|\min(x^0_{(1)},\dots,x^0_{(n)})-z^0|^{-n}$ and the integral is vanishing like $|\min(x^0_{(1)},\dots,x^0_{(n)})|^{-n}$ because all integrations are over a compact fixed region. This shows that, at first order, all graphs with a spectral vertex are vanishing at large times if we do not perform the adiabatic limit.

We now analyse $A_{\beta_i,s_1,\dots,s_n}\big(x^0_{(1)}+t,\bx_{(1)},\dots,x^0_{(n)}+t,\bx_{(n)}\big)$ for large times after performing the adiabatic limit. Replacing all factors in the integrand by their Fourier transforms, it is manifest that the $z$-integral exists in the sense of a distribution in the adiabatic limit, uniformly in $t$. Thus, for sufficiently large $t$, we may replace the $\Delta_{+,\beta,s}$-propagators by $\Delta^{(s)}_{+,\beta}$-propagators up to an $O(t^{-1})$ error. After doing so, and performing the $u$ integral, we find, for a suitable constant $C$,
$$
A_{\beta_i,s_1,\dots,s_n}\big(x^0_{(1)}+t,\bx_{(1)},\dots,x^0_{(n)}+t,\bx_{(n)}\big)=C\sum_{r_1,\dots,r_n=\pm1} \int_{\bR^{3n}}\prod^n_{j=1}d\bp^{(j)} \frac{r_j b_{\beta_i}(r_j \om_{\bp^{(j)}})}{\om_{\bp^{(j)}}}\Theta\big(s_j r_j p^{(j)}_1\big)\quad\times
$$
$$
\times\quad \exp \left(i\sum^n_{j=1}\left(r_j \om_{\bp^{(j)}} (x^0_{(j)}+t) + \bp^{(j)} \cdot \bx_{(j)}\right)\right)\frac{\exp\left(\beta_i\sum^n_{j=1}r_j \om_{\bp^{(j)}}\right)-1}{\sum^n_{j=1}r_j \om_{\bp^{(j)}}}\quad\times
$$
$$
\times\quad \widetilde{\dot\psi}\left(\sum^n_{j=1}r_j \om_{\bp^{(j)}}\right)\delta\left(\sum^n_{j=1}\bp^{(j)} \right) + O(t^{-1})\,.
$$
We face the situation that we sum $n$ four-vectors on the mass-shell whose spatial components sum up to zero. This implies
$$
\sum^n_{j=1}r_j \om_{\bp^{(j)}} = \sum^n_{i=1}r_i \sum^n_{j=1} \om_{\bp^{(j)}}\,.
$$
The contributions with $\sum_i r_i \neq 0$ are $O(t^{-1})$ by stationary phase arguments. The $\delta$-factor implies in conjuction with the $\Theta$-factors that the prevailing contributions with $\sum_i r_i = 0$ are only non-vanishing if
$$
\sum^n_{j=1}r_j s_j = 0\,,
$$
i.e.~if $s_1 = \dots = s_n$. Consequently, we arrive at the conclusion that the graphs with a spectral vertex are asymptotically independent of $\beta_3$ in the adiabatic limit. We note that a similar observation implies that $\omNV = \ombV$ for $\beta_1 = \beta_2$, meaning that the $\Theta$-terms cover all the momentum space which is accessible to an interacting KMS state if we sum over both possible signs of $s_1 = \dots = s_n = \pm1$.

\section{Asymptotic properties of $\omb$ and $\omN$}
\label{sec_asymptotic}

\begin{proposition}  $\Delta_{+/F,\beta}(t,\bx)$ and their spacetime derivatives are for  large $t$ and $m>0$, $d=4$ bounded by $C_0 |t|^{-3/2}$ uniformly in $\bx$ (varying over compact sets). This also holds for the analytic continuation $\Delta_{+}(t-iu,\bx)$, $0\le u \le \beta$. In the absence of spacetime derivatives, the $|t|^{-3/2}$ decay is sharp.

$\Delta_{+/F,N}(t,\bx)$ and their spacetime derivatives are for  large $t$ and $m>0$, $d=4$ bounded by $C |t|^{-1}$ uniformly in $\bx$ (varying over compact sets), for large $x_1$ and $m>0$ bounded by $C_1 |x_1|^{-1}$ and for large $x_2$, $x_3$ exponentially decaying. In the absence of spacetime derivatives, the $|t|^{-1}$ decay and $|x_1|^{-1}$ decay are sharp.
\end{proposition}

\begin{proof} For large $t$, $\Delta_F=\Delta_+$ or $\Delta_F=\overline{\Delta_+}$ thus is is sufficient to consider the Wightman functions. We have
$$
\Delta_{+,\beta/N} = \Delta_{+,\infty} + W_{\beta/N}\,,
$$
and 
$\Delta_{+,\infty}(t,\bx)$ is a modified Bessel function which along with its derivatives is known to have $|t|^{-3/2}$ decay and expontential decay in space -- see e.g. [DFP]. Thus it is convenient to consider the asymptotic properties of $W_{\beta/N}$ to avoid convergence issues in the UV. For $\Delta_{+,\beta}(t-iu,\bx)$, $0< u  < \beta$ is is however more convenient to consider the Wightman function itself, because the analytic continuation is smooth.

We consider 
$$
Z_{s_1,s_2,\pm,\beta}(t,\bx) =\frac{1}{(2\pi)^3}\int_{\bR_\pm \times \bR^2 } \frac{d\bp}{2\omp} e^{ i s_1 (\omp t - \bp \cdot \bx)}s_2b(s_2\om)\,,
$$
$$
\omp \doteq \sqrt{p^2 + m^2}\,,\qquad p \doteq |\bp|\,,\qquad b (\om) \doteq \frac{1}{e^{\beta \om} -1}
$$
and we have
$$
W_{\beta} = \sum_{s_1=\pm 1}Z_{s_1,+1,+,\beta}+Z_{s_1,+1,-,\beta}\,,\qquad W_{N} = \sum_{s_1=\pm 1}Z_{s_1,+1,+,\beta_1}+Z_{s_1,+1,-,\beta_2}
$$
$$
\Delta_{+,\beta} = \sum_{s_1=s_2,s_1=\pm 1}Z_{s_1,s_2,+,\beta}+Z_{s_1,s_2,-,\beta}\,.
$$

We follow the proof strategy of \cite[Appendix A]{BB02}. We first consider $W_\beta$, assume w.l.o.g. that $t>0$, set $w=(\om - m)t$ (note the typo in \cite{BB02})  and find ($r = |\bx|$)
$$
W_{\beta}(t,\bx) = \sum_{s = \pm 1}\frac{e^{is mt}}{2\sqrt{2\pi}t^{3/2}} C_s(t)\,,\qquad C_s(t)\doteq \int^\infty_0 dw \,\sqrt{w} \,e^{i s w}\,c_s(w/t)
$$
$$
c_s(y)\doteq b (y + m)\, \sqrt{y + 2 m}\,\int^1_{-1} dz \, e^{- i s r \sqrt{y^2 + 2ym} z}= 2   b (y + m)\, \sqrt{y + 2 m}\, \frac{\sin (s r \sqrt{y^2 + 2ym})}{s r \sqrt{y^2 + 2ym}}
$$
We decompose $C_s(t)$ as
$$
C_s(t) = c_s(0) \lim_{\epsilon \downarrow 0} \int^\infty_0 dw \,\sqrt{w} \,e^{i sw-\epsilon w}  + \lim_{\epsilon \downarrow 0} \int^\infty_0 dw \,\sqrt{w} \,(c_s(w/t)-c_s(0)) e^{i sw-\epsilon w} 
$$
The first term is finite and constant in $t$. In order to obtain an estimate for the second term we perform two partial integrations starting from
$$
e^{i sw-\epsilon w}  = \frac{1}{(i s -\epsilon)^2} \pa^2 _w(e^{i s w-\epsilon w} - 1)\,.
$$
The boundary terms at 0 vanish by construction also without introducing the $-1$, but this extra regularity at 0 will prove helpful later. We find
$$
\lim_{\epsilon \downarrow 0}\int^\infty_0 dw \,\sqrt{w} \,(c_s(w/t)-c_s(0)) e^{i s w-\epsilon w}
$$
$$ 
 = \lim_{\epsilon \downarrow 0}\int^\infty_0 dw \,(e^{i s w-\epsilon w/t}-1)w^{-3/2+\delta} t^{-\delta} \sum^2_{n=0}c_n (w/t)^{n-\delta} \pa^n_{w/t}(c_s(w/t)-c_s(0)) )
$$
for suitable constants $c_n$ and an arbitrary $\delta$ which we choose as $0<\delta < 1/2$. For this choice of $\delta$, $y^{n-\delta} \pa^n_y (c_s(y)-c_s(0))$ is smooth and bounded on $y\in(-2m,\infty)$ for all $n\ge 0$ uniformly in $r$ varying over compact sets. Thus we can estimate the last expression by
$$ 
 \lim_{\epsilon \downarrow 0}C  t^{-\delta} \int^\infty_0 dw \,|e^{i s w-\epsilon w}-1|w^{-3/2+\delta} 
$$
where $C$ is independent of $t$. The last integral exists and as constant in $t$. The proof can be repeated for $\Delta_{+,\beta}(t-iu,\bx)$, $0< u  < \beta$ and for arbitrary spacetime derivatives with minor modifications.

For $W_N$, we have to reconsider the proof because the fact that we integrate separately over one half of the momentum space for different values of $\beta$ modifies the regularity for small momenta. We have $W_N = X_{+,\beta_1} + X_{-,\beta_2}$, 
$$
X_{\pm,\beta}(t,\bx) \doteq \sum_{s = \pm 1}\frac{ e^{ismt}}{2 \sqrt{2\pi}t} D_{\pm,s}(t)\,,\quad D_{\pm,s}(t)\doteq \int^\infty_0 dw \,e^{i s w}\,d_{\pm,s}(w/t)
$$
$$
d_{\pm,s}(y)\doteq  \pm b(y + m)\, \sqrt{y + 2 m}\,\int^1_0 dz\, e^{\mp  i s r \sqrt{y^2 + 2ym} z} = i r^{-1}   b (y + m)\,(e^ {\mp i s r \sqrt{y^2 + 2ym}}-1)
$$
Once again we decompose $D_{\pm,s}(t)$ 
$$
D_{\pm,s}(t)(t) = d_{\pm,s}(0) \lim_{\epsilon \downarrow 0} \int^\infty_0 dw \,e^{i s w-\epsilon w}  + \lim_{\epsilon \downarrow 0} \int^\infty_0 dw  \,(d_{\pm,s}(w/t)-d_{\pm,s}(0)) e^{i s w-\epsilon w}
$$
finding that the first term is finite and constant in $t$ and estimating the second one by means of two partial integrations and for $0<\delta < 1$ as
$$
\lim_{\epsilon \downarrow 0}\int^\infty_0 dw  \,(d_{\pm,s}(w/t)-d_{\pm,s}(0)) e^{i s w-\epsilon w}
$$
$$ 
 = - \lim_{\epsilon \downarrow 0}\int^\infty_0 dw \,(e^{i s w-\epsilon w/t}-1)w^{-2+\delta} t^{-\delta}  (w/t)^{2-\delta} \pa^2_{w/t}d_{\pm,s}(w/t) 
$$
$$ 
\le    \lim_{\epsilon \downarrow 0}C  t^{-\delta} \int^\infty_0 dw \,|e^{i s w-\epsilon w}-1|w^{-2+\delta}  \le C' t^{-\delta}
$$

In order to investigate the spatial asymptotics of $W_N$  we note that 
$$
-x_1 \sqrt{2\pi}^3 X_{+,\beta}(t,\bx) =  \int_{\bR_+ \times \bR^2} \frac{d\bp}{\omp} \,\frac{\partial_{p_1}\sin \left(\omp t - \bp \cdot \bx\right)}{\exp(\beta\om)-1} = D(t,\bx)
$$
with $D(t,\bx)$ bounded but non-vanishing for $x_1 \to \infty$ because of the boundary term at $p_1=0$ after partial integration. Note that this boundary term can not be cancelled by a contribution from the ``second temperature'' unless $\beta_1 = \beta_2$. For $x_2$ and $x_3$ we can use the same trick like for $x_1$. As we integrate $p_2$ and $p_3$ over the full real line there is no boundary term at zero and we can reiterate the trick to obtain decay in $x_2$, $x_3$ faster than any polynomial.
\end{proof}

\end{document}